\documentclass[12pt,onecolumn,letter]{IEEEtran}
\usepackage{verbatim}
\usepackage{amsmath}
\usepackage{amsthm}
\usepackage{amssymb}
\usepackage{color}
\usepackage{bm}

\usepackage[pdftex]{graphicx}
\newtheorem{proposition}{Proposition} 
\newtheorem{theorem}{Theorem}
\newtheorem{corollary}{Corollary}

\newtheorem{lemma}{Lemma}

\newtheorem{statement}{Statement}

\DeclareMathOperator{\Tr}{\mathrm{Tr}}
\def\bV{\bm{V}}

 \newenvironment{derivationof}[1]{\vspace*{5mm} \par \noindent
         \quad{\it Derivations of #1:\hspace{2mm}}}{\endproof
}

\def\FF{\mathbb{F}}

\def\Univ{\mathop{\rm Univ}}
\def\Unif{\mathop{\rm Unif}}
\def\sgn{\mathop{\rm sgn}}

\def\bx{\bm{x}}

\def\ba{\bm{a}}
\def\bb{\bm{b}}

\def\bu{\bm{u}}

\def\bt{\bm{t}}

\def\cl{\mathop{cl.}}

\def\QED{\mbox{\rule[0pt]{1.5ex}{1.5ex}}}

\def\endproof{\hspace*{\fill}~\QED\par\endtrivlist\unskip}

\def\sM{\mathsf{M}}

\def\Label#1{\label{#1}\ [\ \text{#1}\ ]\ }
\def\Label{\label}

\begin{document}

\title{Universal classical-quantum superposition coding and universal classical-quantum multiple access channel coding}

\author{Masahito Hayashi \IEEEmembership{Fellow, IEEE}
and Ning Cai \IEEEmembership{Fellow, IEEE}
\thanks{The work of MH was supported 
in part by
Guangdong Provincial Key Laboratory (Grant No. 2019B121203002).
}
\thanks{Masahito Hayashi is with 
Shenzhen Institute for Quantum Science and Engineering, Southern University of Science and Technology,
Shenzhen, 518055, China,
Guangdong Provincial Key Laboratory of Quantum Science and Engineering,
Southern University of Science and Technology, Shenzhen 518055, China,
Shenzhen Key Laboratory of Quantum Science and Engineering, Southern
University of Science and Technology, Shenzhen 518055, China,
and
the Graduate School of Mathematics, Nagoya University, Nagoya, 464-8602, Japan
(e-mail:hayashi@sustech.edu.cn).
Ning Cai is with the School of Information Science and Technology, ShanghaiTech University, Middle Huaxia Road no 393,
Pudong, Shanghai  201210, China
(e-mail: ningcai@shanghaitech.edu.cn).} }

\markboth{M. Hayashi 
and N. Cai: Universal c-q superposition coding \& universal c-q MAC coding}{}

\maketitle

\begin{abstract}
We derive universal classical-quantum superposition coding and universal classical-quantum multiple access channel code
by using generalized packing lemmas for the type method.
Using our classical-quantum universal superposition code, we establish the capacity region of a classical-quantum compound broadcast channel with degraded message sets.
Our universal classical-quantum multiple access channel codes have two types of codes.
One is a code with joint decoding and the other is a code with separate decoding.
It is not so easy to construct a former code that universally achieves general points of the capacity region 
beyond corner points.
First, we construct the latter code that universally achieves general points of the capacity region.
Then, converting the latter code to the former coder, 
we construct the above desired code with the former type.
\end{abstract}


\section{Introduction}
Reliable transmission of messages via communication channel is a fundamental problem
whichever classical or quantum channel is used.
Even when the channel can be regarded as a discrete memoryless, it is not so easy to perfectly identify the form of channel.
To address this problem,
it is natural to consider the worst decoding error probability among the set of possible channels.
Such a problem is called a compound channel, and
has been introduced, in the classical setting, independently by
Blackwell, Breiman and Thomasian \cite{BBT}, Dobrusin \cite{Dobrusin}, and Wolfowitz \cite{Wolfowitz}.
For its systematic study, Csisz\'{a}r and K\"{o}rner \cite{CK}
established universal channel coding based on the method of type.
They showed the existence of a pair of an encoder and a decoder that works with an arbitrary channel when the mutual information is larger than the transmission rate.
Such a code is called a universal code, whose construction was given by the packing lemma \cite[Lemma 10.1]{CK}, which is a key lemma in the method of type.
In their setting, the number of possible channels is continuous, which is the key point of their  construction due to the following reason.
If the number of possible channels is finite, a very simple method derives the existence of a pair of an encoder and decoder to work with all possible channels.

Using this simple idea, Datta and Dorlas \cite{D-D} showed the existence of universal code for the classical-quantum channel coding when the number of possible channels is finite.
Later, the two papers \cite{BB1,Ha1} independently showed the existence of universal code for the classical-quantum channel coding for the finite-dimensional case even when the number of possible channels is infinite.
To find a universal encoder, the paper \cite{Ha1} used the packing lemma \cite[Lemma 10.1]{CK} in a way different from the way originally used by Csisz\'{a}r and K\"{o}rner \cite{CK}.
That is, the encoder in \cite{Ha1} does not depend on the dimension of the output system while the encoder in \cite{CK} depends on the cardinality of the output alphabet.
Hence, the method by \cite{Ha1} was applied to the classical case with a general output system including the classical continuous system so that
the succeeding paper \cite{Conti} constructed a universal code even for such a general output system.
To construct a universal decoder, the paper \cite{Ha1} employed a notable combination of Schur duality and the method of types, which was used in various settings for universal quantum data compression \cite[Chapter 6]{Group2} \cite{Ha2,Ha3,Ha4,Ha5}.

The broadcast channel with degraded message sets (BCD) was introduced considered by K\"{o}rner and Marton \cite{korner77}.
This problem has one sender and two receivers $Y$ and $Z$, where
we treat the private message $M_B$ intended to be sent to Receiver $Y$,
and the common message $M_A$ intended to be sent to both receivers $Y$ and $Z$.
Here, the confidentiality of the private message $M_B$ for Receiver $Z$ is not required\footnote{
Many papers \cite{Devetak,BJS1,BJS2,SAHJF} in quantum information use the word ``private message'' as the message to be confidential to the other receiver.
However, the representative papers \cite{korner77,Korner-Sgarro,kaspi11} for classical BCD
use the word ``private message'' in the same way as this paper.}.
To show the achievability, the method of superposition code is used.
K\"{o}rner and Sgarro \cite{Korner-Sgarro} proposed
universal codes for this problem with exponentially small decoding error probability
by generalizing the packing lemma.
In the case of random coding for superposition code,
Kaspi and Merhav \cite{kaspi11} derived a lower bound of the error exponent.
Superposition codes are a key tool for the broadcast channel, and used in various tasks including broadcast channels with confidential messages \cite{csiszar78}.

For the quantum version,
Yard, Hayden, and Devetak \cite{YHD} constructed a quantum superposition code, which automatically
derives the achievability part of the classical-quantum BCD (c-q BCD) \cite[Section II-C]{YHD}.
However, the converse part was not shown.
Recently, Boche, Jan{\ss}en, and Saeedinaeeni \cite{BJS1} proposed a universal random construction for a quantum superposition code.
In their construction, the choice of the encoder is a random choice, whose ensemble does not depend on the channel, 
and their decoder works with all possible channels and depends on the choice of the encoder.
Hence, it was an open problem to show the existence of a deterministic encoder that universally works 
for c-q BCD

This paper shows the converse part of the c-q BCD, i.e., it proves 
the optimality of the quantum superposition code given by \cite{YHD}.  
Then, this paper shows 
a pair of a deterministic code based on the packing lemma and a decoder 
encoder works universally
with any pair of classical-quantum channels
when the pair satisfies a certain condition with respect to the mutual information and the transmission rates.
In our construction, our encoder is given by using 
the generalized packing lemma, i.e., the result for the method of types, 
given by K\"{o}rner and Sgarro \cite{Korner-Sgarro} 
while our use of the generalized packing lemma is similar to the use of the packing lemma
in \cite{Ha1} and is different from the use in \cite{Korner-Sgarro}.
Our decoder is based on a similar combination of Schur duality and the method of types 
in a way similar to the paper \cite{Ha1}.
In addition, when we have a family of pairs of channels, to address the worst case, 
we formulate the c-q compound BCD.
Applying our universal code, we derived the capacity region for c-q compound BCD.
Further, we apply our method to another problem, 
universal construction of classical-quantum multiple access channel (c-q MAC) code that achieves 
the corner points in the capacity region.

The MAC was introduced by Ahlswede \cite{Ahlswede} and Liao \cite{Liao}.
Universal codes for this problem were proposed with exponentially small decoding error probability
by generalizing the packing lemma
by Pokorny and Wallmeier \cite{PW} and Liu and Hughes \cite{Liu-Hughes}.
For the quantum version, Winter \cite{Winter} derived the capacity region for classical-quantum MAC (c-q MAC).
The paper \cite{Ahlswede-Cai} showed the strong converse part of this problem. 
The paper \cite{HMW} addressed compound cq-MAC, and discussed the achievable rate pair for compound cq-MAC by using Plolar code. However, their rate-region is not tight.
Recently, Boche, Jan{\ss}en, and Saeedinaeeni \cite{BJS2}
proposed a universal random construction for a quantum superposition code when one sender is classical and the other sender is quantum.
Similar to another their paper \cite{BJS1}, in their construction, the choice of the encoder is a random choice, whose ensemble does not depend on the channel, and their decoder works with all possible channels and depends on the choice of the encoder.

Applying our method for universal c-q superposition code based on packing lemma, this paper shows the existence of a deterministic pair of an encoder and a decoder that universally and directly achieves corner points of the capacity region for any c-q MAC
when the c-q MAC satisfies a certain condition with respect to the mutual information and the transmission rates.
Our deterministic encoder is given by using the result for the generalized packing lemma showed by Liu and Hughes \cite{Liu-Hughes} while our use of the result by \cite{Liu-Hughes} is similar to the use in \cite{Ha1} and is different from the use in \cite{Liu-Hughes}.
Our decoder is based on a combination of Schur duality and the method of types in a way similar to the paper \cite{Ha1}.

Since the encoder of this method is based on the generalized packing lemma, it can be easily extended to a general point of the capacity region beyond corner points.
However, the decoder cannot be directly extended to this general case because this generalization requires the handling of projections that are not commutative with each other.
To avoid this problem, we introduce the concept of separate decoding.
In this setting, the receiver has two decoders. One is a decoder to decode the message from one sender,
and the other is a decoder to decode the message from another sender.
If we allow such a separate decoder, our code universally achieves general points of the capacity region for any c-q MAC.
Fortunately, by using gentle operator lemma \cite{Win,O-G,Springer}, a separate decoder can be converted to 
a joint decoder. Using this conversion, we construct 
a code with joint decoding that universally achieves general points of the capacity region for any c-q MAC.
In fact, while the previous withdrawn paper \cite{unsol} tried to construct a code to achieve general points in capacity region without time sharing,
it has a serious gap so that a code construction without time sharing had been an open problem.
Since our construction does not employ time sharing,
it can be considered as a solution for this open problem. 


Finally, to address the worst case with a given family of c-q MACs,
we discuss c-q compound MAC.
In the classical case, a single-letterized form of the capacity region of
a compound MAC is known \cite{A74,PW,Liu-Hughes}.
The recent paper \cite{BJS2} derived the capacity region of a compound MAC with a limiting expression, whose classical case is different from the above single-letterized form.
Using the above universal code with joint decoding,
this paper derives a single-letterized form of the capacity region of a c-q compound MAC as a quantum extension of the above single-letterized form.

The remaining part of this paper is organized as follows.
Section \ref{S2} states our main results for c-q BCD including the converse part of the capacity region of a c-q BCD, the existence of a universal classical-quantum superposition code, and the capacity region of c-q compound BCD.
Section \ref{S3} states our main results for c-q MAC coding including universal c-q MAC codes with joint and separate decoding and the capacity region of c-q compound MAC.
Section \ref{S3-5} discusses the capacity region of c-q compound MAC and various quantities in several examples for a family of c-q MAC.
Section \ref{S4} proves the various converse results part for c-q BCD.
Section \ref{SS6} proves the converse part of c-q compound MAC.
Section \ref{S5} gives several new results for type methods, which are
the preparation for our universal codes.
Section \ref{S6} gives our universal c-q superposition code, and shows its exponent (exponential decreasing rate of the decoding error probability).
Section \ref{S7} gives our universal c-q MAC code with joint decoding, and shows its exponent.
Section \ref{S7V} gives our universal c-q MAC code with separate decoding, and shows its exponent.
Section \ref{S8} gives the discussions and conclusions.
Appendix \ref{A2} shows the exponent for another decoder for universal c-q superposition code.
The decoder of Appendix \ref{A2}
is similar to that given in Section \ref{S7V} and it has an exponent different from that Section \ref{S6}.

\section{Classical-Quantum Broadcast Channels with Degraded Message Sets}\Label{S2}
\subsection{Fixed channel case}
This section states our results for classical-quantum broadcast channels with degraded message sets (c-q BCD) including 
universal classical-quantum superposition code.
c-q BCD is formulated with 
two classical-quantum channels with a common classical input system ${\cal X}$, which is a finite set.
One channel is a channel from the classical system ${\cal X}$
to a receiver $Y$ having a quantum system ${\cal H}_Y$, which is written as $x \mapsto W_x$.
The other channel is a channel from the classical system ${\cal X}$
to another receiver $Z$ having another quantum system ${\cal H}_Z$, which is written as $x \mapsto W_{Z,x}$.

The aim of classical-quantum broadcast channels with degraded message sets
is the transmission of two kinds of messages.
One is the common message, which needs to be correctly sent to both receivers.
The other is the private message, which needs to be correctly sent only to Receiver $Y$,
where its confidentiality to Receiver $Z$ is not required.

We define the $n$-fold c-q memoryless channel of the channel $\{W_x\}_{x \in {\cal X}}$.
\begin{align}
W_{\bx}^{(n)}:=W_{x_1}\otimes \cdots \otimes W_{x_n}
\end{align}
for $\bx=(x_1,\ldots,x_n)\in {\cal X}^n$.
Similarly, 
we define the $n$-fold c-q memoryless channel $W_{Z,\bx}^{(n)}$
of the channel $\{W_{Z,x}\}_{x \in {\cal X}}$.

An encoder is a map 
$\psi_{n}$ from $\hat{M}_{A,n}\times \hat{M}_{B,n}$
to ${\cal X}^n$, where
$\hat{M}_{A,n}:=\{1, \ldots, \sM_{A,n} \}$ and $\hat{M}_{B,n}:=\{1, \ldots, \sM_{B,n} \}$.
A decoder is given by a pair of POVMs 
$D^n:=\{D_{j,k}^n\}_{(j,k)\in \hat{M}_{A,n}\times \hat{M}_{B,n}}$
on ${\cal H}_Y^{\otimes n}$
and 
$D^{Z,n}:=\{D_{j}^{Z,n}\}_{j\in \hat{M}_{A,n}}$
on ${\cal H}_Z^{\otimes n}$

Then, the triplet $(\psi_{n},D^n,D^{Z,n})$
is called a code for classical-quantum broadcast channels with degraded message sets, and is denoted by $\Psi_n$.
In the following, it is simplified to a code. 
The message sizes $\sM_{A,n}$ and $\sM_{B,n}$ are written as $|\Psi_n|_A $ and $|\Psi_n|_B$, respectively.
The average decoding error probabilities for Receivers $Y$ and $Z$ 
are given as
\begin{align}
\epsilon_Y(\Psi_n;W^{(n)}):=&\sum_{(j,k)\in\hat{M}_{A,n}\times \hat{M}_{B,n}}
\frac{1}{\sM_{A,n} \sM_{B,n}}
\Tr W_{\psi_{n}(j,k)}^{(n)}
(I- D_{j,k}^n) \\
\epsilon_Z(\Psi_n;W^{(n)}):=&\sum_{(j,k)\in\hat{M}_{A,n}\times \hat{M}_{B,n}}
\frac{1}{\sM_{A,n} \sM_{B,n}}
\Tr W_{\psi_{n}(j,k)}^{(n)}
(I- D_{j}^{Z,n}).
\end{align}
We describe the transmission rates of the common and private messages as $R_A$ and $R_B$.
The rate pair $(R_A,R_B)$ is achievable when there exists a sequence of codes
$\{\Psi_n\}$ such that 
$R_A= \lim_{n\to \infty} \frac{1}{n}\log |\Psi_n|_A$,
$R_B= \lim_{n\to \infty} \frac{1}{n}\log |\Psi_n|_B$,
$\epsilon_Y(\Psi_n;W^{(n)})\to 0$, and 
$\epsilon_Z(\Psi_n;W^{(n)})\to 0$.
The closure of the set of achievable rate pairs $(R_A,R_B)$ is called the capacity region, and is 
denoted by ${\cal C}$, i.e.,\par\noindent
${\cal C}:= \cl \{ (R_A,R_B)| (R_A,R_B) \hbox{ is achievable.}\}$.
We can calculate the capacity region as follows.
\begin{theorem}\Label{T1}
The following equations hold;
\begin{align}
{\cal C}
=&\cl \bigcup_{P_{UX}}\Big\{(R_A,R_B)\Big| R_A \le \min(I(U;Y)_{P_{UX}},I(U;Z)_{P_{UX}}), R_B \le I(X;Y|U)_{P_{UX}} \Big)\Big\}_{P_{UX}} \nonumber \\
=&\cl \bigcup_{P_{UX}}\Big\{(R_A,R_B)\Big| R_A \le \min(I(U;Y)_{P_{UX}},I(U;Z)_{P_{UX}}), R_A+R_B I(UX;Y)_{P_{UX}} \Big)\Big\}_{P_{UX}},
\end{align}
where $\cl$ expresses the closure of the convex hull.
\end{theorem}

The achievability of the above region was shown by using the quantum superposition code in \cite{YHD}.
The recent paper \cite{AHW} essentially derived an exponential decreasing rate of the decoding error probability of
the randomly generated quantum superposition code.
We show the converse part in this paper.

To state our universal code, we prepare information measure as follows.
For $\alpha>0$, 
and the state $\rho_{XY}:=\sum_{x}P_{X}(x)|x\rangle\langle x|\otimes W_x$,
Petz's version of R\'{e}nyi mutual information is given as 
\begin{align}
I_{\alpha}(X;Y):=\min_{\sigma_Y}D_\alpha(\rho_{XY}\| \rho_X \otimes \sigma_Y ),
\end{align}
where Petz's version of R\'{e}nyi divergence $D_\alpha(\rho\|\sigma)$ is defined as \cite{Petz}
\begin{align}
e^{(\alpha-1)D_\alpha(\rho\|\sigma)}:=\Tr \rho^\alpha \sigma^{1-\alpha}. 
\end{align}
It is known in \cite{Sibson}\cite[Lemma 2]{q-wire} that
this measure has the Gallager form \cite{Gallager};
\begin{align}
e^{(\alpha-1)I_{\alpha}(X;Y)}=
\sgn(\alpha-1) \min_{\sigma} \sgn(\alpha-1) \sum_{x\in {\cal X}} P_{X}(x) \Tr  W_x^{\alpha} \sigma^{1-\alpha}
=\Big(
\Tr
\Big( \sum_{x\in {\cal X}} P_{X}(x) 
 W_x^{\alpha} \Big)^{\frac{1}{\alpha}}\Big)^{\alpha}.\Label{Sibson}
\end{align}

This information measure can be extended to the case with tripartite case.
For $\alpha>0$, 
and the state $\rho_{UXY}:=\sum_{u,x}P_{UX}(u,x)|u,x\rangle\langle u,x|\otimes W_x$,
we define Petz's version of R\'{e}nyi conditional mutual information; 
\begin{align}
I_{\alpha}(X;Y|U):=\min_{\sigma_{X-U-Y}}D_\alpha(\rho_{UXY} \|\sigma_{X-U-Y} ),
\end{align}
where $\sigma_{X-U-Y}$ is restricted to the form $
\sum_{uv} Q_U(u)P_{X|U}(x|u)|u,x\rangle\langle u,x|\otimes \sigma_u$
and $Q_U$ is an arbitrary distribution and $\sigma_u$ is an arbitrary state on ${\cal H}_Y$.
If we need to express the distribution $P_{XU}$, we denote it by 
$I_{\alpha}(X;Y|U)_{P_{XU}}$.

This measure can be written as follows.
\begin{lemma}
The following equation holds;
\begin{align}
e^{(\alpha-1)I_{\alpha}(X;Y|U)}=
\Big(
\sum_{u}P_U(u)
\Tr
\Big( \sum_{x\in {\cal X}} P_{X|U}(x|u) 
 W_x^{\alpha} \Big)^{\frac{1}{\alpha}}\Big)^{\alpha}.
\end{align}
\end{lemma}
The classical version of the right hand side 
was used in \cite[Sec. IV, Th. 1]{kaspi11}\cite[Eq.(12)]{SMC} .

\begin{proof}
We show only the case with $\alpha<1$ because 
the other case can be shown by changing the maximum by the minimum.
\begin{align}
&e^{(\alpha-1)I_{\alpha}(X;Y)}=
\max_{Q_U}\max_{\sigma_u}
\sum_u P_U(u)^{\alpha} Q_U(u)^{1-\alpha}\sum_x P_{X|U}(x|u) (\Tr W_x^\alpha \sigma_u^{1-\alpha})\nonumber \\
=&\max_{Q_U}
\sum_u P_U(u)^{\alpha} Q_U(u)^{1-\alpha}
\Big(\max_{\sigma_u} \sum_x P_{X|U}(x|u) (\Tr W_x^\alpha \sigma_u^{1-\alpha}) \Big)\nonumber \\
\stackrel{(a)}{=}  &
\max_{Q_U}
\sum_u P_U(u)^{\alpha} Q_U(u)^{1-\alpha}
\Big(
\Tr \big(\sum_x P_{X|U}(x|u) W_x^\alpha \big)^{\frac{1}{\alpha}}
\Big)^{\alpha} \nonumber \\
=&\max_{Q_U}
\Big( \sum_u P_U(u)
\Tr \big(\sum_x P_{X|U}(x|u) W_x^\alpha \big)^{\frac{1}{\alpha}}
\Big)^{\alpha} 
Q_U(u)^{1-\alpha} \nonumber \\
\stackrel{(b)}{=}  &
\Big(
\sum_{u}P_U(u)
\Tr
\Big( \sum_{x\in {\cal X}} P_{X|U}(x|u) 
 W_x^{\alpha} \Big)^{\frac{1}{\alpha}}\Big)^{\alpha},
\end{align}
where 
Step $(a)$ follows from the application of 
\eqref{Sibson} to the state $\sum_{x}P_{X|U}(x|u)|x\rangle\langle x|\otimes W_x$,
and 
Step $(b)$ follows from the application of H\"{o}lder inequality to the two real vectors
$(\sum_u P_U(u)
\max_{\sigma_u} \sum_x P_{X|U}(x|u) 
(\Tr W_x^\alpha \sigma_u^{1-\alpha}))_{u}$
and $(Q_U(u))_u$.
\end{proof}

\subsection{Universal code}
Next, we consider a universal code construction.
That is, the capacity region can be universally achieved as follows.
\begin{theorem}\Label{T2}
For any $P_{UX}$, there exists a sequence of codes $\Psi_n$ with the rate pair $(R_A,R_B)$
with the positive parameters $r_A$ and $r_B$
to satisfy the following conditions.
For any channel $(\{W_x\}_{x \in {\cal X}},\{W_{Z,x}\}_{x \in {\cal X}})$,
the decoding error probability $\epsilon_Y(\Psi_n;W^{(n)})$ 
of Receiver $Y$ exponentially goes to zero, i.e.,
\begin{align}
\lim_{n\to\infty}\frac{-1}{n}\log \epsilon_Y(\Psi_n;W^{(n)})
\ge &
\min \Big(
\min \big(\max_s s(I_{1-s}(U;Y)-R_A-r_A),r_A\big),\nonumber \\
&\quad \min \big(\max_s s(I_{1-s}(X;Y|U)-R_B-r_B),r_B\big)
\Big) . \Label{Ex1}
\end{align}
The decoding error probability $\epsilon_Z(\Psi_n;W^{(n)})$ 
of Receiver $Z$ exponentially goes to zero, i.e.,
\begin{align}
\lim_{n\to\infty}\frac{-1}{n}\log \epsilon_Z(\Psi_n;W^{(n)})
\ge
\min (\max_s s(I_{1-s}(U;Z)-R_A-r_A),r_A).
\Label{Ex2}
\end{align}
\end{theorem}

The decoder POVM to achieve the performance presented in the above theorem is constructed by considering the irreducible decomposition. 
The use of irreducible decomposition can be considered as a quantum version of Type method, which is summarized in \cite{Group2}.
The optimization of the respective exponents can be done as follows.
The optimization of the exponents of Receiver $Y$ is done as
\begin{align}
&
\max_{r_A,r_B}
\min \Big(
\min \big(\max_s s(I_{1-s}(U;Y)-R_A-r_A),r_A\big),
\min \big(\max_s s(I_{1-s}(X;Y|U)-R_B-r_B),r_B\big)
\Big) \nonumber \\
=&
\min \Big(
\max_{r_A}
\min \big(\max_s s(I_{1-s}(U;Y)-R_A-r_A),r_A\big),
\max_{r_B}
\min \big(\max_s s(I_{1-s}(X;Y|U)-R_B-r_B),r_B\big)
\Big) \nonumber \\
= &
\min \Big(
\max_{0 \le s\le 1}\frac{s(I_{1-s}(U;Y) -R_A)}{1+s},
\max_{0 \le s\le 1}\frac{s(I_{1-s}(X;Y|U) -R_B)}{1+s}
\Big) ,
\end{align}
where the maximum is achieved when
\begin{align}
r_A=&\max_{0 \le s\le 1}\frac{s(I_{1-s}(U;Y) -R_A)}{1+s} \\
r_B=&\max_{0 \le s\le 1}\frac{s(I_{1-s}(X;Y|U) -R_B)}{1+s} .
\end{align}
The optimization of the exponent of Receiver $Z$ is done as
\begin{align}
&
\max_{r_A}
\min \big(\max_s s(I_{1-s}(U;Z)-R_A-r_A),r_A\big)
= 
\max_{0 \le s\le 1}\frac{s(I_{1-s}(U;Z) -R_A)}{1+s},
\end{align}
where the maximum is achieved when
\begin{align}
r_A=&\max_{0 \le s\le 1}\frac{s(I_{1-s}(U;Z) -R_A)}{1+s} .
\end{align}

However, the optimum choice of $r_A$ for the exponent of Receiver $Y$ is different from that of 
Receiver $Z$.
Further, the optimum choice depends on the choice of channel.
However, when $R_A+r_A < \min(I(U;Y)_{P_{UX}},I(U;Z)_{P_{UX}})$
and $R_B +r_A < I(X;Y|U)_{P_{UX}}$, both exponents are strictly positive.
Therefore, we have the following corollary.

\begin{corollary}\Label{Cor1}
For any $P_{UX}$, there exists a sequence of codes with the rate pair $(R_A,R_B)$
with arbitrary small positive parameters $r_A$ and $r_B$
to satisfy the following conditions.
When a channel $(\{W_x\}_{x \in {\cal X}},\{W_{Z,x}\}_{x \in {\cal X}})$
satisfies $R_A+r_A < \min(I(U;Y)_{P_{UX}},I(U;Z)_{P_{UX}})$
and $R_B +r_A < I(X;Y|U)_{P_{UX}}$,
the decoding error probabilities of both receivers exponentially go to zero.
\end{corollary}

The code given in the above theorem can be considered as a universal superposition code.

\subsection{Compound channel}
Next, to rigorously handle the optimization of the worst case, 
we consider a compound channel model ${\cal W}:=\{
(\{W_{x,\theta}\}_{x \in {\cal X}},\{W_{Z,x,\theta}\}_{x \in {\cal X}})\}_{\theta \in \Theta}$  with a parametric space 
$\Theta$.
In this model, we do not know what $\theta \in \Theta$ is the true parameter.
Hence, we need to consider the worst case.
That is, a rate pair $(R_A,R_B)$ is achievable under the channel model ${\cal W}$ when there exists a sequence of codes with the transmission rate pair $(R_A,R_B)$ such that the decoding error probabilities of both receivers are goes to zero
when the true channel is any element of the channel model ${\cal W}$.
We denote the capacity region of the compound channel model ${\cal W}$
by ${\cal C}_{\cal W}$, i.e.,
\begin{align}
{\cal C}_{\cal W}:= \cl {\{ (R_A,R_B)| (R_A,R_B) \hbox{ is achievable under the channel model }
{\cal W}.\}}.
\end{align}

\begin{theorem}\Label{T2-5}
The following equation holds;
\begin{align}
{\cal C}_{\cal W}
=&\cl \bigcup_{P_{UX}} \Big\{(R_A,R_B)\Big| R_A \le  \min_\theta \min(I(U;Y)_{P_{UX},\theta},I(U;Z)_{P_{UX},\theta}), 
R_B \le \min_\theta I(X;Y|U)_{P_{UX},\theta} \Big)\Big\}_{P_{UX}} .
\end{align}
\end{theorem}

\section{Classical-Quantum multiple access channel}\Label{S3}
\subsection{Universal code with joint decoding for corner points}
This section states our results for classical-quantum multiple access channel (cq-MAC) with two input systems ${\cal A}$ and ${\cal B}$, in which
the output state on ${\cal H}_Y$ is given as
$W_{a,b}$ dependently of $a \in {\cal A}$ and $a \in {\cal B}$.
The aim of classical-quantum multiple access channel is transmission of two kinds of messages to 
the quantum receiver $Y$.
One message is sent from Sender $A$ and the other message is sent from Sender $B$.


An encoder is a pair of maps 
$\psi_{A,n}$ from $\hat{M}_{A,n}:=\{1, \ldots, \sM_{A,n} \}$ to ${\cal A}^n$ and 
$\psi_{B,n}$ from $\hat{M}_{B,n}:=\{1, \ldots, \sM_{B,n} \}$ to ${\cal B}^n$.
A decoder with joint decoding
is given by a POVM $D^n:=\{D_{j,k}^n\}_{(j,k)\in \hat{M}_{A,n}\times \hat{M}_{B,n}}$
on ${\cal H}_Y^{\otimes n}$.
Then, the triplet $(\psi_{A,n},\psi_{B,n},D^n)$
is called a code with joint decoding, and is denoted by $\Psi_n$.
In the following, it is simplified to a code. 
The message sizes $\sM_{A,n}$ and $\sM_{B,n}$ are written as $|\Psi_n|_A $ and $|\Psi_n|_B$, respectively.
The average decoding error probability is given as
\begin{align}
\epsilon(\Psi_n;W^{(n)}):=\sum_{(j,k)\in\hat{M}_{A,n}\times \hat{M}_{B,n}}
\frac{1}{\sM_{A,n} \sM_{B,n}}
\Tr W_{\psi_{A,n}(j),\psi_{B,n}(k)}^{(n)}
(I- D_{j,k}^n).
\end{align}

The transmission rates from $A$ and $B$ are written as $R_A$ and $R_B$.
The rate pair $(R_A,R_B)$ is called achievable when there exists a sequence of codes
$\{\Psi_n\}$ such that 
$R_A= \lim_{n\to \infty} \frac{1}{n}\log |\Psi_n|_A$,
$R_B= \lim_{n\to \infty} \frac{1}{n}\log |\Psi_n|_B$,
and $\epsilon(\Psi_n;W^{(n)})\to 0$.
The closure of the set of achievable rate pairs $(R_A,R_B)$ is called 
the capacity region, and is denoted by ${\cal C}$.
Winter \cite{Winter} showed that  
\begin{align}
{\cal C}=\cl \bigcup_{P_A \times P_B}
\Big\{ (R_A,R_B) \Big|
R_A \le I(A;Y|B)_{P_A \times P_B},
R_B \le I(B;Y|A)_{P_A \times P_B},
R_A+ R_B \le I(AB;Y)_{P_A \times P_B}
\Big\}.
\Label{Winter}
\end{align}

Next, we consider universal codes for cq-MAC.
To reuse the derivation of our universal codes for 
classical-quantum broadcast channels with degraded message sets,
we focus on universal codes to achieve only the external points
$(I(A;Y),I(B;Y|A))$ and $(I(A;Y|B),I(B;Y))$.
These values depend on the choice of the distributions $P_A,P_B$ and the cq-MAC.
Consider the case when we have two choices of the cq-MAC, $W_{a,b|0}$ and $W_{a,b|1}$.
Then, the mutual information and the conditional mutual information 
is denoted by $I(A;Y)_i,I(B;Y|A)_i$ for $i=0,1$ to express the dependence of the channel.
When we optimize the rate $I(A;Y)$ under a constraint for another rate $I(B;Y|A)$, we need to consider the following problem for a given $R_B>0$;
\begin{align}
\max_{P_A,P_B} \{ \min(I(A;Y)_0,I(A;Y)_1)| I(B;Y|A)_0,I(B;Y|A)_1 \ge R_B\}.\Label{Hi1}
\end{align}
Here, we consider only the product distribution $P_A\times P_B$.
However, it is possible to choose probabilistic mixture for this choice.
That is, alternatively, we consider the maximization;
\begin{align}
\max_{P_{A-T-B}} \{ \min(I(A;Y|T)_0,I(A;Y|T)_1)| I(B;Y|AT)_0,I(B;Y|AT)_1 \ge R_B\},
\Label{Hi2}
\end{align}
where 
the joint distribution on $A,T,B$ satisfies the Markov chain condition $A-T-B$.
Clearly, \eqref{Hi1} $\le$ \eqref{Hi2}.
We have examples for the gap between \eqref{Hi1} and \eqref{Hi2}.
As shown in Section \ref{S7},
we construct universal codes to achieve \eqref{Hi2}.

\begin{theorem}\Label{T3}
For any distribution $P_{A-T-B} $ with Markov condition $A-T-B$, 
there exists a sequence of codes $\{\Psi_n\}$ with the rate pair $(R_A,R_B)$ and arbitrary small
 positive parameters $r_A$ and $r_B$
to satisfy the following conditions.
For any channel $W=\{W_{a,b}\}_{a\in {\cal A},b\in {\cal B}}$,
the exponent of the average decoding error probability $\epsilon(\Psi_n;W^{(n)})$ 
is not smaller than
\begin{align}
\min \Big(& \min \big(\max_s s(I_{1-s}(A;Y|T)_{P_{A-T-B}}-R_A-r_A),r_A\big),\nonumber \\
&\min \big(\max_s s(I_{1-s}(B;Y|AT)_{P_{A-T-B}}-R_B-r_B),r_B\big) \Big). 
\Label{Ex3}
\end{align}
That is, when $R_A+r_A< I(A;Y|T)_{P_{A-T-B}}$ and $R_B +r_B< I(B;Y|AT)_{P_{A-T-B}}$,
the average decoding error probability $\epsilon(\Psi_n;W^{(n)})$ exponentially goes to zero.
\end{theorem}

Theorem \ref{T3} can be shown in a similar way as Theorem \ref{T2} when 
$T$ takes a single value.
Given a product distribution $P_A\times P_B$, 
when $U=A$, $X=(A,B)$, and $P_{UX}(a,a',b)=P_A(a)\delta_{a,a'}P_{B}(b) $,
we have $I(U;Y)=I(A;Y)$ and $I(X;Y|U)=I(AB;Y|A)=I(B;Y|A)$,
that is, the rate pair of the superposition code equals the rate pair of the multiple access code.
The case with a general $T$ needs more complicated treatment.
This correspondence plays an essential role in our proof of Theorem \ref{T3}. 

A general point of the capacity region can be achieved by applying time sharing to two corner points achieved by 
Theorem \ref{T3}.
Since the decoding error probability in Theorem \ref{T3} goes to zero exponentially,
the constructed code by the time sharing also has an exponentially small decoding error probability.
Therefore, we have the following corollary.

\begin{corollary}\Label{Cor2}
For any distributions $P_{A-T-B}^0 $ and $P_{A-T-B}^1 $ with Markov condition $A-T-B$, 
there exists a sequence of codes $\{\Psi_n\}$ with the rate pair $(\lambda R_{A,0}+(1-\lambda)R_{A,1},
\lambda R_{B,0}+(1-\lambda) R_{B,1})$ 
with $\lambda \in [0,1]$ and an arbitrary small positive parameter $\epsilon >0$
to satisfy the following conditions.
When a channel $\{W_{a,b}\}_{a\in {\cal A},b\in {\cal B}}$
satisfies 
the conditions 
$R_{A,0} +\epsilon\le I(A;Y|T)_{P_{A-T-B}^0 }$,
$R_{A,1} +\epsilon\le I(A;Y|BT)_{P_{A-T-B}^1 }$,
$R_{B,0} +\epsilon\le I(B;Y|AT)_{P_{A-T-B}^0 }$, and
$R_{B,1} +\epsilon\le I(B;Y|T)_{P_{A-T-B}^1 }$
the average decoding error probability $\epsilon(\Psi_n;W^{(n)})$ exponentially goes to zero.
\end{corollary}

\subsection{Universal code with separate decoding}
When the rate region is not a corner point,
we construct only a universal code with separate decoding, which is defined as follows.
Given 
an encoder, a pair of maps 
$\psi_{A,n}$ from $\hat{M}_{A,n}:=\{1, \ldots, \sM_{A,n} \}$ to ${\cal A}^n$ and 
$\psi_{B,n}$ from $\hat{M}_{B,n}:=\{1, \ldots, \sM_{B,n} \}$ to ${\cal B}^n$,
a decoder with separate decoding
is given a pair of POVMs 
$D^{A,n}:=\{D_{j}^{A,n}\}_{j\in \hat{M}_{A,n}}$
on ${\cal H}_Y^{\otimes n}$
and $D^{B,n}:=\{D_{k}^{B,n}\}_{k\in \hat{M}_{B,n}}$
on ${\cal H}_Y^{\otimes n}$.
Then, the quadruple $(\psi_{A,n},\psi_{B,n},D^{A,n},D^{B,n})$
is called a code with separate decoding, and is denoted by $\Psi_{S,n}$.
The message sizes $\sM_{A,n}$ and $\sM_{B,n}$ are written as $|\Psi_{S,n}|_A $ and $|\Psi_{S,n}|_B$, respectively.
The respective average decoding error probabilities are given as
\begin{align}
\epsilon_A(\Psi_{S,n};W^{(n)})&:=\sum_{(j,k)\in\hat{M}_{A,n}\times \hat{M}_{B,n}}
\frac{1}{\sM_{A,n} \sM_{B,n}}
\Tr W_{\psi_{A,n}(j),\psi_{B,n}(k)}^{(n)}
(I- D_{j}^{A,n}) \\
\epsilon_B(\Psi_{S,n};W^{(n)})&:=\sum_{(j,k)\in\hat{M}_{A,n}\times \hat{M}_{B,n}}
\frac{1}{\sM_{A,n} \sM_{B,n}}
\Tr W_{\psi_{A,n}(j),\psi_{B,n}(k)}^{(n)}
(I- D_{k}^{B,n}) .
\end{align}
Then, we consider their maximum as
\begin{align}
\epsilon(\Psi_{S,n};W^{(n)}):=
\max(\epsilon_A(\Psi_{S,n};W^{(n)}),\epsilon_B(\Psi_{S,n};W^{(n)})).
\end{align}
A rate pair $(R_A,R_B)$ is called achievable with separation decoding
when there exists a sequence of codes of separation decoding
$\{\Psi_{S,n}\}$ such that 
$R_A= \lim_{n\to \infty} \frac{1}{n}\log |\Psi_n|_A$,
$R_B= \lim_{n\to \infty} \frac{1}{n}\log |\Psi_n|_B$,
and $\epsilon(\Psi_{S,n};W^{(n)})\to 0$.
The closure of the set of achievable rate pairs with separation decoding
$(R_A,R_B)$ is called 
the capacity region with separation decoding, and is denoted by ${\cal C}_S$.
The definition implies the relation ${\cal C}\subset {\cal C}_S$.
As shown in Subsection \ref{S3-C}, we have the opposite relation, i.e., we have
 \begin{align}
{\cal C}_S={\cal C}.\Label{SMF}
\end{align}
Also, as shown in Section \ref{S7V},
we have a separate decoding version of universal codes as follows.

\begin{theorem}\Label{T3-J}
For any distribution $P_{A-T-B} $ with Markov condition $A-T-B$, 
there exists a sequence of codes $\{\Psi_{S,n}\}$ of separate decoding with the rate pair $(R_A,R_B)$ with 
arbitrary small
positive parameters $r_A$ and $r_B$
to satisfy the following conditions.
For any channel $W=\{W_{a,b}\}_{a\in {\cal A},b\in {\cal B}}$, we have
\begin{align}
\lim_{n\to \infty}\frac{-1}{n}\log \epsilon_A(\Psi_{S,n};W^{(n)}) \ge & 
\min\Big( \max_s s(I_{1-s}(A;Y|BT)_{P_{A-T-B}}-R_A-r_A),r_A,  \nonumber \\
&\quad \max_s s(I_{1-s}(AB;Y|T)_{P_{A-T-B}}-R_A-R_B-r_A-r_B),r_A+r_B \Big)
\Label{Ex3A}\\
\lim_{n\to \infty}\frac{-1}{n}\log \epsilon_B(\Psi_{S,n};W^{(n)}) \ge & 
\min \Big( \max_s s(I_{1-s}(B;Y|AT)_{P_{A-T-B}}-R_B-r_B),r_B, \nonumber \\
&\quad \max_s s(I_{1-s}(AB;Y|T)_{P_{A-T-B}}-R_A-R_B-r_A-r_B),r_A+r_B \Big)
\Label{Ex3B}.
\end{align}
That is, when $R_A+r_A< I(A;Y|BT)_{P_{A-T-B}}$, 
$R_B +r_B< I(B;Y|AT)_{P_{A-T-B}}$, and 
$R_A +r_A+R_B +r_B< I(AB;Y|T)_{P_{A-T-B}}$,
the error probability $\epsilon(\Psi_{S,n};W^{(n)})$ exponentially  goes to zero.
\end{theorem}

\subsection{Universal code with joint decoding for general points}\Label{S3-C}
We construct universal code with joint decoding for general points
from universal code with separate decoding for general points.
We choose 
a code with separate decoding $\Psi_{S,n}:=
(\psi_{A,n},\psi_{B,n},D^{A,n},D^{B,n})$, where
$\psi_{A,n}$ is a map from $\hat{M}_{A,n}:=\{1, \ldots, \sM_{A,n} \}$ to ${\cal A}^n$,
and
$\psi_{B,n}$ is a map from $\hat{M}_{B,n}:=\{1, \ldots, \sM_{B,n} \}$ to ${\cal B}^n$,
the decoder is given a pair of POVMs 
$D^{A,n}:=\{D_{j}^{A,n}\}_{j\in \hat{M}_{A,n}}$
on ${\cal H}_Y^{\otimes n}$
and $D^{B,n}:=\{D_{k}^{B,n}\}_{k\in \hat{M}_{B,n}}$
on ${\cal H}_Y^{\otimes n}$.
Now, we construct the decoder with joint decoding 
as a POVM $D^n=\{D_{j,k}^n\}_{(j,k)\in \hat{M}_{A,n}\times \hat{M}_{B,n}}$
on ${\cal H}_Y^{\otimes n}$ by
\begin{align}
D_{j,k}^n:= (D_{k}^{B,n})^{1/2} D_{j}^{A,n} (D_{k}^{B,n})^{1/2}.
\end{align}
Since
\begin{align}
& \sum_{j,k }D_{j,k}^n
= \sum_{j,k }(D_{k}^{B,n})^{1/2} D_{j}^{A,n} (D_{k}^{B,n})^{1/2} \nonumber \\
= &\sum_{k }(D_{k}^{B,n})^{1/2} (\sum_j D_{j}^{A,n}) (D_{k}^{B,n})^{1/2}
= \sum_{k }(D_{k}^{B,n})^{1/2} I (D_{k}^{B,n})^{1/2}
=I,
\end{align}
$\{D_{j,k}^n\}_{(j,k)\in \hat{M}_{A,n}\times \hat{M}_{B,n}}$
satisfies the condition for POVM.
This code with joint decoding is denoted by $\Psi$.
The average decoding error probability $\epsilon(\Psi_n;W^{(n)})$ is evaluated as
\begin{align}
\epsilon(\Psi_n;W^{(n)})
=& \sum_{(j,k)\in\hat{M}_{A,n}\times \hat{M}_{B,n}}
\frac{1}{\sM_{A,n} \sM_{B,n}}
\Tr W_{\psi_{A,n}(j),\psi_{B,n}(k)}^{(n)}
(I- (D_{k}^{B,n})^{1/2} D_{j}^{A,n} (D_{k}^{B,n})^{1/2}) \nonumber \\
= &\sum_{(j,k)\in\hat{M}_{A,n}\times \hat{M}_{B,n}}
\frac{1}{\sM_{A,n} \sM_{B,n}}
\Tr W_{\psi_{A,n}(j),\psi_{B,n}(k)}^{(n)}
(I- D_{k}^{B,n}) \nonumber \\
&+\sum_{(j,k)\in\hat{M}_{A,n}\times \hat{M}_{B,n}}
\frac{1}{\sM_{A,n} \sM_{B,n}}
\Tr W_{\psi_{A,n}(j),\psi_{B,n}(k)}^{(n)}
(D_{k}^{B,n}- (D_{k}^{B,n})^{1/2} D_{j}^{A,n} (D_{k}^{B,n})^{1/2}) \nonumber \\
= &
\epsilon_B(\Psi_{S,n};W^{(n)})
+\sum_{(j,k)\in\hat{M}_{A,n}\times \hat{M}_{B,n}}
\frac{1}{\sM_{A,n} \sM_{B,n}}
\Tr (D_{k}^{B,n})^{1/2} W_{\psi_{A,n}(j),\psi_{B,n}(k)}^{(n)} (D_{k}^{B,n})^{1/2}
(I-  D_{j}^{A,n} )\nonumber  \\
\le &
\epsilon_B(\Psi_{S,n};W^{(n)})
+\sum_{(j,k)\in\hat{M}_{A,n}\times \hat{M}_{B,n}}
\frac{1}{\sM_{A,n} \sM_{B,n}}
\Tr W_{\psi_{A,n}(j),\psi_{B,n}(k)}^{(n)} 
(I-  D_{j}^{A,n} )\nonumber  \\
&+\sum_{(j,k)\in\hat{M}_{A,n}\times \hat{M}_{B,n}}
\frac{1}{\sM_{A,n} \sM_{B,n}}
\big\|
(D_{k}^{B,n})^{1/2} W_{\psi_{A,n}(j),\psi_{B,n}(k)}^{(n)} (D_{k}^{B,n})^{1/2}
-W_{\psi_{A,n}(j),\psi_{B,n}(k)}^{(n)}\big\|_1 \nonumber \\
\stackrel{(a)}{\le} &
\epsilon_B(\Psi_{S,n};W^{(n)})
+\epsilon_A(\Psi_{S,n};W^{(n)})
\nonumber \\
&+2\sum_{(j,k)\in\hat{M}_{A,n}\times \hat{M}_{B,n}}
\frac{1}{\sM_{A,n} \sM_{B,n}}
\Big( W_{\psi_{A,n}(j),\psi_{B,n}(k)}^{(n)} (I-D_{k}^{B,n})\Big)^{1/2}\nonumber \\
\stackrel{(b)}{\le} &
\epsilon_A(\Psi_{S,n};W^{(n)})+\epsilon_B(\Psi_{S,n};W^{(n)})
\nonumber \\
&+2\Big(\sum_{(j,k)\in\hat{M}_{A,n}\times \hat{M}_{B,n}}
\frac{1}{\sM_{A,n} \sM_{B,n}}
 W_{\psi_{A,n}(j),\psi_{B,n}(k)}^{(n)} (I-D_{k}^{B,n})\Big)^{1/2}\nonumber \\
= &
\epsilon_A(\Psi_{S,n};W^{(n)})
+\epsilon_B(\Psi_{S,n};W^{(n)})
+2(\epsilon_B(\Psi_{S,n};W^{(n)}))^{1/2} \Label{MGY}.
\end{align}
where
$(b)$ follows from the concavity of $x \mapsto \sqrt{x}$, and
$(a)$ follows from 
gentle operator lemma \cite[Lemma 9]{Win}, where the coefficient $2$ is given in
\cite[Appendix C]{O-G} \cite[Eqs. (9.65) and (9.66)]{Springer}.

Therefore, if 
$\epsilon(\Psi_{S,n};W^{(n)})$ goes to zero, 
$\epsilon(\Psi_{n};W^{(n)})$ also goes to zero.
Hence, we ave the relation ${\cal C}\supset {\cal C}_S$, which implies \eqref{SMF}.
Also, as a corollary of Theorem \ref{T3-J}, we obtain the following.

\begin{corollary}\Label{T3-J-C}
For any distribution $P_{A-T-B} $ with Markov condition $A-T-B$, 
there exists a sequence of codes $\{\Psi_{S}\}$ of joint decoding
with the rate pair $(R_A,R_B)$ with 
arbitrary small
positive parameters $r_A$ and $r_B$
to satisfy the following conditions.
For any channel $W=\{W_{a,b}\}_{a\in {\cal A},b\in {\cal B}}$, we have
\begin{align}
\lim_{n\to \infty}\frac{-1}{n}\log \epsilon(\Psi_{S};W^{(n)}) \ge & 
\min\Big( \max_s s(I_{1-s}(A;Y|BT)_{P_{A-T-B}}-R_A-r_A),r_A,  \nonumber \\
&\quad \max_s s(I_{1-s}(AB;Y|T)_{P_{A-T-B}}-R_A-R_B-r_A-r_B),r_A+r_B , \nonumber \\
&\quad \frac{1}{2} 
\max_s s(I_{1-s}(B;Y|AT)_{P_{A-T-B}}-R_B-r_B),\frac{r_B}{2} , \nonumber \\
&\quad \frac{1}{2}  \max_s s(I_{1-s}(AB;Y|T)_{P_{A-T-B}}-R_A-R_B-r_A-r_B),
\frac{r_A+r_B}{2} \Big)
\Label{Ex3BC}.
\end{align}
That is, when $R_A+r_A< I(A;Y|BT)_{P_{A-T-B}}$, 
$R_B +r_B< I(B;Y|AT)_{P_{A-T-B}}$, and 
$R_A +r_A+R_B +r_B< I(AB;Y|T)_{P_{A-T-B}}$,
the error probability $\epsilon(\Psi_{n};W^{(n)})$ exponentially  goes to zero.
\end{corollary}

\subsection{
Classical-quantum compound MAC}\Label{S7-5}
When the channel is not known,
we need to address classical-quantum compound MAC.
That is, to rigorously handle the optimization of the worst case in classical-quantum compound MAC, 
we consider a compound channel model ${\cal W}^{MAC}:=\{W_{a,b,\theta}\}_{a \in {\cal A},b \in {\cal B}}$  with a parametric space 
$\Theta$.
In this model, we do not know what $\theta \in \Theta$ is the true parameter.
Hence, we need to consider the worst case.
That is, a rate pair $(R_A,R_B)$ is achievable under the channel model ${\cal W}^{MAC}$ when there exists a sequence of codes 
$\{\Psi_n\}$
with the transmission rate pair $(R_A,R_B)$ such that the decoding error probability $\epsilon(\Psi_n,W_\theta^{(n)}) $ goes to zero
for any channel $W_\theta \in {\cal W}^{MAC}$.
The closure of the set of achievable rate pairs 
under the channel model ${\cal W}^{MAC}$ is called
the capacity region of the compound channel model ${\cal W}^{MAC}$,
and is denoted by ${\cal C}_{{\cal W}^{MAC}}$, i.e.,
\begin{align}
{\cal C}_{{\cal W}^{MAC}}:= \cl  {\{ (R_A,R_B)| (R_A,R_B) \hbox{ is achievable under the channel model }
{\cal W}^{MAC}.\}}
\end{align}

The above concept can be extended to the case with separate decoding.
A rate pair $(R_A,R_B)$ is achievable with separate decoding under the channel model ${\cal W}^{MAC}$ when there exists a sequence of codes 
$\{\Psi_{S,n}\}$
with separate decoding and the transmission rate pair $(R_A,R_B)$ such that the decoding error probability 
$\epsilon(\Psi_{S,n},W_\theta^{(n)}) $
goes to zero
for any channel $W_\theta \in {\cal W}^{MAC}$.
The closure of the set of achievable rate pairs 
with separate decoding under the channel model ${\cal W}^{MAC}$ is called
the capacity region with separate decoding of the compound channel model ${\cal W}^{MAC}$,
and is denoted by ${\cal C}_{S,{\cal W}^{MAC}}$, i.e.,
\begin{align}
{\cal C}_{S,{\cal W}^{MAC}}:= \cl  {\{ (R_A,R_B)| (R_A,R_B) \hbox{ is achievable with separate decoding under the channel model }
{\cal W}^{MAC}.\}}
\end{align}

Then, we obtain the following single-letterized form of the capacity region of 
a c-q compound MAC;
\begin{theorem}\Label{T4-5}
The relations
\begin{align}
&{\cal C}_{S,{\cal W}^{MAC}}={\cal C}_{{\cal W}^{MAC}}
= 
\hat{\cal C}_{{\cal W}^{MAC}}
\Label{MF}
\end{align}
hold, where
\begin{align*}
&\hat{\cal C}_{{\cal W}^{MAC}}
\nonumber \\
:= &
\cl \bigcup_{P_{A-T-B}}
\Big\{ (R_A,R_B) \Big|
R_A \le \min_\theta I(A;Y|BT)_{P_{A-T-B},\theta},\nonumber \\
&\hspace{10ex}R_B \le \min_\theta I(B;Y|AT)_{P_{A-T-B},\theta},\nonumber \\
&\hspace{10ex} R_A+ R_B \le \min_\theta I(AB;Y|T)_{P_{A-T-B},\theta}
\Big\},
\end{align*}
where $P_{A-T-B}$ is an arbitrary joint distribution on ${\cal A}\times {\cal B}\times {\cal T}$ to satisfy Markov
condition $A-T-B$.
\end{theorem}
The converse part relation of Theorem \ref{T4-5}; 
\begin{align}
{\cal C}_{{\cal W}^{MAC}}
\subset \hat{\cal C}_{{\cal W}^{MAC}}
\Label{MFC}
\end{align}
will be shown in Section \ref{SS6}.
Also, \eqref{MGY} implies 
${\cal C}_{S,{\cal W}^{MAC}} \subset {\cal C}_{{\cal W}^{MAC}}$.
Finally, Theorem \ref{T3-J} implies 
${\cal C}_{{\cal W}^{MAC}}\supset \hat{\cal C}_{{\cal W}^{MAC}}$.
Therefore, we obtain Theorem \ref{T4-5}.

One might consider that the direct application of Theorem \ref{T3} could derive 
Theorem \ref{T4-5} without use of \eqref{MGY}.
However, the direct application of Theorem \ref{T3} without use of \eqref{MGY}, i.e., Corollary \ref{Cor2} implies only the following weak version of the direct part.

\begin{corollary}\Label{Cor3}
We define two regions
\begin{align}
{\cal C}_{{\cal W}^{MAC}}^1
:= &
\cl \bigcup_{P_{A-T-B}}
\Big\{ (\min_\theta I(A;Y|T)_{P_{A-T-B},\theta}, \min_\theta I(B;Y|AT)_{P_{A-T-B},\theta})
\Big\} \\
{\cal C}_{{\cal W}^{MAC}}^2
:= &
\cl \bigcup_{P_{A-T-B}}
\Big\{ (\min_\theta I(A;Y|BT)_{P_{A-T-B},\theta}, \min_\theta I(B;Y|T)_{P_{A-T-B},\theta})
\Big\},
\end{align}
where $P_{A-T-B}$ is an arbitrary joint distribution on ${\cal A}\times {\cal B}\times {\cal T}$ to satisfy Markov
condition $A-T-B$.
Then, we have the following inclusion relation for capacity region;
\begin{align}
{\cal C}_{{\cal W}^{MAC}}
\supset &
\cl ({\cal C}_{{\cal W}^{MAC}}^1\cup {\cal C}_{{\cal W}^{MAC}}^2).
\Label{INC}
\end{align}
\end{corollary}

In fact, there is an example when the above inclusion relation is strict.
That is, 
Eq. \eqref{Ex3BC} and Theorem \ref{T3-J}
are essential for the direct part of 
Theorem \ref{T4-5}.
In the following, we introduce a measure to to check whether this inclusion relation is strict.
Corollary \ref{Cor3} shows only the following rate
 is achievable for $R_A+R_B$;
\begin{align}
R_1({\cal W}^{MAC}):=&
\max_{R_A,R_B}\{R_A+R_B| (R_A, R_B) \in 
\cl ({\cal C}_{{\cal W}^{MAC}}^1\cup {\cal C}_{{\cal W}^{MAC}}^2)
\} \nonumber\\
=&
\max\Big(\max_{P_{A-T-B}}
\min_\theta I(A;Y|T)_{P_{A-T-B},\theta}+\min_\theta I(B;Y|AT)_{P_{A-T-B},\theta}, \nonumber \\
&\quad\max_{P_{A-T-B}}
\min_\theta I(A;Y|BT)_{P_{A-T-B},\theta}+\min_\theta I(B;Y|T)_{P_{A-T-B},\theta}\Big).
\Label{GP1}
\end{align}
However, 
the combination of 
Eq. \eqref{Ex3BC} and Theorem \ref{T3-J} shows the achievability of 
the following value;
\begin{align}
& R_2({\cal W}^{MAC}):=
\max_{R_A,R_B}\{R_A+R_B| (R_A, R_B) \in 
{\cal C}_{{\cal W}^{MAC}}\} \nonumber\\
=&\max_{P_{A-T-B}}
\min \Big(
 \min_\theta I(AB;Y|T)_{P_{A-T-B},\theta},
 \min_\theta I(A;Y|BT)_{P_{A-T-B},\theta}+
 \min_\theta I(B;Y|AT)_{P_{A-T-B},\theta}\Big),
\Label{GP2L}
\end{align}
which is different from the following value;
\begin{align}
R_3({\cal W}^{MAC}):=
\max_{P_{A-T-B}} \min_\theta I(AB;Y|T)_{P_{A-T-B},\theta}.
\Label{GP2}
\end{align}
That is, we have the inequalities
\begin{align}
R_3({\cal W}^{MAC})\ge
R_2({\cal W}^{MAC})\ge
R_1({\cal W}^{MAC})
\Label{GP2M}.
\end{align}
As seen in examples in Section \ref{S3-5}, these three quantities are different values.
Such an example for the gap between $R_1({\cal W}^{MAC})$ and $R_2({\cal W}^{MAC})$
shows that the above inclusion relation \eqref{INC} is strict.
In addition, the difference between 
$R_3({\cal W}^{MAC})$ and 
$R_2({\cal W}^{MAC})$
shows the importance of 
$ \min_\theta I(A;Y|BT)_{P_{A-T-B},\theta}+
 \min_\theta I(B;Y|AT)_{P_{A-T-B},\theta}$.

Finally, we compare our single-letterized formula for the capacity region of cq compound channel
with the existing formula \cite[Theorem 3]{BJS2}. 
To state their result, we introduce the following notation.
When the channel is given as $n$ times use of the MAC $W_\theta$, 
and 
the two input systems $A$ and $B$ are subject to distributions $P_A$ and $P_B$ on ${\cal A}^n$ and ${\cal B}^n$
independently, 
the mutual information between $A$ ($B$) and the output quantum system $Y$ is written as
$I(A;Y)_{P_A \times P_B,W^{(n)}_\theta}$
($I(B;Y)_{P_A \times P_B,W^{(n)}_\theta}$).
By using this notation, the following relation can be shown. 
\begin{lemma}[\protect{\cite[Theorem 3]{BJS2}}]\Label{LL6}
The following relation holds;
\begin{align}
{\cal C}_{{\cal W}^{MAC}}
= &
\cl \bigcup_{n=1}^{\infty}
\bigcup_{P_{A}, P_B}
\bigcap_{\theta \in \Theta}
\Big\{ \Big(\frac{1}{n}I(A;Y)_{P_A \times P_B,W^{(n)}_\theta},\frac{1}{n}I(B;Y)_{P_A\times P_B,W^{(n)}_\theta}\Big)
\Big \},\Label{MFL}
\end{align}
\end{lemma}
The above preceding result \cite[Theorem 3]{BJS2} contains 
a limiting expression while our obtained formula does not contain such limiting expression.

\section{Examples}\Label{S3-5}

\subsection{Classical example 1}\Label{S93A}
To see the two types of gaps, the gap between \eqref{Hi1} and \eqref{Hi2}
and the gap between 
$R_1({\cal W}^{MAC})$ and 
$R_2({\cal W}^{MAC})$,
we consider a compound channel model ${\cal W}^{MAC}_1 $
composed of two classical MACs with ${\cal A}={\cal B}=\FF_2$ as follows.
We define the output variable $Y_\theta$ for $\theta=0,1$ as
\begin{align}
Y_0=A \oplus B \in \FF_2, \quad
Y_1=&(A \oplus Z_A, B \oplus Z_B)  \in \FF_2^2,
\end{align}
where $Z_A$ and $Z_B$ are independent variables 
and $P_{Z_A}(1)=P_{Z_B}(1)=p_0$ such that $h(p_0)=\frac{1}{2} $. 
The first MAC is called the sum modulo-2 multiple-access channel (S2MAC) \cite{Poltyrev}.
Using the parameters $p:= P_B(1)$ and  $q:= P_A(1)$,
the mutual information and the conditional mutual information are calculated as 
\begin{align}
I(B;Y_0|A)&= h(p),\quad  I(A;Y_0)= h(pq+(1-p)(1-q))- h(p) \\
I(A;Y_0|B)&= h(q),\quad  I(B;Y_0)= h(pq+(1-p)(1-q))- h(q) \\
I(B;Y_1|A)&= I(B;Y_1)= h(p p_0+(1-p)(1- p_0))-  \frac{1}{2} \\
I(A;Y_1|B)&= I(A;Y_1)= h(q p_0+(1-q)(1- p_0)) -\frac{1}{2} .
\end{align}
Based on the above formulas, the quantities \eqref{Hi1} and \eqref{Hi2} are calculated as follows.
We have
\begin{align}
&\max_{P_A,P_B} \{ \min(I(A;Y_0),I(A;Y_1))| I(B;Y_0|A),I(B;Y_1|A) \ge R\} 
=
 1- h(p_R) 
\Label{GJT} \\
&\max_{P_{A-T-B}} \{ \min(I(A;Y_0|T),I(A;Y_1|T))| I(B;Y_0|AT),I(B;Y_1|AT) \ge R\} \nonumber \\
=&\left\{
\begin{array}{ll}
0 & \hbox{ when } R \ge 1/2 \\
1-2R & \hbox{ when } \frac{1}{4} \le R \le 1/2 \\
\frac{1}{2} & \hbox{ otherwise},
\end{array}
\right.
\Label{GJT1}
\end{align}
where $p_R \in [0,1/2] $ satisfies $h(p_R p_0+(1-p_R)(1- p_0))-  \frac{1}{2}=R$.
These two quantities are numerically plotted in Fig. \ref{FF2}.
The quantities \eqref{GP1}, \eqref{GP2L}, and \eqref{GP2} are calculated as follows.
The relations
\begin{align}
R_1({\cal W}^{MAC}_1)=\frac{3}{4}, \quad
R_2({\cal W}^{MAC}_1)=R_3({\cal W}^{MAC}_1)=1\Label{OY0}
\end{align}
hold. 
In addition, we have
\begin{align}
{\cal C}_{{\cal W}_1^{MAC}}
=&
\Big\{(R_A,R_B)\Big| 
R_A\le \frac{1}{2}, R_B\le \frac{1}{2}
\Big\} \Label{OY1}\\
\cl ({\cal C}_{{\cal W}_1^{MAC}}^1\cup {\cal C}_{{\cal W}_1^{MAC}}^2)
=&
\Big\{(R_A,R_B)\Big| R_A+R_B\le \frac{3}{4}, 
R_A\le \frac{1}{2}, R_B\le \frac{1}{2} \Big\}.\Label{OY2}
\end{align}

Eq. \eqref{OY0} shows that
$R_2({\cal W}^{MAC}_1)$ is strictly larger than 
$R_1({\cal W}^{MAC}_1)$.
More precisely, 
Eqs. \eqref{OY1}  and \eqref{OY2} show that
the region ${\cal C}_{{\cal W}_1^{MAC}}$ 
is strictly larger than 
$\cl ({\cal C}_{{\cal W}_1^{MAC}}^1\cup {\cal C}_{{\cal W}_1^{MAC}}^2)$, as plotted in Fig. \ref{Fregion}.

\begin{figure}[htpb]
  \centering
    \begin{tabular}{cc}
 
 
      \begin{minipage}{0.48\hsize}
        \centering
 \includegraphics[width=0.8\linewidth]{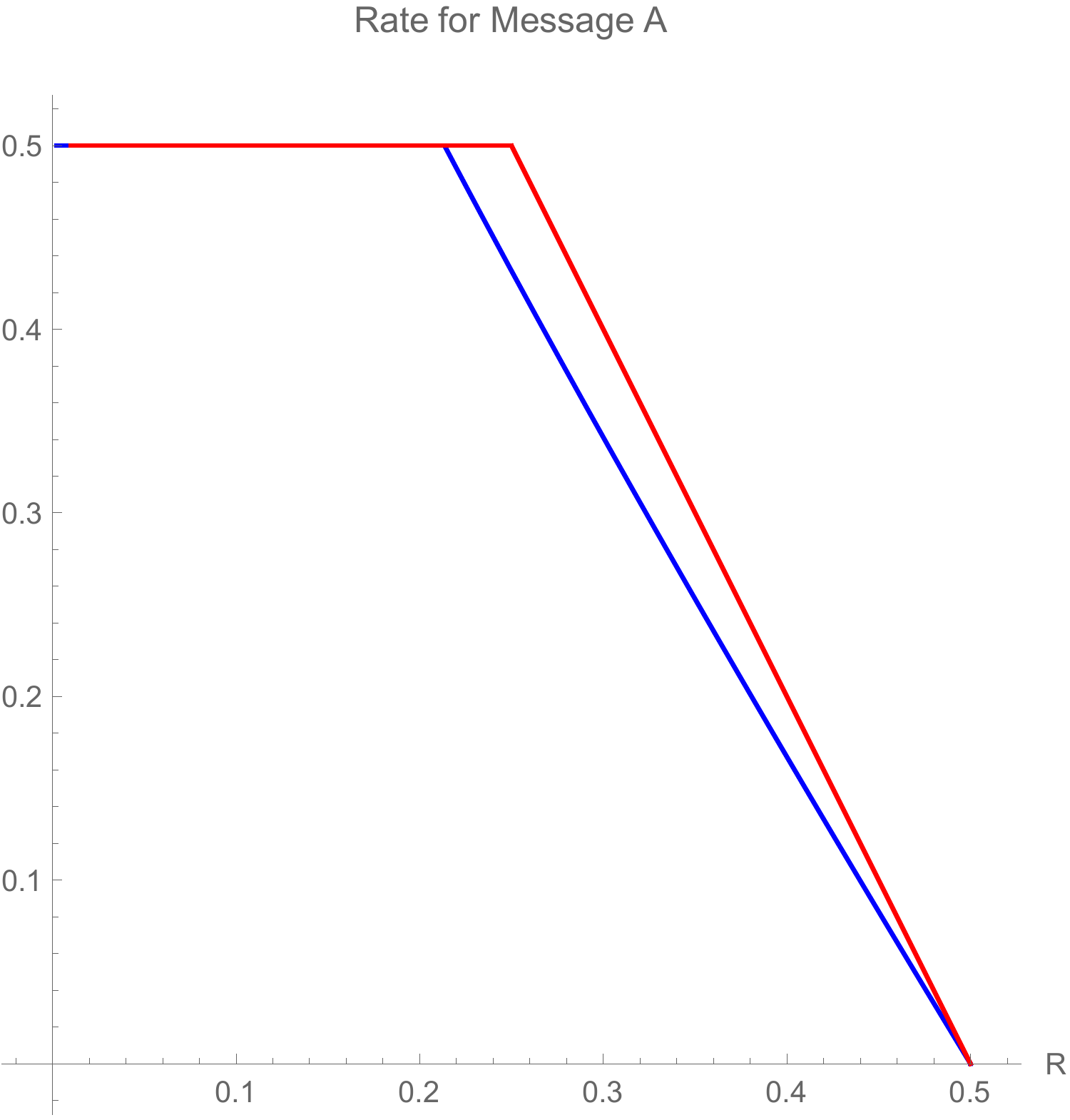}
 \caption{Numerical comparison between Eqs. \eqref{GJT1} and \eqref{GJT}.
Upper red line expresses Eq. \eqref{GJT1}.
Lower blue line expresses Eq. \eqref{GJT}.
}
\Label{FF2}
      \end{minipage}&
 
 
      \begin{minipage}{0.48\hsize}
        \centering
  \includegraphics[width=0.8\linewidth]{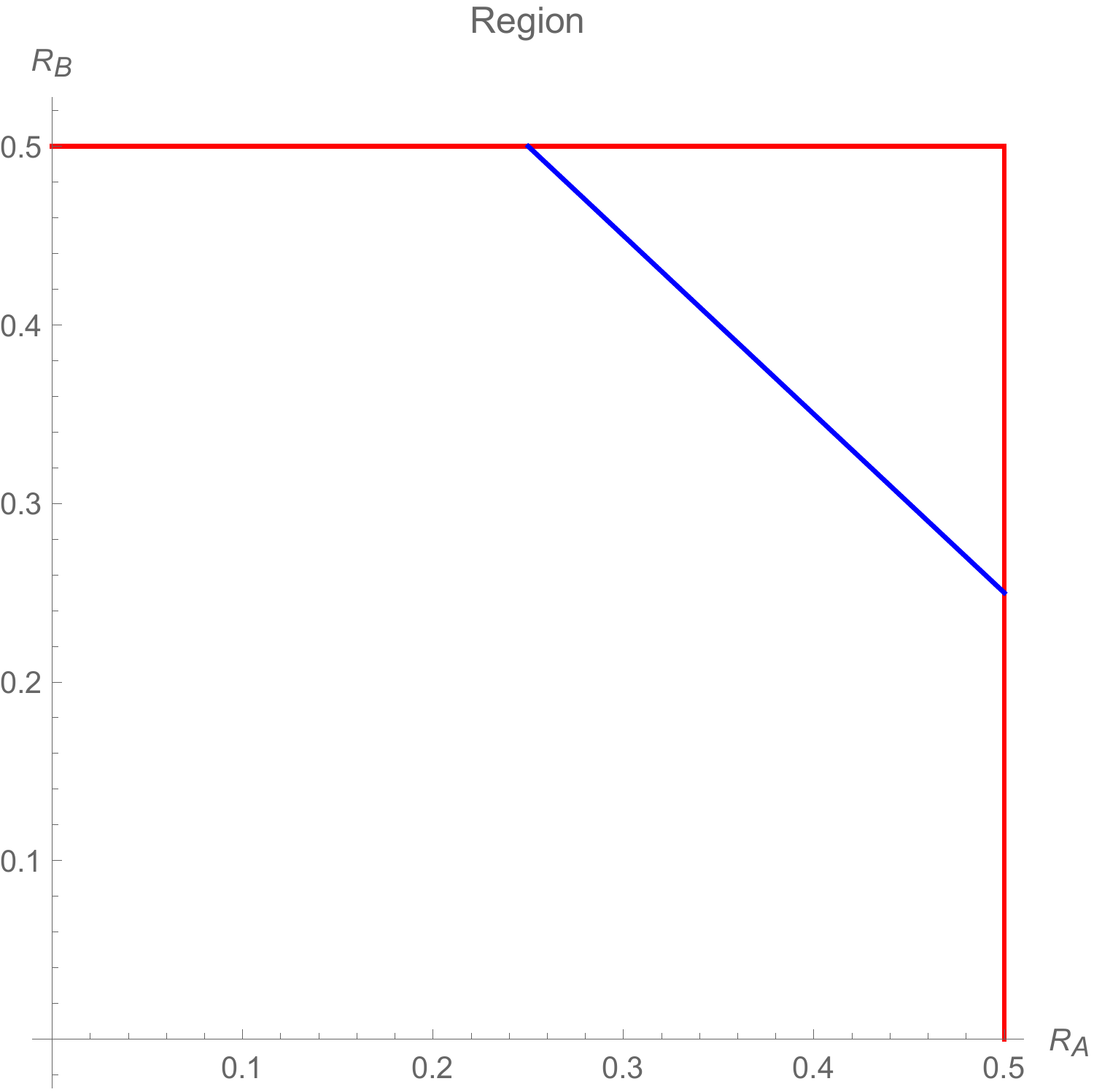}
\caption{Comparison of two regions
${\cal C}_{{\cal W}_1^{MAC}}$ 
and $\cl ({\cal C}_{{\cal W}_1^{MAC}}^1\cup {\cal C}_{{\cal W}_1^{MAC}}^2)$.
The former is strictly larger than the latter.
The red line expresses the boundary of 
${\cal C}_{{\cal W}_1^{MAC}}$. 
The blue line expresses the boundary of 
$\cl ({\cal C}_{{\cal W}_1^{MAC}}^1\cup {\cal C}_{{\cal W}_1^{MAC}}^2)$.}
\Label{Fregion}
      \end{minipage} 
\end{tabular}
\end{figure}

\begin{derivationof}{Eqs. \eqref{GJT} -- \eqref{OY2}}
The following derivations are partially based on a numerical calculation.
When $T$ is singleton,
we have
\begin{align}
\min(I(A;Y_0),I(A;Y_1))=&\min(1-h(p),\frac{1}{2}) ,\\
\min(I(B;Y_0|A),I(B;Y_1|A))
=&
\min(h(p),h(p p_0+(1-p)(1- p_0))-  \frac{1}{2}) \nonumber\\
=&
h(p p_0+(1-p)(1- p_0))-  \frac{1}{2},
\end{align}
where the final equation follows from the inequality
\begin{align}
h(p p_0+(1-p)(1- p_0))\le h(p)+h(p_0).
\end{align}

We define the function $f_1(R)$ as follows.
We choose $p_R$ such that $h(p p_0+(1-p)(1- p_0))-  \frac{1}{2}=R$.
Then, we set $f_1(R)=1-h(p_R)$.
LHS of \eqref{GJT} equals $\min(f_1(R),1/2)$. Hence, we obtain \eqref{GJT}.

We have $f_1(0)=1$ and $f_1(1/2)=0$
Since the function $\frac{f_1(R)-f_1(0)}{R}$ is monotonically increasing, as shown by Fig. \ref{FF3},
we have
\begin{align}
f_1(R)\le (1-2R)f_1(0)+ 2R f_1(1/2) = 1-2R.
\end{align}
for $R \in [0,1/2]$.
Thus, when $\sum_t P_T(t) h(p_t p_0+(1-p_t)(1- p_0))-  \frac{1}{2}=R$, we have
\begin{align}
\sum_t P_T(t) (1-h(p_t)) \le 1-2R.
\end{align}
Then, we have
\begin{align}
&\max_{P_T(t) }\Big\{
\min( \sum_t P_T(t) (1- h(p_t)), \frac{1}{2})
\Big|\sum_t P_T(t) (h(p_t p_0+(1-p_t)(1- p_0))-  \frac{1}{2} ) \le R
\Big\} \nonumber\\
=&
\min\Big(1-2R,
\frac{1}{2}\Big).
\end{align}
Hence, we obtain \eqref{GJT1}.

Eq. \eqref{GJT1} shows
\begin{align}
\max_{P_{A-T-B}}\min_\theta I(A;Y|T)_{P_{A-T-B},\theta}+\min_\theta I(B;Y|AT)_{P_{A-T-B},\theta}=
\frac{3}{4}. \Label{NLP}
\end{align}
Since this model is symmetric with respect to the exchange of $A$ and $B$,
Eq. \eqref{NLP} yields
\begin{align}
&\max \Big( \max_{P_{A-T-B}}\min_\theta I(A;Y|T)_{P_{A-T-B},\theta}+\min_\theta I(B;Y|AT)_{P_{A-T-B},\theta}, \nonumber \\
&\quad  \max_{P_{A-T-B}}\min_\theta I(A;Y|BT)_{P_{A-T-B},\theta}+\min_\theta I(B;Y|T)_{P_{A-T-B},\theta}
\Big)
\nonumber  \\
=& \frac{3}{4},
\Label{GJT3}
\end{align}
implies Eq. the first equation of \eqref{OY0}.

Considering the convex full of the region defined by \eqref{NLP} and its transposed region, 
we obtain \eqref{OY2}.

Then, we have
\begin{align}
\max_{p,q}\min_{\theta \in \Theta}I(AB;Y)_\theta
=&
\max_{p,q}\min
\Big(h(p p_0+(1-p)(1- p_0))-  \frac{1}{2} 
+h(q p_0+(1-q)(1- p_0)) -\frac{1}{2} ,\nonumber \\
& \quad h(q p+(1-q)(1- p))\Big)
= 1 \\
\max_{p,q}\min_{\theta \in \Theta}I(A;Y|B)_\theta
=&
\max_{p,q}\min
\Big(h(q p_0+(1-q)(1- p_0)) -\frac{1}{2} ,h(q)\Big)
= 1/2 \\
\max_{p,q}\min_{\theta \in \Theta}I(B;Y|A)_\theta
=&
\max_{p,q}\min
\Big(h(p p_0+(1-p)(1- p_0))-  \frac{1}{2} ,h(p) \Big)
= 1/2 .
\end{align}
The above maximum is attained when $p=q=1/2$.
Hence, we obtain the remaining equations in \eqref{OY0} and \eqref{OY1}.
\end{derivationof}

\subsection{Classical example 2}
\Label{S91}
We consider a compound channel model ${\cal W}^{MAC}_2 $ of two classical MAC
with ${\cal A}={\cal B}=\FF_2$ by defining the output variable $Y_i$ for $i=0,1$ as follows.
\begin{align}
Y_1=&
\left\{
\begin{array}{ll}
1 & \hbox{ when } A=B=1 \\
0 & \hbox{ otherwise.}
\end{array}
\right.
\\
Y_0=&
\left\{
\begin{array}{ll}
0 & \hbox{ when } A=B=0\\
1 & \hbox{ otherwise.}
\end{array}
\right.
\end{align}
Using the parameters $p:= P_B(1)$ and  $q:= P_A(1)$,
we have
\begin{align}
I(B;Y_1|A)&= q h(p),\quad  I(A;Y_0)= h(pq)-q h(p) \\
I(B;Y_0|A)&=(1- q) h(p),\quad  I(A;Y_1)= h((1-p)(1-q)) -(1-q) h(p) ,
\end{align}
where $h(p)$ is the binary entropy function. 
The quantities \eqref{Hi1} and \eqref{Hi2} are calculated as follows.
We have
\begin{align}
&\max_{P_A,P_B} \{ \min(I(A;Y_0),I(A;Y_1))| I(B;Y_0|A),I(B;Y_1|A) \ge R\} \nonumber \\
=&
\left\{
\begin{array}{ll}
h(1/4)-1/2 & \hbox{ when } R \le 1/2 \\
0 & \hbox{ otherwise}
\end{array}
\right.\Label{LGH}\\
&\max_{P_{A-T-B}} \{ \min(I(A;Y_0|T),I(A;Y_1|T))| I(B;Y_0|AT),I(B;Y_1|AT) \ge R\} \nonumber \\
\ge &\max_{q} \frac{h(\bar{p}_R q)+h((1-\bar{p}_R)(1-q))}{2} -R,\Label{LL2P}
\end{align}
where $\bar{p}_R$ chosen as $\frac{h(\bar{p}_R)}{2}  = R$.
These two quantities are numerically plotted in Fig. \ref{FF4}.
The quantities \eqref{GP1}, \eqref{GP2L}, and \eqref{GP2} are calculated as follows.
We have
\begin{align}
&{\cal C}_{{\cal W}_2^{MAC}}
=
\cl ({\cal C}_{{\cal W}_2^{MAC}}^1\cup {\cal C}_{{\cal W}_2^{MAC}}^2) \nonumber\\
=&
\bigcup_{p,q}
\Big\{(R_A,R_B)\Big| R_A+R_B\le \frac{H(p q)+H((1-p)(1- q))}{2}  , 
R_A\le \frac{h(q)}{2}, R_B\le \frac{h(p)}{2} \Big\}.\Label{OY4}
\end{align}
The relations
\begin{align}
R_1({\cal W}^{MAC}_2)=
R_2({\cal W}^{MAC}_2)=R_3({\cal W}^{MAC}_2)=\max_{p,q} 
\frac{H(p q)+H((1-p)(1- q))}{2} 
\Label{OY6}
\end{align}
hold. 
Hence, this example has no gap among
Eq. \eqref{GP2}, Eq. \eqref{GP2L}, and Eq. \eqref{GP1}.

\begin{figure}[htpb]
  \centering
    \begin{tabular}{cc}
 
       \begin{minipage}{0.48\hsize}
        \centering
  \includegraphics[width=0.8\linewidth]{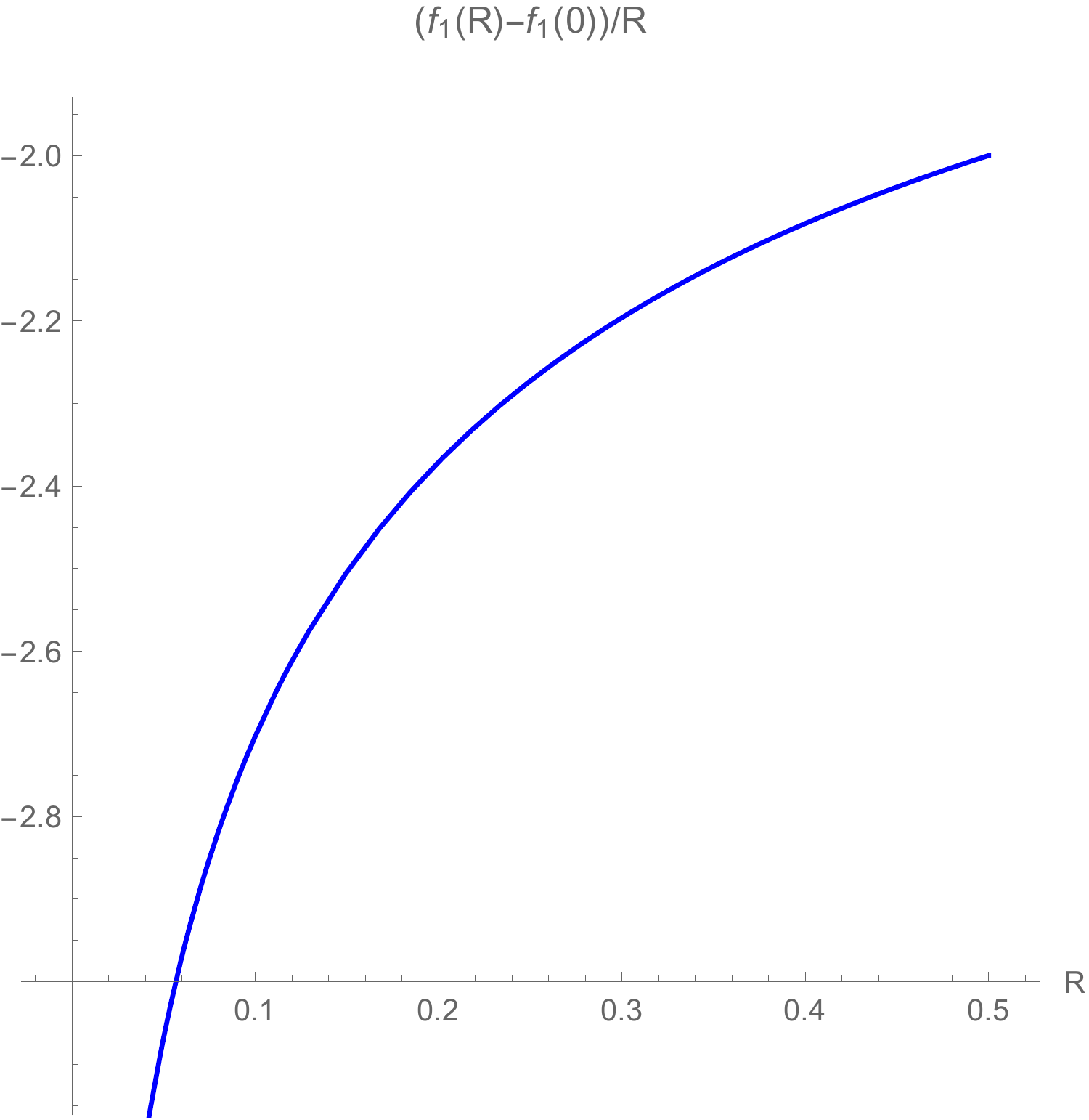}
\caption{Numerical verification for 
convexity of $f_1(R)$.
This graph shows that the function $\frac{f_1(R)-f_1(0)}{R}$
is monotonically increasing.}
\Label{FF3}
      \end{minipage} &

 
      \begin{minipage}{0.48\hsize}
        \centering
 \includegraphics[width=0.8\linewidth]{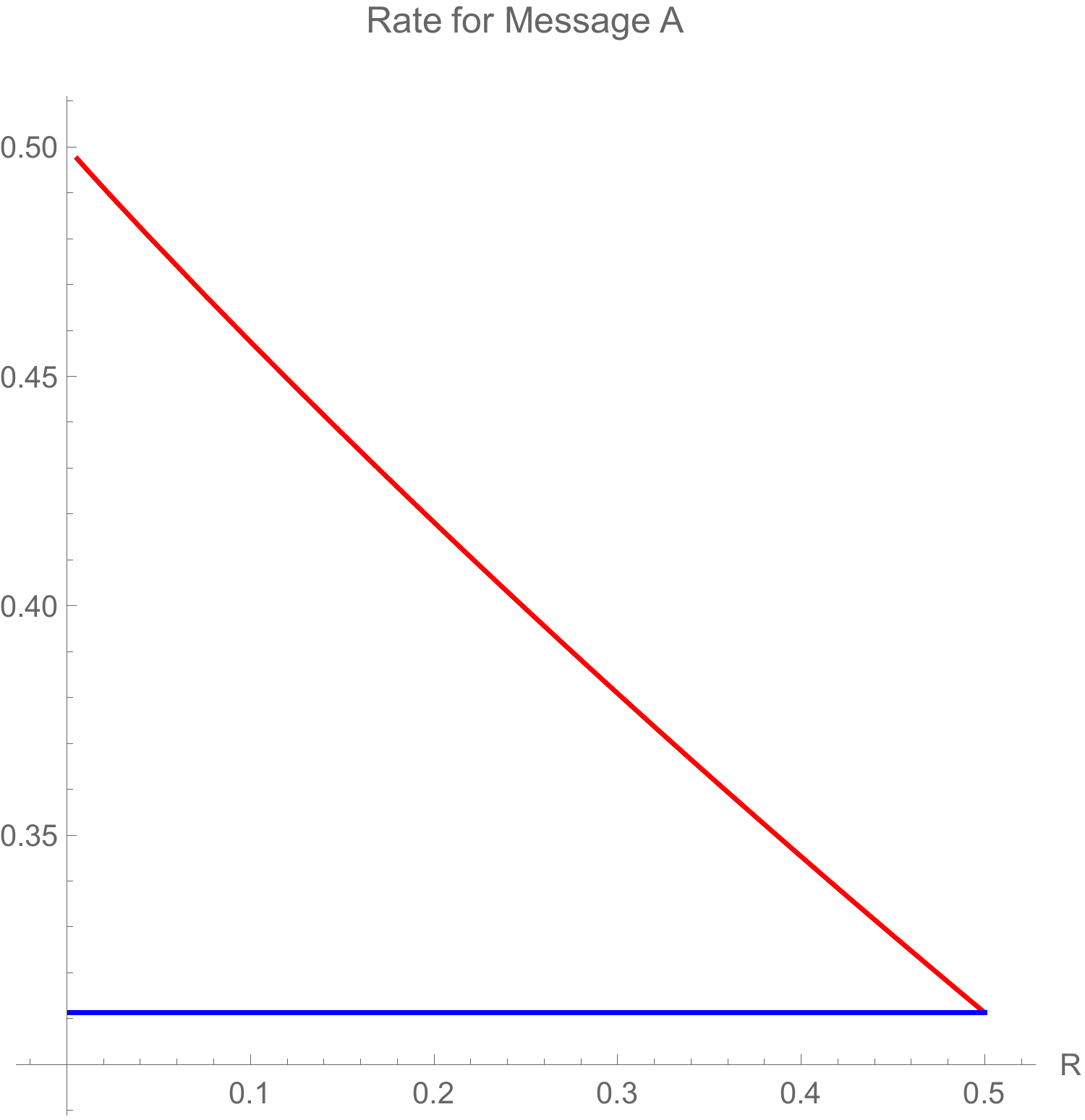}
 \caption{Numerical comparison between Eqs. \eqref{LGH} and \eqref{LL2P}.
Upper red line expresses Eq. \eqref{LL2P}.
Lower blue line expresses Eq. \eqref{LGH}.
}
\Label{FF4}
      \end{minipage}
\end{tabular}
\end{figure}
 
To derive Eqs. \eqref{LGH} -- \eqref{OY6}, we prepare the following statement.
\begin{statement}\Label{LM3}
The inequality 
\begin{align}
\min(h(pq)-q h(p),h((1-p)(1-q)) -(1-q) h(p))
\le h(1/4)-1/2
\end{align}
holds. The equality holds only when $p=q=1/2$.
\end{statement}
Statement \ref{LM3} is numerically shown by the numerical plot given in Fig. \ref{PO1}.

\begin{derivationof}{Eqs. \eqref{LGH} -- \eqref{OY6}}
We have
\begin{align}
&\max_{P_A,P_B} \{ \min(I(A;Y_0),I(A;Y_1))| I(B;Y_0|A),I(B;Y_1|A) \ge R\} \nonumber \\
=&
\max_{p,q}
\{\min(h(pq)-q h(p),h((1-p)(1-q)) -(1-q) h(p))|
\min(qh(p),(1-q)h(p)) \ge R\}.
\end{align}
The maximum $\max_{p,q}\min(qh(p),(1-q)h(p))$ is $1/2$, and it is attained only when $p=q=1/2$.
Hence, combining Lemma \ref{LM3}, we obtain \eqref{LGH}.

Let $\bar{P}_{A-T-B}$ be an arbitrary distribution on ${\cal A}\times {\cal T}\times {\cal B}$.
We define $(p_t,q_t)$ as
$p_t:=\bar{P}_{B|T}(1|t)$ and $q_t:=\bar{P}_{A|T}(1|t)$.
Then, we define the symmetrized distribution ${P}_{A-(T,J)-B}$ with ${\cal J}=\FF_2$ as follows.
\begin{align}
&{P}_{T,J}(t,0)={P}_{T,J}(t,1)=\bar{P}_{T}(t)/2 \nonumber\\
&{P}_{B|T,J}(1|t,0)=p_t, \quad 
{P}_{B|T,J}(1|t,1)=1-p_t \nonumber\\
&{P}_{A|T,J}(1|t,0)=q_t, \quad
{P}_{A|T,J}(1|t,1)=1-q_t .
\end{align}
Then, we have
\begin{align}
& \min(I(A;Y_0|T)_{P_{A-T-B}},I(A;Y_1|T)_{P_{A-T-B}}) \nonumber\\
=&\min( \sum_t P_T(t) h(p_tq_t)-q_t h(p_t), \sum_t P_T(t)h((1-p_t)(1-q_t)) -(1-q_t) h(p_t)) \nonumber\\
\le &\sum_t P_T(t) \frac{h(p_tq_t)-q_t h(p_t)+h((1-p_t)(1-q_t)) -(1-q_t) h(p_t)}{2} \nonumber\\
= &\sum_t P_T(t) \frac{h(p_tq_t)+h((1-p_t)(1-q_t)) - h(p_t))}{2} \nonumber\\
=& \min(I(A;Y_0|TJ)_{{P}_{A-(T,J)-B}},I(A;Y_1|JT)_{{P}_{A-(T,J)-B}}). \Label{FP1}
\end{align}
Similarly, we have
\begin{align}
& \min(I(B;Y_0|AT)_{P_{A-T-B}},I(B;Y_1|AT)_{P_{A-T-B}}) \nonumber\\
\le 
& \min(I(B;Y_0|ATJ)_{{P}_{A-(T,J)-B}},I(B;Y_1|AJT)_{{P}_{A-(T,J)-B}}) \nonumber\\
=&\sum_t P_T(t) \frac{h(p_t)}{2} \Label{FP2}\\
& \min(I(AB;Y_0|T)_{P_{A-T-B}},I(AB;Y_1|T)_{P_{A-T-B}}) \nonumber\\
\le
 & \min(I(AB;Y_0|TJ)_{{P}_{A-(T,J)-B}},I(AB;Y_1|JT)_{{P}_{A-(T,J)-B}}) \nonumber\\
= &\sum_t P_T(t) \frac{h(p_tq_t)+h((1-p_t)(1-q_t))}{2} .\Label{FP5}
\end{align}
Also, we have the same relations by exchanging $A$ and $B$ as
\begin{align}
& \min(I(A;Y_0|T)_{P_{A-T-B}},I(A;Y_1|T)_{P_{A-T-B}}) \nonumber\\
\le & \min(I(A;Y_0|TJ)_{{P}_{A-(T,J)-B}},I(A;Y_1|JT)_{{P}_{A-(T,J)-B}})\nonumber\\
= &\sum_t P_T(t) \frac{h(p_tq_t)+h((1-p_t)(1-q_t)) - h(q_t))}{2} \Label{FP3}\\
& \min(I(A;Y_0|BT)_{P_{A-T-B}},I(A;Y_1|BT)_{P_{A-T-B}}) \nonumber\\
\le 
& \min(I(A;Y_0|BTJ)_{{P}_{A-(T,J)-B}},I(A;Y_1|BJT)_{{P}_{A-(T,J)-B}}) \nonumber\\
=&\sum_t P_T(t) \frac{h(q_t)}{2} \Label{FP4}.
\end{align}
Due to these relations, we can restrict the joint distribution $P_{A-T-B}$
to the symmetrized distribution ${P}_{A-(T,J)-B}$.

As a simple case, we consider the case $T$ is singleton.
That is, we focus on $P_{AJB}=P_{A-J-B}$ as follows. 
$P_J$ is the uniform distribution.
\begin{align}
P_{A|J}(1|0)&=q,P_{A|J}(1|1)=1-q,\\
P_{B|J}(1|0)&=p,P_{B|J}(1|1)=1-p.
\end{align}
Considering the above joint distribution, we have
\begin{align}
&\max_{P_{A-T-B}} \{ \min(I(A;Y_0|T),I(A;Y_1|T))| I(B;Y_0|AT),I(B;Y_1|AT) \ge R\} \nonumber \\
\ge &\max_{p,q} 
\Big\{ \frac{(h(pq)-q h(p))+(h((1-p)(1-q)) -(1-q) h(p))}{2}\Big| \frac{q h(p)+(1-q) h(p)}{2}  \ge R\Big\} \nonumber \\
= &\max_{p,q} \Big\{ \frac{h(pq)+h((1-p)(1-q))}{2} -\frac{h(p)}{2}\Big| \frac{h(p)}{2}  \ge R\Big\} .
\end{align}
Choosing $\bar{p}_R$ as $\frac{h(\bar{p}_R)}{2}  = R$,
 we obtain \eqref{LL2P}.

Next, we show the remaining equations \eqref{OY4} and \eqref{OY6}.
The relations \eqref{FP1}, \eqref{FP2}, \eqref{FP3}, and \eqref{FP4} imply the relation
\begin{align}
&\Big\{(R_A,R_B)\Big| R_A+R_B\le \frac{H(p q)+H((1-p)(1- q))}{2}  , 
R_A\le \frac{h(q)}{2}, R_B\le \frac{h(p)}{2} \Big\} \nonumber\\
\subset & \cl ({\cal C}_{{\cal W}_2^{MAC}}^1\cup {\cal C}_{{\cal W}_2^{MAC}}^2)
\end{align}
for any pair of $(p,q)$.
Also, the relation \eqref{FP2}, \eqref{FP5}, and \eqref{FP5} implies 
\begin{align}
&{\cal C}_{{\cal W}_2^{MAC}}
\nonumber \\
\subset &
\bigcup_{p,q}
\Big\{(R_A,R_B)\Big| R_A+R_B\le \frac{H(p q)+H((1-p)(1- q))}{2}  , 
R_A\le \frac{h(q)}{2}, R_B\le \frac{h(p)}{2} \Big\}.
\end{align}
Hence, we obtain \eqref{OY4}.

The inequality 
$R_3({\cal W}^{MAC}_2) \le \max_{p,q} 
\frac{H(p q)+H((1-p)(1- q))}{2} $
holds as
\begin{align}
&\max_{P_{A-T-B}} \min (I(AB;Y|T)_{P_{A-T-B},0}, (AB;Y|T)_{P_{A-T-B},1}) \nonumber \\
= &\max_{(P_T, p_t,q_t)} 
\min \Big(   \sum_t P_T(t) h(p_t q_t), \sum_t P_T(t) h((1-p_t)(1- q_t))  \Big)\nonumber \\
= &\max_{(P_T, p_t,q_t)} 
 \sum_t P_T(t) \frac{h(p_t q_t)+h((1-p_t)(1- q_t))}{2}  \nonumber \\
= &\max_{p,q} 
\frac{h(p q)+h((1-p)(1- q))}{2}  .
\end{align}
Since 
$R_1({\cal W}^{MAC}_2)=R_2({\cal W}^{MAC}_2)=
\frac{H(p q)+H((1-p)(1- q))}{2} $ follows from \eqref{OY4},
combining \eqref{GP2M}, we obtain \eqref{OY6}.
\end{derivationof}

\begin{figure}[htpb]
  \centering
    \begin{tabular}{cc}
 
 
      \begin{minipage}{0.48\hsize}
        \centering
 \includegraphics[width=0.8\linewidth]{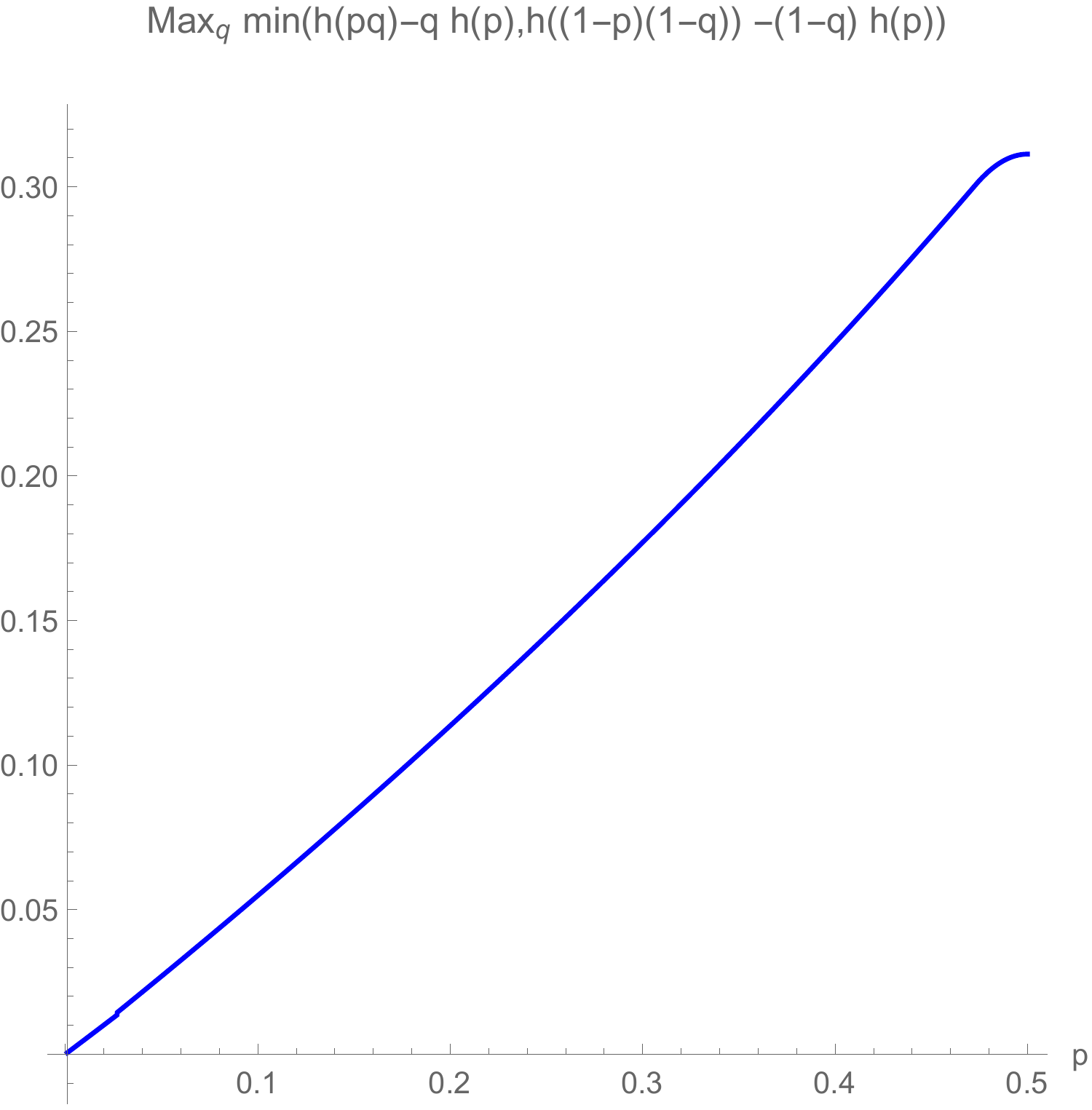}
 \caption{Numerical verification of Lemma \ref{LM3}.
This graph shows 
$\max_q \min(h(pq)-q h(p),h((1-p)(1-q)) -(1-q) h(p))$
as a function of $p$.}
\Label{PO1}
      \end{minipage} &

 \end{tabular}
\end{figure}

\subsection{Quantum examples}\Label{S92}
In this subsection, modifying the families of classical MACs given in Subsections \ref{S93A} and \ref{S91}, 
we show the existence of quantum examples to have gaps similar to classical examples given in Subsections \ref{S93A} and \ref{S91}.
First, we convert the classical system $\FF_2$ in each output system to the qubit system spanned by
$\{|0\rangle, |1\rangle\}$.
For example, the output system with $\theta=1$ of example ${\cal W}_1^{MAC}$
is converted to a two-qubit system.
We define the vector $
|\phi\rangle:= \cos \phi |0\rangle+ \sin \phi |1\rangle$.
The output information $0$ in the output system
is converted to the state $| \phi \rangle$,
and 
the output information $1$ in the output system
is converted to the state $|1\rangle$.

All the mutual information and all the conditional mutual information are
continuous for $\phi$.
Hence, when $\phi$ is close to zero,
these information quantities are close to the values in the above classical examples.
That is, the quantum example has the gaps presented in the above classical examples.
This example shows the importance of our universal code for c-q MAC as well as the codes given in Lemma \ref{LL6}.

\subsection{Quantum example for gap between $R_2({\cal W}^{MAC})$ and 
$R_3({\cal W}^{MAC})$}\Label{S93}
The examples in Subsections \ref{S93A} and \ref{S91} have no 
gap between $R_2({\cal W}^{MAC})$ and $R_3({\cal W}^{MAC})$.
To find an example for such a gap, 
we consider a c-q channel 
$\{W_x\}_{x \in {\cal X}}$ on the quantum system ${\cal H}_Y$.
For ${\cal A}={\cal B}={\cal X}$, 
we define two classical-quantum MACs as $W_{a,b,0}:=W_a$ and  $W_{a,b,1}:=W_b$.
Then, for any joint distribution $P_{A-T-B}$,
we have
\begin{align}
I(A;Y|T)_{P_{A-T-B},1} &=I(A;Y|BT)_{P_{A-T-B},1}=0, \\
I(B;Y|T)_{P_{A-T-B},0} &=I(B;Y|AT)_{P_{A-T-B},0}=0,
\end{align}
Hence, $R_2({\cal W}^{MAC})$ is zero.
\begin{align}
&\max_{P_{A-T-B}} \min (I(AB;Y|T)_{P_{A-T-B},0}, (AB;Y|T)_{P_{A-T-B},1}) \nonumber \\
= &\max_{P_A=P_B=P_X} \min (I(AB;Y|T)_{P_A\times P_B,0}, (AB;Y|T)_{P_A\times P_B,1}) \nonumber \\
= &\max_{P_X} I(X,Y)_{P_X},
\end{align}
where $I(X,Y)_{P_X}$ is the mutual information for 
the c-q channel $\{W_x\}_{x \in {\cal X}}$ on the quantum system ${\cal H}_Y$.
Hence, $R_3({\cal W}^{MAC})$ is strictly larger than zero unless the capacity of $W_x$ is zero.
That is,  
this example has a gap between $R_3({\cal W}^{MAC})$ and 
$R_2({\cal W}^{MAC})$.

\section{Proofs of Theorems \ref{T1} and \ref{T2-5}}\Label{S4}
\subsection{Proof of Theorem \ref{T1}}\Label{S4-1}
This section shows Theorem \ref{T1}.
Hayden, and Devetak \cite{YHD} showed the relation
\begin{align}
{\cal C}
\supset&
\cl
\bigcup_{P_{UX}} \Big\{(R_A,R_B)\Big| R_A \le \min(I(U;Y)_{P_{UX}},I(U;Z)_{P_{UX}}), R_B \le I(X;Y|U)_{P_{UX}} \Big)\Big\}_{P_{UX}} ,
\end{align}
which can be also proven by Corollary \ref{Cor1}.
When a rate pair $(R_A,R_B)$ is achievable,
$(R_A-r,R_B+r)$ is also achievable with an arbitrary $r \in [0,R_A]$ 
by converting a part of common message with rate $r$ into a private message.
Hence, we have the relations 
\begin{align}
&
\cl
\bigcup_{P_{UX}} \Big\{(R_A,R_B)\Big| R_A \le \min(I(U;Y)_{P_{UX}},I(U;Z)_{P_{UX}}), R_B \le I(X;Y|U)_{P_{UX}} \Big)\Big\}_{P_{UX}} 
\nonumber \\
\subset
&\cl
\bigcup_{P_{UX}} \Big\{(R_A,R_B)\Big| R_A \le \min(I(U;Y)_{P_{UX}},I(U;Z)_{P_{UX}}), R_A+R_B I(UX;Y)_{P_{UX}} \Big)\Big\}_{P_{UX}}\subset {\cal C}.
\end{align}
Therefore, it is sufficient to show the relation
\begin{align}
{\cal C}
\subset 
\cl
\bigcup_{P_{UX}} \Big\{(R_A,R_B)\Big| R_A \le \min(I(U;Y)_{P_{UX}},I(U;Z)_{P_{UX}}), R_B \le I(X;Y|U)_{P_{UX}} \Big)\Big\}_{P_{UX}} .
\end{align}

For this aim, we focus on a sequence of codes $\{\Psi_{n}\}$
with a transmission rate pair $(R_A,R_B)$, where
the encoder $\phi_{n}$ of $\Psi_{n}$ maps $(M_{A,n},M_{B,n})$ to $X^n$.
Then, we find that
\begin{align}
& I(M_{B,n};Y^n)
\stackrel{(a)}{\le}  I(M_{B,n};Y^n|M_{A,n} )
\stackrel{(b)}{\le}  I(X^n;Y^n|M_{A,n} ) \nonumber \\
=& \sum_{i=1}^n  I(X^n;Y_i| Y^{i-1}M_{A,n})
= \sum_{i=1}^n  I(X_i;Y_i| Y^{i-1}M_{A,n})
\stackrel{(c)}{\le} \sum_{i=1}^n  I(X_i;Y_i| M_{A,n}),
\end{align}
where each step can be shown as follows.
Step $(a)$ holds because $M_{B,n}$ is independent of $M_{A,n}$.
Step $(b)$ follows from the Markov chain 
$M_{B,n}-X^n-Y^n$ when $M_{A,n}$ is fixed.
Step $(c)$ follows from the Markov chain 
$Y^{i-1}-X_i-Y_i$ when $M_{A,n}$ is fixed.

Also, we find that
\begin{align}
& I(M_{A,n};Y^n)= \sum_{i=1}^n  I(M_{A,n};Y_i| Y^{i-1})
\stackrel{(a)}{\le} \sum_{i=1}^n  I(M_{A,n};Y_i),
\end{align}
where $(a)$ follows from the Markov chain $Y^{i-1}-M_{A,n}-Y_i $. 
Similarly, we have
\begin{align}
& I(M_{A,n};Z^n)
\le \sum_{i=1}^n  I(M_{A,n};Z_i),
\end{align}

Now, we introduce a new variable 
$I_n$ subject to the uniform distribution on $\{1, \ldots, n\}$.
We also define the conditional distribution 
$P_{UX|I_n}(u,x|i):=
P_{M_{A,n}X_i}(u,x)$.
That is, when $I_n=i$, $U$ and $X$ take the value $M_{A,n}$ and $X_i$.
Also, we define the variable $U_n:=(U I_n)$, and denote the joint distribution 
for $X$ and $U_n$ by $P_n$.
Hence, 
\begin{align}
&\frac{1}{n} I(M_{B,n};Y^n) \le I(X;Y| U I_n) = I(X;Y| U_n)_{P_n} \\
&\frac{1}{n} I(M_{A,n};Y^n) \le I(U;Y| I_n) \le I(U I_n;Y)= I(U_n;Y)_{P_n}\\ 
&\frac{1}{n} I(M_{A,n};Z^n) \le I(U;Y| I_n)\le I(U I_n;Z)= I(U_n;Z)_{P_n}.
\end{align}
Combining Fano's inequality, we can show that
\begin{align}
R_A &\le \liminf_{n \to \infty }\min (I(U_n;Y)_{P_n},I(U_n;Z)_{P_n}) \\
R_B&\le \liminf_{n \to \infty } I(X;Y|U_n)_{P_n}.
\end{align}
The above relation shows that the capacity region ${\cal C}$ is contained in the following set.
\begin{align}
\cl
\bigcup_{P_{UX}} \Big\{(R_A,R_B)\Big| R_A \le \min(I(U;Y)_{P_{UX}},I(U;Z)_{P_{UX}}), R_B \le I(X;Y|U)_{P_{UX}} \Big)\Big\}_{P_{UX}} .
\end{align}

\subsection{Proof of Theorem \ref{T2-5}}\Label{S4-2}
Since Corollary \ref{Cor1} shows the relation
\begin{align}
{\cal C}_{\cal W}
\supset &\cl
\bigcup_{P_{UX}} \Big\{(R_A,R_B)\Big| R_A \le  \min_\theta \min(I(U;Y)_{P_{UX},\theta},I(U;Z)_{P_{UX},\theta}), 
R_B \le \min_\theta I(X;Y|U)_{P_{UX},\theta} \Big)\Big\}_{P_{UX}} ,
\end{align}
it is sufficient to show the opposite relation.

We apply the discussion of the above subsection.
Then,
our choice of $P_n$ does not depend on the channel parameter $\theta$. 
Hence, we have 
\begin{align}
{\cal C}_{\cal W}
\subset &\cl
\bigcup_{P_{UX}}
\bigcap_{\theta\in \Theta}\Big\{(R_A,R_B)\Big| R_A \le  \min(I(U;Y)_{P_{UX},\theta},I(U;Z)_{P_{UX},\theta}), 
R_B \le I(X;Y|U)_{P_{UX},\theta} \Big)\Big\}_{P_{UX}} \nonumber \\
=&\cl
\bigcup_{P_{UX}}\Big\{(R_A,R_B)\Big| R_A \le  \min_\theta \min(I(U;Y)_{P_{UX},\theta},I(U;Z)_{P_{UX},\theta}), 
R_B \le \min_\theta I(X;Y|U)_{P_{UX},\theta} \Big)\Big\}_{P_{UX}} .
\end{align}

\section{Proof of Converse part of Theorem \ref{T4-5}}\Label{SS6}
This section shows Eq. \eqref{MFC}, i.e., the converse part of Theorem \ref{T4-5}.
For this aim, we focus on a sequence of codes $\{\Psi_{n}\}$
with a transmission rate pair $(R_A,R_B)$, where the respective encoders of
$\Psi_{n}$ map $M_{A,n}$ and $M_{B,n}$ to $A^n$ and $B^n$, respectively.
Then, for any $\theta\in \Theta$, we find that
\begin{align}
& I(M_{B,n};Y^n)
\stackrel{(a)}{\le}  I(M_{B,n};Y^n|M_{A,n} )_\theta
=  I(B^n;Y^n|A^n )_\theta \nonumber \\
=& \sum_{i=1}^n  I(B^n;Y_i| Y^{i-1} A^n)_\theta
= \sum_{i=1}^n  I(B_i;Y_i| Y^{i-1}A^n)_\theta
\stackrel{(b)}{\le} \sum_{i=1}^n  I(B_i;Y_i| A^n)_\theta,
= \sum_{i=1}^n  I(B_i;Y_i| A_i)_\theta,
\end{align}
where each step can be shown as follows.
Step $(a)$ holds because $M_{B,n}$ is independent of $M_{A,n}$.
Step $(b)$ follows from the Markov chain 
$Y^{i-1}-B_i-Y_i$ when $A_n$ is fixed.
Similarly, we have
\begin{align}
 I(M_{A,n};Y^n)_\theta
\le \sum_{i=1}^n  I(A_i;Y_i| B_i)_\theta.
\end{align}

Also, we find that
\begin{align}
&
 I(M_{A,n} M_{B,n};Y^n)_\theta= 
 I(A^n,B^n ;Y^n)_\theta= 
\sum_{i=1}^n  I(A^n,B^n;Y_i| Y^{i-1})_\theta
\stackrel{(a)}{\le} \sum_{i=1}^n  I(A^n,B^n;Y_i)_\theta,
\stackrel{(b)}{=} \sum_{i=1}^n  I(A_i,B_i;Y_i),
\end{align}
where $(a)$ 
follows from the Markov chain $Y^{i-1}-(A^n,B^n)-Y_i $. 
$(b)$
follows from the Markov chain $A^n,B^n-(A_,b_i)-Y_i $. 

Now, we introduce a new variable 
$U_n$ subject to the uniform distribution on $\{1, \ldots, n\}$.
We also define the conditional distribution 
$P_{AB|I_n}(a,b|i):=
P_{A_i B_i}(a,b)$.
That is, when $I_n=i$, $A$ and $B$ take the value $A_i$ and $B_i$.
Also, we denote the joint distribution 
for $A,B$ and $U_n$ by $P_n$.
Hence, 
\begin{align}
&\frac{1}{n} I(M_{B,n};Y^n) \le I(B;Y| A U_n)_{P_n,\theta} \\
&\frac{1}{n} I(M_{A,n};Y^n) \le I(A;Y| B U_n)_{P_n,\theta}\\ 
&\frac{1}{n} I(M_{A,n} M_{A,n};Y^n) \le I(AB;Y| U_n)_{P_n,\theta}.
\end{align}
Combining Fano's inequality, we can show that
\begin{align}
R_A &\le \liminf_{n \to \infty }I(B;Y| A U_n)_{P_n,\theta} \\
R_B&\le \liminf_{n \to \infty } I(A;Y| B U_n)_{P_n,\theta}\\
R_A+R_B&\le \liminf_{n \to \infty } I(AB;Y|  U_n)_{P_n,\theta}.
\end{align}
Since the above inequalities hold for any $\theta \in \Theta$, we have
\begin{align}
R_A &\le \min_\theta \liminf_{n \to \infty }I(B;Y| A U_n)_{P_n,\theta} \\
R_B&\le \min_\theta \liminf_{n \to \infty } I(A;Y| B U_n)_{P_n,\theta}\\
R_A+R_B&\le \min_\theta  \liminf_{n \to \infty } I(AB;Y|  U_n)_{P_n,\theta}.
\end{align}
The above relation shows 
\begin{align}
{\cal C}_{{\cal W}^{MAC}}
\subset &
\cl \bigcup_{P_{A-T-B}}
\Big\{ (R_A,R_B) \Big|
R_A \le \min_\theta I(A;Y|BU)_{P_{A-T-B},\theta},
R_B \le \min_\theta I(B;Y|AU)_{P_{A-T-B},\theta},\nonumber\\
&\hspace{10ex} R_A+ R_B \le \min_\theta I(AB;Y|U)_{P_{A-T-B},\theta}
\Big\},
\end{align}
which implies \eqref{MFC}.

\section{Method of types}\Label{S5}
The aim of this section is 
the derivation of simple consequences of generalized packing lemmas by \cite{Korner-Sgarro,Liu-Hughes} 
as the preparation of our proofs of Theorems \ref{T2}, \ref{T3}, and \ref{T3-J}.
Subsection \ref{S5-1} reviews the existing results for the method of types given in
\cite{CK},\cite[Section 4]{Ha1}, \cite[Chapter 6]{Group2}, \cite[Section IV]{Ha1}.
Then, the remaining two subsections give extensions of the above contents to 
the cases with superposition codes and MAC codes by using 
generalized packing lemmas by \cite{Korner-Sgarro,Liu-Hughes}.
These contents take an essential role for our universal construction of codes of both settings.

\subsection{Single terminal}\Label{S5-1}
In this subsection, we prepare the notations for the method of types
and reviews the existing result on this topic.
Before starting this discussion, we introduce one notation.
Given a distribution $P_U$ on ${\cal U}$
and a conditional distribution $P_{X|U}$ on ${\cal X}$ with condition in ${\cal U}$,
we define the joint distribution $P_{X|U}\cdot P_U$ on ${\cal U}\times {\cal X}$ as
\begin{align}
P_{X|U}\cdot P_U(x,u):= P_{X|U}(x|u) P_U(u).
\end{align}
For any subset $\Omega \subset {\cal X}$,
we define the uniform distribution $P_{\Unif,\Omega}$ on $ \Omega$ as
\begin{align*}
P_{\Unif,\Omega}(x):= 
\left\{
\begin{array}{cc}
\frac{1}{|\Omega|} & x \in \Omega \\
0 & x \notin \Omega.
\end{array}
\right.
\end{align*}
Also, we denote the cardinality, i.e., the number of elements, of the set ${\cal X}$ by $d_X$. 

The content of this subsection follows the content of 
\cite[Section 4]{Ha1}, \cite[Chapter 6]{Group2}, \cite[Section IV]{Ha1}.
The remaining subsections of this section are two types of extensions of this content
by using the results by \cite[Lemma]{Korner-Sgarro} and \cite[Section IV]{Ha1}.
The key point of this section is to provide a subset to satisfy the following property by using the method of types.
In information theory, we usually employ the random coding method.
However, to construct a deterministic universal code
unlike the existing papers \cite{BJS1,BJS2},
we need to avoid such random construction of the encoder
because a code whose decoding error probability is less than the average might depend on the true channel.
To resolve this problem, we employ the packing lemma of the method of types and its two types of generalizations.

First, we prepare notations for the method of types.
Given an element $\bx \in {\cal X}^n$ and an element $x\in {\cal X}$, 
we define the subset ${\cal N}(\bx,x) := \{ i| x_i=x\} $, 
the integer $n(\bx,x):=| {\cal N}(\bx,x) |$,
and the empirical distribution 
$T_Y(\bx):= (\frac{n_1}{n},\ldots, \frac{n_{d_X}}{n})$, which is called a type, where 
$n(\bx,x)$ is simplified to $n_x$.
The set of types is denoted by $T_n({\cal X})$.
For ${P} \in T_n({\cal X})$, a subset of ${\cal X}^n$ is defined by:
\begin{align*}
&T_{P}^n({\cal X})
:= \{\bx \in {\cal X}^n| 
T_Y(\bx)=P\}.
\end{align*}
We simplify $T_{P}^n({\cal X})$ to $T_P$ when we do not need to identify $n$ and ${\cal X}$.
Since $ \frac{e^{nH(P)}}{|T_P|} \le (1+n)^{d_X}$,
the uniform distribution $P_{\Unif,T_P}$ on the subset $T_P$ satisfies
\begin{align}
 P_{\Unif,T_P}(\bx)  \le |T_n({\cal X})|   P^{n}(\bx)
\le (1+n)^{d_X}  P^{n}(\bx)
\Label{eq3-X}.
\end{align}
As a generalization, 
for a type $P_{UX}\in T_n({\cal U}\times{\cal X})$
and $\bu \in T_{P_U}^n({\cal U})$,
we define
\begin{align*}
&T_{P_{X|U}}^n({\cal X}|\bu)
:= \{\bx \in {\cal X}^n| 
T_Y(\bx,\bu)=P_{UX}\}.
\end{align*}
The occurring probability of $\bx \in  T_P$ under the distribution $P^n$ is characterized as
\begin{align}
 P^{n}(\bx)=e^{\sum_{i=1}^d n_i \log P(i) } =e^{-nH(P)} 
\Label{eq3-2}.
\end{align}

Given another finite set ${\cal T}$, the sequence of types ${\bV}=({v}_1, \ldots, {v}_d)\in 
T_{n_1}({\cal T})\times \cdots \times T_{n_d}({\cal T})$ 
is called a conditional type for $\bx$ and can be regarded as a conditional distribution
when the type of $\bx$ is $(\frac{n_1}{n}, \ldots, \frac{n_d}{n})$
\cite{CK}.
We denote the set of conditional types for $\bx$ 
by $V(\bx,{\cal T})$, i.e.,
\begin{align}
V(\bx,{\cal T}):= \{\bV|
{\bV} \cdot (T_Y(\bx)) \in T^n({\cal X}\times {\cal T}) \}.
\end{align}
A conditional type $\bV \in V(\bx,{\cal X})$ is called identical when
$\bV(x|x')=\delta_{x,x'}$.
This concept is generalized to the case when the input system is composed of two system 
${\cal U}$ and ${\cal X}$.
For an element $(\bu,\bx) \in ({\cal U}\times {\cal X})^n$,
a conditional type $\bV \in V((\bu,\bx),{\cal X})$ is called identical when
$\bV(u,x|x')=\delta_{x,x'}$.
For any conditional type $\bV\in V(\bx,{\cal T})$,
we define the subset of ${\cal T}^n$:
\begin{align*}
 T_{{\bV}}(\bx) 
:= \left\{\bt \in {\cal T}^n\left|
T_Y(\bx, \bt)
= {\bV} \cdot (T_Y (\bx))
\right.\right\}.
\end{align*}
For a type $P \in T_n({\cal X})$ and an element $u_o \in {\cal U}$, 
we define a type $P\times u_o \in T_n({\cal X}\times {\cal U})$
as $P\times u_o(x,u):= P(x)\delta_{u_o,u} $.
For a conditional type $\bV \in V(x,{\cal T})$ and an element $u_o \in {\cal U}$, 
we define a conditional type $\bV\times u_o \in V(x,{\cal T}\times {\cal U})$
as $\bV\times u_o(t,u|x):= \bV(t|x)\delta_{u_o,u} $.

Then, 
the previous studies \cite[Section 4]{Ha1}, \cite[Chapter 6]{Group2}, \cite[Section IV]{Ha1}
stated a modification of 
Csisz\'{a}r-K\"{o}rner's packing lemma \cite[Lemma 10.1]{CK} as follows.

\begin{proposition}\Label{l1}
For a positive number $R>0$, there exists a sufficiently large integer $N$ satisfying the following.
For any integer $n \ge N$ and any type ${P} \in T_n({\cal X})$ satisfying $R< H(P)$,
there exist $\mathsf{M}_n:=e^{n R-n^{3/4}}$ distinct elements 
\begin{align*}
\hat{\cal M}_n:=\{ \bx(1),\ldots, \bx(\mathsf{M}_n)\} \subset T_P
\end{align*}
such that 
the inequality
\begin{align}
|T_{{\bV}}(\bx) \cap (\hat{\cal M}_n\setminus \{\bx\})|
\le |T_{{\bV}}(\bx)| 
e^{-n(H(P)-R)} 
\Label{20}
\end{align}
holds for every $\bx \in \hat{\cal M}_n \subset T_{{P}}$ and 
every conditional type ${\bV} \in V(\bx,{\cal X})$.
\hfill $\square$\end{proposition}

This proposition is shown in \cite[Appendix C]{Conti} by using Csisz\'{a}r and K\"{o}rner\cite[Lemma 10.1]{CK}.
This proposition was used to make an universal encoder for one-to-one channel coding
in the existing studies \cite[Section 4]{Ha1}, \cite[Chapter 6]{Group2}, \cite[Section IV]{Ha1}, in which
the choice of the universal encoder does not depend on the output alphabet nor the output quantum system
because the employed packing lemma treats the conditional types from the input alphabet to the input alphabet.
Using this proposition, the paper \cite{Conti} derives a useful proposition.

To state it, we focus on the permutation group $S_n$ on $\{1, \ldots, n\}$.
For any $\bx \in {\cal X}^n$, we define an invariant subgroup 
$S_{\bx}\subset S_n$, where $S_n$ is the permutation group with degree $n$:
\begin{align*}
S_{\bx}: = \{g \in S_n | g(\bx)=\bx \}.
\end{align*}
Then, we have the following proposition, which takes a central role to
reduce our evaluation of the decoding error probability of this deterministic encoder given by Proposition \ref{l1}
to the evaluation of the decoding error probability under the random coding \cite[Section 4]{Ha1}, \cite[Chapter 6]{Group2}, \cite[Section IV]{Ha1}.

\begin{proposition}[\protect{\cite[Eq.(31)]{Conti}}]\Label{PP7}
Assume that $\bx \in \hat{\cal M}_n$.
Any element $\bx'(\neq \bx) \in T_{P}^n({\cal X})$ satisfy
\begin{align}
&\sum_{g \in S_{\bx}}
\frac{1}{|S_{\bx}|}
P_{\Unif,\hat{\cal M}_n}\circ g (\bx')
\le  P^n (\bx') e^{n^{3/4} } 
.  \Label{8}
\end{align}
\end{proposition}

\subsection{Superpostion code}
This subsection extends the contents of the previous subsection
to the setting for superpostion codes.
K\"{o}rner and Sgarro \cite[Lemma]{Korner-Sgarro}
extended the packing lemma by Csisz\'{a}r and K\"{o}rner\cite{CK}
to the case with superpostion code.
In the same way as Proposition \ref{l1},
Lemma 1 of \cite{Korner-Sgarro} can be rewritten as follows when $\hat{V}$ is the identical conditional type.

\begin{proposition}\Label{lB-3}
For two positive numbers $R_U,R_V>0$, there exists a sufficiently large integer $N$ satisfying the following.
For any integer $n \ge N$ and any types 
${P}_{UX} \in T_n({\cal U}\times {\cal X})$
 satisfying $R_U< H(P_U)$ and $R_B< H(V|U)_{P_{UX}}$,
we define
$\mathsf{M}_{U,n}:=e^{n R_U-n^{3/4}}$ 
and $\mathsf{M}_{X,n}:=e^{n R_X-n^{3/4}}$.
There exist a subset $\hat{\cal M}_{U,n}$ with $\mathsf{M}_{U,n}$ distinct elements
and a subset $\hat{\cal M}_{X,n,j}$
 with $\mathsf{M}_{X,n}$ distinct elements 
for $j=1,\ldots, \mathsf{M}_{U,n}$ as
\begin{align*}
\hat{\cal M}_{U,n} &:=\{\bu(1),\ldots, \bu(\mathsf{M}_{U,n})\} \subset T_{P _U}^n({\cal U})\\
\hat{\cal M}_{X,n,j} &:=\{ \bx(j,1),\ldots, \bx(j,\mathsf{M}_{X,n})\} 
\end{align*}
such that 
$(\bu(j), \bx(j,k)) \in T_{P_{UX}}^n({\cal U}\times {\cal X})$ and
the inequalities
\begin{align}
& \Big|T_{{\bV}} (\bu(j),\bx(j,k)) \cap 
\Big(\bigcup_{j'\neq j}
\big(\{\bu(j')\}\times \hat{\cal M}_{X,n,j'}  \big) \Big) \Big| \nonumber \\
\le & |T_{{\bV}}(\bu(j),\bx(j,k))| 
e^{-n(H(P_{UX})-R_A-R_B)} 
\Label{Y3} 
\\
& |T_{{\bV}_X}(\bu(j),\bx(j,k)) \cap 
(\hat{\cal M}_{X,n,j}\setminus \{\bx(j,k)\})| \nonumber \\
\le & |T_{{\bV}_X}(\bu(j),\bx(j,k))| 
e^{-n(H(X|U)_{P_{UX}}-R_V)} 
\Label{Y2} \\
& |T_{{\bV}_U}(\bu(j)) \cap 
(\hat{\cal M}_{U,n}\setminus \{\bu(j) \})
| \nonumber \\
\le & |T_{{\bV}_U}(\bu(j),\bx(j,k)) | 
e^{-n(H(P_U)-R_U)} 
\Label{Y1} 
\end{align}
hold for any $j \in \{1, \ldots, \mathsf{M}_{U,n}\}$, $k \in \{1, \ldots, \mathsf{M}_{X,n}\}$,
and
any conditional types ${\bV} \in V((\bu(j),\bx(j,k)),{\cal U}\times {\cal X})$,
${\bV}_X \in V((\bu(j),\bx(j,k)),{\cal X})$, and
${\bV}_U \in V(\bu(j),{\cal U})$.
\hfill $\square$\end{proposition}

Our universal encoder for classical-quantum superposition code
is given by the above construction, and has
a decoding error probability essentially equivalent to the average performance of random coding.
To derive Proposition \ref{lB-3}, we choose $\delta=n^{-\frac{1}{4}}$ in \cite[Lemma 1]{Korner-Sgarro}.
Eq. \eqref{Y3} of Proposition \ref{lB-3} follows from Eq. (1) of \cite[Lemma 1]{Korner-Sgarro} 
with substituting ${\cal U}\times {\cal X}$ and the identical conditional type into ${\cal Y}$ and 
$\hat{V}$, respectively.
Eq. \eqref{Y2} of Proposition \ref{lB-3} follows from Eq. (3) of \cite[Lemma 1]{Korner-Sgarro} 
with substituting ${\cal X}$ and $\hat{V}$ into ${\cal Z}$ and $\hat{V}$, respectively.
Eq. \eqref{Y1} of Proposition \ref{lB-3} follows from Eq. (2) of \cite[Lemma 1]{Korner-Sgarro} 
with substituting ${\cal U}$ and $\hat{V}$ into ${\cal Y}$ and $\hat{V}$, respectively.

As a generalization of Proposition \ref{PP7}, we have the following lemma for
the set ${\cal M}_{UX,n}:=\bigcup_{j}
\big(\{\bu(j)\}\times \hat{\cal M}_{X,n,j}  \big) $, which will be used for 
our evaluation of the decoding error probability of our universal c-q superposition code.
\begin{lemma}\Label{XP7}
Assume that $(\bu,\bx) \in \hat{\cal M}_{UX,n}$.
Any element $(\bu',\bx')\in T_{P_{UX}}^n({\cal U}\times {\cal X})$ with
$\bu\neq \bu'$ satisfies
\begin{align}
&\sum_{g \in S_{\bu,\bx}}
\frac{1}{|S_{\bu,\bx}|}
P_{\Unif,\hat{\cal M}_{UX,n}}\circ g (\bu',\bx')
\le P_{UX}^n (\bu',\bx') e^{2 n^{3/4} }
  \Label{8K}.
\end{align}
Any element $(\bu,\bx')\in T_{P_{UX}}^n({\cal U}\times {\cal X})$ with
$\bx\neq \bx'$ satisfies
\begin{align}
&\sum_{g \in S_{\bu,\bx}}
\frac{\mathsf{M}_{U,n}}{|S_{\bu,\bx}|}
P_{\Unif,\hat{\cal M}_{UX,n}}\circ g (\bu,\bx')
\le  P_{X|U}^n(\bx'|\bu) e^{n^{3/4} }
  \Label{8I}.
\end{align}
\end{lemma}

\begin{proof}
First, we show \eqref{8K}.
We choose a conditional type $\bV \in V((\bu,\bx),{\cal U}\times {\cal X})$
such that 
$\bV$ is non-identical, i.e., $\bV(u,x|u',x')\neq \delta_{(u,x),(u'x')}$ and
$(\bu',\bx') \in T_{{\bV}}(\bu,\bx)$.
Since any element of the group $S_{\bu,\bx}$ does not change the set 
$T_{{\bV}}(\bu,\bx)$, we have
\begin{align*}
&\sum_{(\bu'',\bx'') \in T_{{\bV}}(\bu,\bx)}
\sum_{g \in S_{\bu,\bx}}
\frac{1}{|S_{\bu,\bx}|}
P_{\Unif,\hat{\cal M}_{UX,n}}\circ g (\bu'',\bx'') \\
=&
\sum_{(\bu'',\bx'') \in T_{{\bV}}(\bu,\bx)}
P_{\Unif,\hat{\cal M}_{UX,n}} (\bu'',\bx'')
=
|T_{{\bV}}(\bu,\bx) \cap \hat{\cal M}_{UX,n}|
\cdot \frac{1}{\mathsf{M}_{U,n}\mathsf{M}_{X,n}}.
\end{align*}
Using this relation, 
we have
$
\sum_{g \in S_{\bu,\bx}}
\frac{1}{|S_{\bu,\bx}|}
P_{\Unif,\hat{\cal M}_{UX,n}}\circ g (\bu',\bx')
=
\frac{|T_{{\bV}}(\bu,\bx) \cap \hat{\cal M}_{UX,n}|}{|T_{{\bV}}(\bu,\bx)|}
\cdot
\frac{1}{\mathsf{M}_{UX,n}}$
because the probability 
$\sum_{g \in S_{\bu,\bx}}
\frac{1}{|S_{\bu,\bx}|}
P_{\Unif,\hat{\cal M}_{UX,n}}\circ g (\bu'',\bx'')$
does not depend on the element $(\bu'',\bx'')$ when
$(\bu'',\bx'') \in T_{{\bV}}(\bu,\bx)  \subset T_{P_{UX}}$.
Therefore,
\begin{align}
&\sum_{g \in S_{\bu,\bx}}
\frac{1}{|S_{\bu,\bx}|}
P_{\Unif,\hat{\cal M}_{UX,n}}\circ g (\bu',\bx')
=
\frac{|T_{{\bV}}(\bu,\bx) \cap \hat{\cal M}_{UX,n}|}{|T_{{\bV}}(\bu,\bx)|}
\cdot
\frac{1}{\mathsf{M}_{U,n}\mathsf{M}_{X,n}} \nonumber \\
\stackrel{(a)}{=}  &
\frac{|T_{{\bV}}(\bu,\bx) \cap (\hat{\cal M}_{UX,n}\setminus \{\bu,\bx\})|}
{|T_{{\bV}}(\bu,\bx)|\mathsf{M}_{U,n}\mathsf{M}_{X,n}} \nonumber \\
\stackrel{(b)}{\le} &
 \frac{e^{-n(H(P_{UX}) -R_U-R_X )}}{\mathsf{M}_{U,n}\mathsf{M}_{X,n}}
=
e^{-nH(P_{UX})} e^{2 n^{3/4} } 
={P_{UX}}^{n}(\bu,\bx) e^{2n^{3/4} },
  \Label{8K-pf}
\end{align}
where each step can be shown as follows.
Step $(a)$ holds because 
the conditional type ${\bV}$ is not identical.
Step $(b)$ follows from \eqref{Y3}.
Hence, we obtain \eqref{8K}.

Next, we show \eqref{8I}. Assume that $\bu=\bu(j)$.
We choose a conditional type $\bV \in V((\bu,\bx),{\cal X})$
such that $\bx \notin T_{{\bV}}(\bu,\bx)$ and 
$\bx' \in T_{{\bV}}(\bu,\bx)$.
Since any element of the group $S_{\bu,\bx}$ does not change the set 
$T_{{\bV}}(\bu,\bx)$, we have
\begin{align*}
&\sum_{\bx'' \in T_{{\bV}}(\bu,\bx)}
\sum_{g \in S_{\bu,\bx}}
\frac{1}{|S_{\bu,\bx}|}
P_{\Unif,\hat{\cal M}_{UX,n}}\circ g (\bu,\bx'') \\
=&
\sum_{\bx'' \in T_{{\bV}}(\bu,\bx)}
P_{\Unif,\hat{\cal M}_{UX,n}} (\bu,\bx'')
=
|T_{{\bV}}(\bu,\bx) \cap 
\hat{\cal M}_{X,n,j}  
|
\cdot \frac{1}{\mathsf{M}_{U,n}\mathsf{M}_{X,n}}.
\end{align*}
Using this relation,
we have
$
\sum_{g \in S_{\bu,\bx}}
\frac{1}{|S_{\bu,\bx}|}
P_{\Unif,\hat{\cal M}_{UX,n}}\circ g (\bu,\bx')
=
\frac{|T_{{\bV}}(\bu,\bx) \cap 
\hat{\cal M}_{X,n,j} 
|}{|T_{{\bV}}(\bu,\bx)|}
\cdot
\frac{1}{\mathsf{M}_{UX,n}}$
because the probability 
$\sum_{g \in S_{\bu,\bx}}
\frac{1}{|S_{\bu,\bx}|}
P_{\Unif,\hat{\cal M}_{UX,n}}\circ g (\bu,\bx'')$
does not depend on the element $\bx''$ when
$\bx'' \in T_{{\bV}}(\bu,\bx)  \subset T_{P_{X}}$.
Therefore,
\begin{align}
&\sum_{g \in S_{\bu,\bx}}
\frac{1}{|S_{\bu,\bx}|}
P_{\Unif,\hat{\cal M}_{UX,n}}\circ g (\bu,\bx')
=
\frac{|T_{{\bV}}(\bu,\bx) \cap 
\hat{\cal M}_{X,n,j} 
|}{|T_{{\bV}}(\bu,\bx)|}
\cdot
\frac{1}{\mathsf{M}_{U,n}\mathsf{M}_{X,n}} \nonumber \\
\stackrel{(a)}{=}  &
\frac{|T_{{\bV}}(\bu,\bx) \cap 
(\hat{\cal M}_{X,n,j}  \setminus \{\bx\}) 
|}
{|T_{{\bV}}(\bu,\bx)|\mathsf{M}_{U,n}\mathsf{M}_{X,n}} \nonumber \\
\stackrel{(b)}{\le} &
\frac{e^{-n(H(X|U)_{P_{UX}}-R_V)} }{\mathsf{M}_{U,n}\mathsf{M}_{X,n}}
=
\frac{e^{-n H(X|U)_{P_{UX}} } e^{ n^{3/4} } }{\mathsf{M}_{U,n}}
=
\frac{P_{X|U}^{n}(\bx|\bu) e^{n^{3/4} }}{\mathsf{M}_{U,n}}
,  \Label{8I-pf}
\end{align}
where each step can be shown as follows.
Step $(a)$ holds because of $\bx \notin T_{{\bV}}(\bu,\bx)$.
Step $(b)$ follows from \eqref{Y2}.
Hence, we obtain \eqref{8I}.
\end{proof}

\subsection{MAC code}
The aim of this subsection is an extension of the results in \cite[Section IV]{Ha1} to the setting for 
MAC codes.
Liu and Hughes \cite[Lemma 1]{Liu-Hughes}
extended the packing lemma by Csisz\'{a}r and K\"{o}rner\cite{CK}
to the case with two terminals ${\cal A}$ and ${\cal B}$.
In the same way as Proposition \ref{l1},
Lemma 1 of \cite{Liu-Hughes} can be rewritten as follows when $\hat{V}$ is the identical conditional type
and ${\cal U}$ is ${\cal T}$.

\begin{proposition}\Label{lB-2W}
For two positive numbers $R_A,R_B>0$, 
there exists a sufficiently large integer $N$ satisfying the following.
We chose an integer $n \ge N$, a joint type
${P}_{A-T-B} \in T_n({\cal A}\times {\cal B}\times {\cal T})$
 satisfying the Markov condition $A-T-B$,
 $R_A< H(A|T)_{{P}_{A-T-B}}$ and $R_B< H(B|T)_{{P}_{A-T-B}}$,
 and $\bu \in T_{P_T}^n({\cal T})$,
there exist 
$\mathsf{M}_{A,n}:=e^{nR_A-n^{3/4}}$ distinct elements in ${\cal A}^n$
and
$\mathsf{M}_{B,n}:=e^{n R_B-n^{3/4}}$ distinct elements in ${\cal B}^n$
as
\begin{align*}
\hat{\cal M}_{A,n} &:=\{ \ba(1),\ldots, \ba({\mathsf{M}_{A,n}})\} \subset T_{P _{A|T}}^n({\cal A}|\bt)\\
\hat{\cal M}_{B,n} &:=\{ \bb(1),\ldots, \bb({\mathsf{M}_{B,n}})\} \subset T_{P_{B|T}}^n({\cal B}|\bt)
\end{align*}
such that 
the inequalities
\begin{align}
& |T_{{\bV}}(\ba,\bb,\bt) \cap 
((\hat{\cal M}_{A,n}\setminus \{\ba\}) \times 
(\hat{\cal M}_{B,n}\setminus \{\bb\})
)| \nonumber \\
\le & |T_{{\bV}}(\ba,\bb,\bt)| 
e^{-n(H(AB|T)_{{P}_{A-T-B}}-R_A-R_B)} 
\Label{X3W} 
\\
& |T_{{\bV}_B}(\ba,\bb,\bt) \cap 
(\hat{\cal M}_{B,n}\setminus \{\bb\})
| \nonumber \\
\le & |T_{{\bV}_B}(\ba,\bb,\bt)| 
e^{-n(H(B|T)_{{P}_{A-T-B}}-R_B)} 
\Label{X2W} \\
& |T_{{\bV}_A}(\ba,\bb,\bt) \cap  (\hat{\cal M}_{A,n}\setminus \{\ba\}) | \nonumber \\
\le & |T_{{\bV}_A}(\ba,\bb,\bt)| 
e^{-n(H(A|T)_{{P}_{A-T-B}}-R_A)} 
\Label{X2WH} 
\end{align}
hold for any elements
$\ba \in \hat{\cal M}_{A,n} \subset T_{P _{A|T}}^n({\cal A}|\bt)$, 
$\bb \in \hat{\cal M}_{B,n} \subset T_{P _{B|T}}^n({\cal B}|\bt)$ 
and 
any conditional types ${\bV} \in V((\ba,\bb,\bt),{\cal A}\times {\cal B})$,
${\bV}_B \in V((\ba,\bb,\bt),{\cal B})$,
and ${\bV}_A \in V((\ba,\bb,\bt),{\cal A})$.
\hfill $\square$\end{proposition}

Our universal encoder classical-quantum MAC code
is given by the above construction, and has
a decoding error probability essentially equivalent to the average performance of random coding.
Proposition \ref{lB-2W} is a special case of \cite[Lemma 1]{Liu-Hughes}
by setting 
$\delta$, ${\cal U}$ ${\cal X}$, ${\cal Y}$, and ${\cal Z}$
to be $n^{-\frac{1}{4}}$, ${\cal T}$, ${\cal A}$, ${\cal B}$, and ${\cal A}\times {\cal B}$, 
respectively.
Eq. \eqref{X3W} of Proposition \ref{lB-2W} follows from Eqs. (11) of \cite[Lemma 1]{Liu-Hughes}
with substituting the identical conditional type into 
$\hat{V}$.
To consider Eq. \eqref{X2W} of Proposition \ref{lB-2W} we choose 
$a_o\in {\cal A}$, and denote the identical conditional type in $V((\ba,\bb,\bt),{\cal B})$
by $\hat{V}_B$.
Eq. \eqref{X2W} of Proposition \ref{lB-2W} follows from Eq. (10) of \cite[Lemma 1]{Liu-Hughes} 
by setting 
$V$ and $\hat{V}$ to be 
$V_B\times a_o$ and $\hat{V}_B\times a_o$, respectively.
Eq. \eqref{X2WH} of Proposition \ref{lB-2W} follows from Eq. (9) of \cite[Lemma 1]{Liu-Hughes}. 

As another generalization of Proposition \ref{PP7}, we have the following lemma for 
$\hat{\cal M}_{AB,n}:=\hat{\cal M}_{A,n}\times \hat{\cal M}_{B,n}$, which will be used for 
our evaluation of the decoding error probability of our universal c-q MAC code.

\begin{lemma}\Label{YP7W}
Assume that $\ba \in \hat{\cal M}_{A,n}$ and $\bb \in \hat{\cal M}_{B,n}$.
Any element $(\ba',\bb')\in T_{P_{AB|U}}^n({\cal A}\times {\cal B}|\bt)$ with
$\ba\neq \ba'$ and $\bb\neq \bb'$
 satisfies
\begin{align}
&\sum_{g \in S_{\ba,\bb}}
\frac{1}{|S_{\ba,\bb}|}
P_{\Unif,\hat{\cal M}_{AB,n}}\circ g (\ba',\bb')
\le P_{AB|T}^n (\ba',\bb'|\bt) e^{2 n^{3/4} }.
  \Label{8YW}
\end{align}
Any element $(\ba,\bb')\in T_{P_{AB|T}}^n({\cal A}\times {\cal B}|\bt)$ with
$\bb\neq \bb'$ satisfies
\begin{align}
&\sum_{g \in S_{\ba,\bb}}
\frac{\mathsf{M}_{A,n}}{|S_{\ba,\bb}|}
P_{\Unif,\hat{\cal M}_{AB,n}}\circ g (\ba,\bb')
\le 
P_{B|T}^n (\bb'|\bt) e^{n^{3/4} }
  \Label{8LW}.
\end{align}
Any element $(\ba',\bb)\in T_{P_{AB|T}}^n({\cal A}\times {\cal B}|\bt)$ with
$\ba\neq \ba'$ satisfies
\begin{align}
&\sum_{g \in S_{\ba,\bb}}
\frac{\mathsf{M}_{B,n}}{|S_{\ba,\bb}|}
P_{\Unif,\hat{\cal M}_{AB,n}}\circ g (\ba',\bb)
\le 
P_{A|T}^n (\ba'|\bt) e^{n^{3/4} }
  \Label{8LWH}.
\end{align}
\end{lemma}

\begin{proof}
We can show \eqref{8YW} in the same way as \eqref{8K} by replacing the role of \eqref{Y3}
by the role of \eqref{X3W}.
Also, we can show \eqref{8LW} in the same way as \eqref{8I} by replacing the role of \eqref{Y2}
by the role of \eqref{X2W}
because $P_{B|A}^n(\bb|\ba,\bt)=P_{B}^n(\bb|\bt)$.
Eq. \eqref{8LWH} follows from \eqref{X2WH} in the same way as \eqref{8LWH}.
\end{proof}

\section{Universal superposition coding}\Label{S6}
This section shows Theorem \ref{T2} by constructing our universal superposition code.

\subsection{Universal encoder}\Label{S6-1}
First, we construct our universal encoder by using Proposition \ref{lB-3}. 
We choose a sufficiently large integer $N$ to satisfy the conditions in Proposition \ref{lB-3}. 
Assume that $n \ge N$.
We choose $\mathsf{M}_{A,n}:=e^{n R_A-{n^{3/4}}}$ and 
$\mathsf{M}_{B,n}:=e^{n R_B-{n^{3/4}}}$.
Given a joint distribution $P_{UX}\in T_n({\cal U}\times {\cal X})$,
applying Proposition \ref{lB-3}, 
a map $\phi_{A,n}$ from $\{1, \ldots, \mathsf{M}_{A,n}\}$ to ${\cal U}^n$
and a map $\phi_{B,n}$ from $\{1, \ldots, \mathsf{M}_{A,n}\}\times \{1, \ldots, \mathsf{M}_{B,n}\}$ to 
${\cal X}^n$
to satisfy the condition in Proposition \ref{lB-3} with the join distribution $P_{UX}$.

\subsection{Universal decoder}
Next, we construct our universal decoders for both receivers. 
\subsubsection{Receiver $Y$}\Label{SS5-1}
Our decoder is constructed by using the idea given in \cite{Ha1}\cite[Chapter 7]{Group2}.
The quantum system ${\cal H}_Y^{\otimes n}$ is decomposed as
\begin{align}
{\cal H}_Y^{\otimes n}=
\oplus_{{\bf n}\in Y_n}{\cal W}_{{\bf n}},\Label{K11}\\
{\cal W}_{{\bf n}}:= {\cal U}_{{\bf n}} \otimes {\cal X}_{{\bf n}} .\Label{K10}
\end{align}
Define the state
\begin{align}
\rho_{{\Univ},n}:= \sum_{{\bf n}\in Y_n} \frac{1}{|Y_n|} \rho_{{\bf n}},
\end{align}
where $\rho_{{\bf n}}$ is the completely mixed state on ${\cal W}_{{\bf n}}$.
Then, we have \cite[Theorem 6.1]{Group2}
\begin{align}
\rho^{\otimes n} \le (n+1)^{\frac{(d_Y+2)(d_Y-1)}{2}} \rho_{{\Univ},n}
\Label{LN9-1}.
\end{align}

For simplicity, we consider the case
when $\bx'=(\underbrace{1,\ldots,1}_{m_1},\underbrace{2,\ldots,2}_{m_2},\ldots, 
\underbrace{d_X,\ldots, d_X}_{m_{d_X}})$.
In this case, we define $\rho_{\bx'}
:= \rho_{\Univ,m_1} \otimes \rho_{\Univ,m_2} \otimes \cdots \otimes \rho_{\Univ,m_{d_X}}$. 
For a general element $\bx\in {\cal X}^n$,
we define $\rho_{\bx}$ as the permutation of $\rho_{\bx'}$ 
with the above special element $\bx'$ satisfying $T_Y(\bx)=T_Y(\bx')$.
Hence, from \eqref{LN9-1}, we have \cite[(6.84)]{Group2}
\begin{align}
W^{(n)}_{\bx} \le (n+1)^{\frac{d_X(d_Y+2)(d_Y-1)}{2}} \rho_{\bx}
\Label{LN9}.
\end{align}
In the same way, for $\bu$, we define $\rho_{\bu}$.
As shown in \cite[(6.40)]{Group2}, 
the commutativity $[\rho_{\Univ,m_1} \otimes \rho_{\Univ,m_2} ,\rho_{\Univ,m_1+m_2} ]=0$ holds.
Hence, $\rho_{\phi_{B,n}(j,k)}$, $\rho_{\phi_{A,n}(j)}$, and $\rho_{{\Univ},n}$ are
commutative each other.

Using two positive numbers $r_A$ and $r_B$, we 
define the projections $\Pi_{j,k}^{(1)}, \Pi_{j,k}^{(2)}$, $\Pi_{j,k}^{(3)}$, and $\Pi_{j,k}$;
\begin{align}
\Pi_{j,k}^{(1)}:=& \{\rho_{\phi_{B,n}(j,k)} \ge C_{n}^{(1)}\rho_{\phi_{A,n}(j)} \}\\
\Pi_{j}^{(2)}:=&\{\rho_{\phi_{A,n}(j)} \ge C_{n}^{(2)} \rho_{{\Univ},n}\} \\
\Pi_{j,k}^{(3)}:=& \{\rho_{\phi_{B,n}(j,k)} \ge C_{n}^{(1)} C_{n}^{(2)}\rho_{{\Univ},n} \}
\ge \Pi_{j,k}:=\Pi_{j,k}^{(1)}\Pi_{j}^{(2)},
\end{align}
where $C_{n}^{(1)}:= e^{n (R_B+r_B)}$, $C_{n}^{(2)}:=e^{n (R_A+r_A)}$.
These projections are commutative each other
because $\rho_{\phi_{B,n}(j,k)}$, $\rho_{\phi_{A,n}(j)}$, and $\rho_{{\Univ},n}$ are
commutative each other.
Finally, we define the decoder of Receiver $Y$ as
\begin{align}
D(j,k):=&
\Big(\sum_{j',k'}\Pi(j',k')\Big)^{-1/2}
\Pi(j,k) \Big(\sum_{j',k'}\Pi(j',k')\Big)^{-1/2}.
\end{align}

\subsubsection{Receiver $Z$}
On the quantum system ${\cal H}_Z^{\otimes n}$,
we define $\rho^Z_{{\Univ},n}$ and $\rho^Z_{\bu}$ based on the same decomposition as \eqref{K11}
in the same was as Subsubsection \ref{SS5-1}.
We define the projection $\Pi_{j}^{Z}$ as
\begin{align}
\Pi_{j}^{Z}:=&\{\rho^Z_{\phi_{A,n}(j)} \ge C_{n}^{(2)} \rho^Z_{{\Univ},n}\} .
\end{align}
The decoder of Receiver $Z$ is given as
\begin{align}
D^Z(j):=&
\Big(\sum_{j}\Pi^Z(j')\Big)^{-1/2}
\Pi^Z(j)\Big(\sum_{j}\Pi^Z(j')\Big)^{-1/2}.
\end{align}

\subsection{Error evaluation}
Finally, we evaluate the decoding error probabilities for both receivers
by deriving our lower bounds of their exponential decreasing rates (exponents).
\subsubsection{Receiver $Y$}\Label{SSS8}
Before starting the evaluation of the decoding error probability of Receiver $Y$, we prepare several notations.
We simplify the average state $\sum_{x}P_X(x)W_x$ as $W_{P_X}$.
Then, similar to $W_{P_X}$, we define $W^{(n)}_{P_{X^n}}$ for any distribution $P_{X^n}$ on ${\cal X}^n$.
Also, the dimension of the quantum system ${\cal H}_Y$ is denoted by $d_Y$.
These notations are applied to the other system ${\cal H}_Z$.

We evaluate the decoding error probability of Receiver $Y$ as
\begin{align}
&\Tr W^{(n)}_{\phi_{B,n}(j,k)} (I-D(j,k))
\stackrel{(a)}{\le} 
 2 \Tr W^{(n)}_{\phi_{B,n}(j,k)} (I-\Pi(j,k))
+ 4 \Tr W^{(n)}_{\phi_{B,n}(j,k)}
\Big(\sum_{(j',k')\neq (j,k)}\Pi(j',k')\Big) \nonumber \\
\le & 2 \Tr W^{(n)}_{\phi_{B,n}(j,k)} (I-\Pi_{j,k}^{(1)})
+ 2 \Tr W^{(n)}_{\phi_{B,n}(j,k)} (I-\Pi_{j}^{(2)}) \nonumber \\
&+ 4 \Tr W^{(n)}_{\phi_{B,n}(j,k)}
\Big(\sum_{j' (\neq j),k'}\Pi(j',k')\Big) 
+ 4 \Tr W^{(n)}_{\phi_{B,n}(j,k)}
\Big(\sum_{k' \neq k}\Pi(j,k')\Big) \nonumber \\
\le& 2 \Tr W^{(n)}_{\phi_{B,n}(j,k)} (I-\Pi_{j,k}^{(1)})
+ 2 \Tr W^{(n)}_{\phi_{B,n}(j,k)} (I-\Pi_{j}^{(2)}) \nonumber \\
&+ 4 \Tr W^{(n)}_{\phi_{B,n}(j,k)}
\Big(\sum_{j' (\neq j),k'}\Pi^{(3)}(j',k')\Big) 
+ 4 \Tr W^{(n)}_{\phi_{B,n}(j,k)}
\Big(\sum_{k' \neq k}\Pi^{(1)}(j,k')\Big) ,\Label{ER1}
\end{align}
where Step $(a)$ follows from \cite[Lemma 2]{HN}.
Although the term $\Tr W^{(n)}_{\phi_{B,n}(j,k)} (I-\Pi_{j,k}^{(3)}) $ does not appear,
we evaluate it as the first step for the reparation for the evaluation of 
$\Tr W^{(n)}_{\phi_{B,n}(j,k)} (I-\Pi_{j,k}^{(1)})$.
For any $t \in (0,1)$, we have
\begin{align}
 &\Tr W^{(n)}_{\phi_{B,n}(j,k)} (I-\Pi_{j,k}^{(3)}) \nonumber \\
\stackrel{(a)}{\le} & (n+1)^{\frac{s d_X(d_Y+2)(d_Y-1)}{2}} (C_{n}^{(1)} C_{n}^{(2)})^s \Tr (W^{(n)}_{\phi_{B,n}(j,k)})^{1-s}  \rho_{{\Univ},n}^s \nonumber \\
\stackrel{(b)}{=} & (n+1)^{\frac{s d_X(d_Y+2)(d_Y-1)}{2}}(C_{n}^{(1)}C_{n}^{(2)})^s 
\frac{1}{|T_{P_{X}}|}\sum_{\bx \in T_{P_{X}}} 
\Tr (W^{(n)}_{\bx})^{1-s} \rho_{{\Univ},n}^s \nonumber \\
\stackrel{(c)}{\le} &(n+1)^{\frac{s d_X(d_Y+2)(d_Y-1)}{2}} (1+n)^{d_X} (C_{n}^{(1)}C_{n}^{(2)})^s 
\sum_{\bx \in {\cal X}^n} P_X^n(\bx) 
\Tr (W^{(n)}_{\bx})^{1-s} \rho_{{\Univ},n}^s \nonumber \\
 &(n+1)^{\frac{s d_X(d_Y+2)(d_Y-1)}{2}} (1+n)^{d_X} (C_{n}^{(1)}C_{n}^{(2)})^s 
\Tr \Big(\sum_{x\in {\cal X}} P_X(x) 
 W_x^{1-s} \Big)^{\otimes n} 
 \rho_{{\Univ},n}^s \nonumber \\
= &(n+1)^{\frac{s d_X(d_Y+2)(d_Y-1)}{2}} (1+n)^{d_X} (C_{n}^{(1)}C_{n}^{(2)})^s 
\max_{\sigma_n}
\Tr \Big(\sum_{x\in {\cal X}} P_X(x) 
 W_x^{1-s} \Big)^{\otimes n} 
 \sigma_n^s \nonumber \\
\stackrel{(d)}{\le} &(n+1)^{\frac{s d_X(d_Y+2)(d_Y-1)}{2}} (1+n)^{d_X} (C_{n}^{(1)}C_{n}^{(2)})^s 
\Big(\Tr
\Big( \Big(\sum_{x\in {\cal X}} P_X(x) 
 W_x^{1-s} \Big)^{\otimes n} \Big)^{\frac{1}{1-s}}\Big)^{1-s}
\nonumber  \\
\le &(n+1)^{\frac{s d_X(d_Y+2)(d_Y-1)}{2}+d_X} (C_{n}^{(1)}C_{n}^{(2)})^s 
\Big(\Tr
\Big( \sum_{x\in {\cal X}} P_X(x) 
 W_x^{1-s} \Big)^{\frac{1}{1-s}}\Big)^{n(1-s)} \nonumber \\
= &(n+1)^{\frac{s d_X(d_Y+2)(d_Y-1)}{2}+d_X} 
e^{n s(R_A+r_A+R_B+r_B- I_{1-s}(X;Y))}
 ,  \Label{KJ1}
\end{align}
where each step can be shown as follows.
Step $(a)$ follows from the combination of \eqref{LN9} and the condition in $\Pi_{j,k}^{(3)}$.
Step $(b)$ holds because 
$\Tr (W^{(n)}_{\bx})^{1-s} \rho_{{\Univ},n}^s $ has the same value for any $\bx \in T_{P_{X}}$ and 
$\phi_{B,n}(j,k) \in T_{P_{X}}$.
Step $(c)$ follows from \eqref{eq3-X}. 
Step $(d)$ follows from Eq.(6.92) of \cite{Group2} or Eq.(20) of \cite{Ha1}., i.e., the H\"{o}lder inequality.

Using any $s \in (0,1)$, 
$n_u:=n(\phi_{A,n(j)},u )$, and ${\cal N}_u:={\cal N}(\phi_{A,n(j)},u )$ for $u \in {\cal U}$,
we evaluate the first term of \eqref{ER1} as
\begin{align}
& \Tr W^{(n)}_{\phi_{B,n}(j,k)} (I-\Pi_{j,k}^{(1)}) \nonumber \\
\stackrel{(a)}{\le} & 
(n+1)^{\frac{s d_X(d_Y+2)(d_Y-1)}{2}}(C_{n}^{(1)})^s \Tr (W^{(n)}_{\phi_{B,n}(j,k)})^{1-s} 
(\rho_{\phi_{A,n}(j)})^s \nonumber \\
= & 
(n+1)^{\frac{s d_X(d_Y+2)(d_Y-1)}{2}}(C_{n}^{(1)})^s 
\prod_{u \in {\cal U}}
\Tr (W^{(n_u)}_{\phi_{B,n}|_{ {\cal N}_u}(j,k)})^{1-s} 
(\rho_{\Univ, n_u})^s \nonumber \\
\stackrel{(b)}{\le} &(n+1)^{\frac{s d_X(d_Y+2)(d_Y-1)}{2}+d_X d_U} (C_{n}^{(1)})^s 
\prod_{ u \in {\cal U}}
\Big(\Tr
\Big(\sum_{x\in {\cal X}}P_{X|U=u}(x)
W_x^{1-s} \Big)^{\frac{1}{1-s}}
\Big)^{n_u (1-s)} \nonumber \\
\stackrel{(c)}{\le} &(n+1)^{\frac{s d_X(d_Y+2)(d_Y-1)}{2}+d_X d_U} (C_{n}^{(1)})^s 
\Big(\sum_{u \in {\cal U}} P_U(u)
\Tr
\Big(\sum_{x\in {\cal X}}P_{X|U=u}(x)
W_x^{1-s} \Big)^{\frac{1}{1-s}}
\Big)^{n (1-s)}\nonumber \\
= &(n+1)^{\frac{s d_X(d_Y+2)(d_Y-1)}{2}+d_X d_U} 
e^{n s(R_B+r_B- I_{1-s}(X;Y|U))}
,\Label{KJ2}
\end{align}
where each step can be shown as follows.
Step $(a)$ follows from the combination of \eqref{LN9} and the condition in $\Pi_{j,k}^{(1)}$.
Step $(b)$ follows from the application of 
\eqref{KJ1} to the case with $n=n_u$ for each $u$.
Step $(c)$ follows from the following inequality;
Due to the concavity of $\log (x)$, any positive numbers $a_u$ and $n_u$ satisfy
\begin{align}
\prod_{u \in {\cal U}} a_u^{n_u}
\le (\sum_{u \in {\cal U}} \frac{n_u}{n}a_u)^n,
\end{align}
because $n=\sum_{u \in {\cal U}} n_u$.

Using any $s \in (0,1)$, we evaluate the second term of \eqref{ER1} as
\begin{align}
 &\Tr W^{(n)}_{\phi_{B,n}(j,k)} (I-\Pi_{j}^{(2)}) \nonumber \\
\stackrel{(a)}{=} 
&\sum_{g \in S_{\phi_{A,n}(j)}} \frac{1}{|S_{\phi_{A,n}(j)}|}\Tr W^{(n)}_{g(\phi_{B,n}(j,k))} (I-\Pi_{j}^{(2)}) \nonumber \\
= &\sum_{\bx} P_{\Unif,T_{P_{X|U}}}(\bx|\phi_{A,n}(j)) \Tr W^{(n)}_{\bx } (I-\Pi_{j}^{(2)}) \nonumber \\
\stackrel{(b)}{\le}  &|T_n({\cal X})| \sum_{\bx} P_{X|U}^n(\bx|\phi_{A,n}(j)) \Tr W^{(n)}_{\bx } (I-\Pi_{j}^{(2)}) \nonumber \\
= &|T_n({\cal X})| \Tr \big(\bigotimes_{i=1}^n W_{P_{X|U=\phi_{A,n}(j)_i} }\big) (I-\Pi_{j}^{(2)})\nonumber  \\
\stackrel{(c)}{\le}  &|T_n({\cal X})| 
(n+1)^{\frac{s d_U(d_Y+2)(d_Y-1)}{2}+d_U} (C_{n}^{(2)})^s 
\Big(\Tr
\Big( \sum_{u\in {\cal U}} P_U(u) 
 W_{P_{X|U=u}}^{1-s} \Big)^{\frac{1}{1-s}}\Big)^{n(1-s)}\nonumber \\
= &|T_n({\cal X})| 
(n+1)^{\frac{s d_U(d_Y+2)(d_Y-1)}{2}+d_U} 
e^{n s(R_A+r_A-s I_{1-s}(U;Y))}
,\Label{2nd-E}
\end{align}
where each step can be shown as follows.
Step $(a)$ holds because 
the projection $(I-\Pi_{j}^{(2)})$ is invariant with respect to $S_{\phi_{A,n}(j)}$.
Step $(b)$ follows from \eqref{eq3-X}.
Step $(c)$ can be shown in the same way as \eqref{KJ1}.

The third term of \eqref{ER1} is evaluated as
\begin{align}
&
\frac{1}{\mathsf{M}_{A,n} \mathsf{M}_{B,n}}
\sum_{j,k} \Tr W^{(n)}_{\phi_{B,n}(j,k)}
\Big(\sum_{j' (\neq j),k'}\Pi^{(3)}(j',k')\Big) 
=
\frac{1}{\mathsf{M}_{A,n} \mathsf{M}_{B,n}}
\sum_{j',k' } \Tr (\sum_{j' (\neq j),k'} W^{(n)}_{\phi_{B,n}(j,k)})
\Pi^{(3)}(j',k') \nonumber \\
=&
\sum_{j',k' } \Tr \Big(
\sum_{ \bu \neq \phi_{A,n}(j')} \sum_{\bx} P_{\Unif,\hat{\cal M}_{UX,n} }(\bu,\bx)W^{(n)}_{\bx} \Big)
\Pi^{(3)}(j',k') \nonumber \\
\stackrel{(a)}{=} &
\sum_{j',k' } \Tr \Big(
\sum_{g \in S_{\phi_{A,n}(j'),\phi_{B,n}(j',k') }} \frac{1}{|S_{\phi_{A,n}(j'),\phi_{B,n}(j',k') }|}
\sum_{ \bu \neq \phi_{A,n}(j')} \sum_{\bx} P_{\Unif,\hat{\cal M}_{UX,n} }\circ g (\bu,\bx)W^{(n)}_{\bx} \Big)
\Pi^{(3)}(j',k') \nonumber \\
\stackrel{(b)}{\le} & 
e^{2n^{3/4}}
\sum_{j',k' } \Tr \Big(
\sum_{ \bu \neq \phi_{A,n}(j')} \sum_{\bx} P^n(\bu,\bx)W^{(n)}_{\bx} \Big)
\Pi^{(3)}(j',k') \nonumber \\
\le & 
e^{2n^{3/4}}
\sum_{j',k' } \Tr \Big(
\sum_{ \bu ,\bx} P^n(\bu,\bx)W^{(n)}_{\bx} \Big)
\Pi^{(3)}(j',k') \nonumber \\
= & 
2e^{2n^{3/4}}
\sum_{j',k' } \Tr W_{P_V}^{\otimes n}
\Pi^{(3)}(j',k') \nonumber \\
\stackrel{(c)}{\le} & 
e^{2n^{3/4}}
\sum_{j',k' } \Tr 
 (n+1)^{\frac{(d_Y+2)(d_Y-1)}{2}}
\rho_{\Univ,n}
\Pi^{(3)}(j',k') \nonumber \\
\le & e^{2n^{3/4}}
 (n+1)^{\frac{(d_Y+2)(d_Y-1)}{2}}
\sum_{j',k' } (C_{n}^{(1)}C_{n}^{(2)})^{-1}
=
e^{2n^{3/4}} (n+1)^{\frac{(d_Y+2)(d_Y-1)}{2}}
{\mathsf{M}_{A,n} \mathsf{M}_{B,n}}
 (C_{n}^{(1)} C_{n}^{(2)})^{-1}\nonumber \\
=&
(n+1)^{\frac{(d_Y+2)(d_Y-1)}{2}}
e^{-n(r_A+r_B)},
\Label{3rd-E}
\end{align}
where each step can be shown as follows.
Step $(a)$ holds because 
the projection $\Pi^{(3)}(j',k')$ is invariant with respect to $S_{\phi_{A,n}(j'),\phi_{B,n}(j',k') }$.
Step $(b)$ follows from \eqref{8K}.
Step $(c)$ follows from the combination of \eqref{LN9-1} and the condition in $\Pi_{j',k'}^{(3)}$.

The fourth term of \eqref{ER1} is evaluated as
\begin{align}
&\frac{1}{\mathsf{M}_{A,n} \mathsf{M}_{B,n}}
\sum_{j,k}
\Tr W^{(n)}_{\phi_{B,n}(j,k)}
\Big(\sum_{k' \neq k}\Pi^{(1)}(j,k')\Big) \nonumber \\
=&
\frac{1}{\mathsf{M}_{A,n} \mathsf{M}_{B,n}}
\sum_{k',j}
\Tr 
(\sum_{k \neq k'} W^{(n)}_{\phi_{B,n}(j,k)})
\Pi^{(1)}(j,k') \nonumber \\
=&
\sum_{k',j}
\Tr \Big(
\sum_{\bx \neq \phi_{B,n}(k')} 
P_{\Unif,\hat{\cal M}_{UX,n} } (\phi_{A,n}(j')) ,\bx )
W^{(n)}_{\bx}\Big)
\Pi^{(1)}(j,k') \nonumber \\
\stackrel{(a)}{=} & 
\frac{1}{\mathsf{M}_{A,n} }
\sum_{j,k'}
\Tr \Big(
\sum_{g \in S_{\phi_{A,n}(j),\phi_{B,n}(j,k') }} \frac{\mathsf{M}_{A,n}}{|S_{\phi_{A,n}(j),\phi_{B,n}(j,k')}|}
\sum_{\bx \neq \phi_{B,n}(k')} 
P_{\Unif,\hat{\cal M}_{UX,n} } \circ g (\phi_{A,n}(j)) ,\bx )
W^{(n)}_{\bx}\Big)
\Pi^{(1)}(j,k') \nonumber \\
\stackrel{(b)}{\le} &
\frac{e^{n^{3/4}}}{\mathsf{M}_{A,n} }
\sum_{k',j}
\Tr \Big(
\sum_{\bx \neq \phi_{B,n}(k')} 
P_{X|U}^n (\bx|\phi_{A,n}(j)))
W^{(n)}_{\bx}\Big)
\Pi^{(1)}(j,k') \nonumber \\
\le &
\frac{e^{n^{3/4}}}{\mathsf{M}_{A,n} }
\sum_{k',j}
\Tr \Big(
\sum_{\bx}
P_{X|U}^n (\bx|\phi_{A,n}(j)))
W^{(n)}_{\bx}\Big)
\Pi^{(1)}(j,k') \nonumber \\
= &
\frac{e^{n^{3/4}}}{\mathsf{M}_{A,n} }
\sum_{k',j}
\Tr \Big(\bigotimes_{i=1}^n
W_{P_{X|U=\phi_{A,n}(j)_i}} \Big)
\Pi^{(1)}(j,k') \nonumber \\
\stackrel{(c)}{\le} &
\frac{e^{n^{3/4}}}{\mathsf{M}_{A,n} }
\sum_{k',j}
\Tr \Big(
 (n+1)^{\frac{d_U(d_Y+2)(d_Y-1)}{2}}
\rho_{\phi_{A,n}(j) ,n}
\Big)
\Pi^{(1)}(j,k') \nonumber \\
\stackrel{(d)}{\le} &
\frac{e^{n^{3/4}}}{\mathsf{M}_{A,n} }
\sum_{k',j}
 (n+1)^{\frac{d_U(d_Y+2)(d_Y-1)}{2}}
(C_n^{(1)})^{-1}
=e^{n^{3/4}}
 (n+1)^{\frac{d_U(d_Y+2)(d_Y-1)}{2}}
\mathsf{M}_{B,n} 
(C_n^{(1)})^{-1} \nonumber \\
=& e^{n^{3/4}}
 (n+1)^{\frac{d_U(d_Y+2)(d_Y-1)}{2}} e^{-n r_B},
\Label{4th-E}
\end{align}
where each step can be shown as follows.
Step $(a)$ holds because 
the projection $\Pi^{(1)}(j,k')$ is invariant with respect to $S_{\phi_{A,n}(j),\phi_{B,n}(j,k') }$.
Step $(b)$ follows from \eqref{8I}.
Step $(c)$ follows from \eqref{LN9} by replacing $d_X$ by $d_U$.
Step $(d)$ follows from the condition in $\Pi_{j,k'}^{(1)}$.

Hence, since we can choose $t$ freely in \eqref{KJ2} and \eqref{2nd-E}, 
the combination of \eqref{ER1}, \eqref{KJ2}, \eqref{2nd-E}, \eqref{3rd-E}, and \eqref{4th-E}
shows the following lower bound of the exponent of the decoding error probability of Receiver $Y$;
\begin{align}
&
\min \Big(
\max_s s(I_{1-s}(U;Y)-R_A-r_A),
\max_s s(I_{1-s}(X;Y|U)-R_B-r_B),r_A,r_A+r_B\Big) \nonumber \\
\ge &
\min \Big(
\max_s s(I_{1-s}(U;Y)-R_A-r_A),
\max_s s(I_{1-s}(X;Y|U)-R_B-r_B),r_A,r_B\Big) \nonumber \\
=&
\min \Big(
\min (\max_s s(I_{1-s}(U;Y)-R_A-r_A),r_A),
\min (\max_s s(I_{1-s}(X;Y|U)-R_B-r_B),r_B)
\Big) .
\end{align}
Thus, we obtain Eq. \eqref{Ex1} as a lower bound of the exponent of Receiver $Y$.

\subsubsection{Receiver $Z$}\Label{SSS9}
We evaluate the decoding error probability of Receiver $Z$ as follows.
\begin{align}
& \frac{1}{\mathsf{M}_{B,n} }
\sum_k \Tr W^{(n)}_{Z,\phi_{B,n}(j,k)} (I-D^Z(j)) \nonumber \\
\le& \frac{2}{\mathsf{M}_{B,n} }
\sum_k\Tr W^{(n)}_{Z,\phi_{B,n}(j,k)} (I-\Pi^Z(j))
+ \frac{4}{\mathsf{M}_{B,n} }
\sum_k\Tr W^{(n)}_{Z,\phi_{B,n}(j,k)}
\Big(\sum_{j'\neq j}\Pi^Z(j')\Big) .
\Label{ER2}
\end{align}
In the same way as \eqref{2nd-E},
using any $s \in (0,1)$, we evaluate the first term of \eqref{ER2} as
\begin{align}
 &\Tr W^{(n)}_{Z,\phi_{B,n}(j,k)} (I-\Pi_{j}^{(2)}) \nonumber \\
\le &|T_n({\cal X})| 
(n+1)^{\frac{s d_U(d_Y+2)(d_Y-1)}{2}+d_U} (C_{n}^{(2)})^s 
\Big(\Tr
\Big( \sum_{u\in {\cal U}} P_U(u) 
 W_{Z,P_{X|U=u}}^{1-s} \Big)^{\frac{1}{1-s}}\Big)^{n(1-s)}\nonumber \\
=&  |T_n({\cal X})| 
(n+1)^{\frac{s d_U(d_Y+2)(d_Y-1)}{2}+d_U} e^{n s(R_A+r_A-s I_{1-s}(U;Z))}.
\Label{KJ2O}
\end{align}

The second term  of \eqref{ER2} is evaluated as
\begin{align}
 &
\sum_{j,k} \Tr  W^{(n)}_{Z,\phi_{B,n}(j,k)}
\Big(\sum_{j'\neq j}\Pi^Z(j')\Big)\nonumber\\
 =&\sum_{j'} \Tr \Big(\sum_{j\neq j'}\sum_k W^{(n)}_{Z,\phi_{B,n}(j,k)}\Big) \Pi^Z(j')\nonumber\\
\stackrel{(a)}{=} &
\sum_{j'} \Tr 
  \sum_{g \in S_{\phi_{A,n}(j'),\phi_{B,n}(j',1))}} \frac{1}{|S_{\phi_{A,n}(j'),\phi_{B,n}(j',1))}|}
  \Big(\sum_{j\neq j'}\sum_k
   W^{(n)}_{Z,g^{-1} (\phi_{B,n}(j,k))}\Big) \Pi^Z(j')\nonumber \\
=&\sum_{j'} \Tr 
  \sum_{g \in S_{\phi_{A,n}(j'),\phi_{B,n}(j',1))}} \frac{1}{|S_{\phi_{A,n}(j'),\phi_{B,n}(j',1))}|}
  \Big(
  \sum_{\bx,\bu\neq \phi_{A,n}(j')}
P_{\Unif,\hat{\cal M}_{UX,n}}\circ g (\bx,\bu) 
   W^{(n)}_{Z,\bx}\Big) \Pi^Z(j')\nonumber \\
\stackrel{(b)}{\le} &\sum_{j'} e^{2 n^{3/4}} \Tr 
  \Big(
  \sum_{\bx,\bu\neq \phi_{A,n}(j')}
P_{UX}^n (\bx,\bu) 
   W^{(n)}_{Z,\bx}\Big) \Pi^Z(j')\nonumber \\
\le &\sum_{j'} e^{2 n^{3/4}} \Tr 
  \Big(
  \sum_{\bx,\bu}
P_{UX}^n (\bx,\bu) 
   W^{(n)}_{Z,\bx}\Big) \Pi^Z(j')\nonumber \\
= &\sum_{j'} 2e^{n^{3/4}} \Tr 
  W_{Z,P_X}^{\otimes n} \Pi^Z(j')\nonumber \\
\stackrel{(c)}{\le} & e^{2n^{3/4}} (n+1)^{\frac{(d_Z+2)(d_Z-1)}{2}}
\mathsf{M}_{A,n} 
 (C_{n}^{(2)})^{-1} \nonumber \\
= & e^{n^{3/4}} (n+1)^{\frac{(d_Z+2)(d_Z-1)}{2}}
e^{-r_A},
\Label{2nd-U}
\end{align}
where each step can be shown as follows.
Step $(a)$ holds because 
the projection $\Pi^{Z}(j')$ is invariant with respect to $S_{\phi_{A,n}(j')}$
and
$S_{\phi_{A,n}(j'),\phi_{B,n}(j',1))}$ is a subgroup of $S_{\phi_{A,n}(j')}$.
Step $(b)$ follows from \eqref{8K} of Lemma \ref{XP7}.
Step $(c)$ follows from the combination of \eqref{LN9-1} and the condition in $\Pi^Z(j')$.

Hence, since we can choose $t$ freely in \eqref{KJ2O},
from the combination of \eqref{ER2}, \eqref{KJ2O}, and \eqref{2nd-U}, 
we obtain Eq. \eqref{Ex2}, i.e., a lower bound of the exponent of the decoding error probability of Receiver $Y$ as follows.
\begin{align}
\min (\max_s s(I_{1-s}(U;Z)-R_A-r_A),r_A).
\end{align}

\section{Universal classical-quantum MAC coding with joint decoding}\Label{S7}
In this section, using our technique for universal classical-quantum superposition coding,
we construct universal classical-quantum MAC code with joint decoding that achieves the corner points. 
In this section, we omit the subscript ${P_{A-T-B}}$ in the (conditional) 
mutual information.

\subsection{Universal encoder}\Label{S7-1}
To construct our universal encoder, 
we apply Proposition \ref{lB-2W}
to a joint distribution $P_{A-T-B}\in T_n({\cal A}\times {\cal B}\times {\cal T})$
when $\bt$ has the form 
\begin{align}
(\underbrace{1,\ldots,1}_{m_1},\underbrace{2,\ldots,2}_{m_2},\ldots, 
\underbrace{d_T,\ldots, d_T}_{m_{d_T}}).
\Label{Defmt}
\end{align}
We choose a sufficiently large $N$ to satisfy the conditions in Proposition \ref{lB-2W}. 
Assume that $n \ge N$.
We choose $\mathsf{M}_{A,n}:=e^{n R_A-{n^{3/4}}}$ and 
$\mathsf{M}_{B,n}:=e^{n R_B-{n^{3/4}}}$.
Applying Proposition \ref{lB-2W}, 
a map $\psi_{A,n}$ from $\{1, \ldots, \mathsf{M}_{A,n}\}$ to ${\cal A}^n$
and a map $\psi_{B,n}$ from $\{1, \ldots, \mathsf{M}_{B,n}\}$ to ${\cal B}^n$
such that the subsets
$\hat{\cal M}_{A,n}:=\{\psi_{A,n}(1), \ldots, \psi_{A,n}(\mathsf{M}_{A,n})\}$ 
and
$\hat{\cal M}_{B,n}:=\{\psi_{B,n}(1), \ldots, \psi_{B,n}(\mathsf{M}_{B,n})\}$ 
satisfy the condition in Proposition \ref{lB-2W} with the two distributions $P_{A}$ and $P_{B}$.
To describe the components,
we use the notation 
$\psi_{A,n}(j)=(\psi_{A,n,1}(j), \ldots, \psi_{A,n,d_T}(j))\in {\cal A}^{m_1}\times \cdots \times {\cal A}^{m_{d_T}}$
and
$\psi_{B,n}(k)=(\psi_{B,n,1}(k), \ldots, \psi_{B,n,d_T}(k))\in {\cal B}^{m_1}\times \cdots \times {\cal B}^{m_{d_T}}$.


\subsection{Universal decoder}
Our decoder is the same as the decoder of Receiver $Y$ in the case of superposition coding.
Using two positive numbers $r_A$ and $r_B$, we 
define the projections $\overline{\Pi}_{j,k}^{(1)}, \overline{\Pi}_{j,k}^{(2)}$, $\overline{\Pi}_{j,k}^{(3)}$, and $\overline{\Pi}_{j,k}$;
\begin{align}
\overline{\Pi}_{j,k}^{(1)}:=& \Big\{
\Big(\bigotimes_{t \in {\cal T}}
\rho_{\psi_{A,n,t}(j),\psi_{B,n,t}(k)} \Big)\ge C_{n}^{(1)}
\Big(\bigotimes_{t \in {\cal T}}
\rho_{\psi_{A,n,t}(j)}\Big) \Big\} \\
\overline{\Pi}_{j}^{(2)}:=&\Big\{\Big(
\bigotimes_{t \in {\cal T}}\rho_{\psi_{A,n,t}(j)}\Big) \ge C_{n}^{(2)} \Big(\bigotimes_{t \in {\cal T}}\rho_{{\Univ},m_t}\Big)\Big\}  \\
\overline{\Pi}_{j,k}^{(3)}:=& \Big\{\Big(
\bigotimes_{t \in {\cal T}}\rho_{\psi_{A,n,t}(j),\psi_{B,n,t}(k)} \Big)\ge C_{n}^{(1)} C_{n}^{(2)}
\Big(\bigotimes_{t \in {\cal T}}\rho_{{\Univ},m_t} \Big)\Big\}
\ge \overline{\Pi}_{j,k}:=\overline{\Pi}_{j,k}^{(1)}\overline{\Pi}_{j}^{(2)},
\end{align}
where $C_{n}^{(1)}:= e^{n (R_B+r_B)}$, $C_{n}^{(2)}:=e^{n (R_A+r_A)}$.
These projections are commutative with each other
because $\bigotimes_{t \in {\cal T}}\rho_{\psi_{A,n,t}(j),\psi_{B,n,t}(k)}$, 
$\bigotimes_{t \in {\cal T}}\rho_{\psi_{A,n,t}(j)}$, and 
$\bigotimes_{t \in {\cal T}}\rho_{{\Univ},m_t}$ are
commutative with each other.

Then, the decoder is given as
\begin{align}
\overline{D}(j,k):=&
\Big(\sum_{j',k'}\overline{\Pi}(j',k')\Big)^{-1/2}
\overline{\Pi}(j,k) \Big(\sum_{j',k'}\overline{\Pi}(j',k')\Big)^{-1/2}.
\end{align}

\subsection{Error evaluation}
The decoding error probability of our code is evaluated in a quite similar way to Subsubsection \ref{SSS8}.
Our decoding error probability is decomposed as 
\begin{align}
&\Tr W^{(n)}_{\psi_{A,n}(j),\psi_{B,n}(k)} (I-\overline{D}(j,k)) \nonumber \\
\stackrel{(a)}{\le}  & 2 \Tr W^{(n)}_{\psi_{A,n}(j),\psi_{B,n}(k)} (I-\overline{\Pi}(j,k))
+ 4 \Tr W^{(n)}_{\psi_{A,n}(j),\psi_{B,n}(k)}
\Big(\sum_{(j',k')\neq (j,k)}\overline{\Pi}(j',k')\Big) \nonumber \\
\le & 2 \Tr W^{(n)}_{\psi_{A,n}(j),\psi_{B,n}(k)} (I-\overline{\Pi}_{j,k}^{(1)})
+ 2 \Tr W^{(n)}_{\psi_{A,n}(j),\psi_{B,n}(k)} (I-\overline{\Pi}_{j}^{(2)}) \nonumber \\
&+ 4 \Tr W^{(n)}_{\psi_{A,n}(j),\psi_{B,n}(k)}
\Big(\sum_{j' (\neq j),k'(\neq k)}\overline{\Pi}(j',k')\Big) 
+ 4 \Tr W^{(n)}_{\psi_{A,n}(j),\psi_{B,n}(k)}
\Big(\sum_{k' \neq k}\overline{\Pi}(j,k')\Big) \nonumber \\
&+ 4 \Tr W^{(n)}_{\psi_{A,n}(j),\psi_{B,n}(k)}
\Big(\sum_{j' \neq j}\overline{\Pi}(k,j')\Big) \nonumber \\
\le& 2 \Tr W^{(n)}_{\psi_{A,n}(j),\psi_{B,n}(k)} (I-\overline{\Pi}_{j,k}^{(1)})
+ 2 \Tr W^{(n)}_{\psi_{A,n}(j),\psi_{B,n}(k)} (I-\overline{\Pi}_{j}^{(2)}) \nonumber \\
&+ 4 \Tr W^{(n)}_{\psi_{A,n}(j),\psi_{B,n}(k)}
\Big(\sum_{j' (\neq j),k'(\neq k)}\overline{\Pi}^{(3)}(j',k')\Big) 
+ 4 \Tr W^{(n)}_{\psi_{A,n}(j),\psi_{B,n}(k)}
\Big(\sum_{k' \neq k}\overline{\Pi}^{(1)}(j,k')\Big) \nonumber \\
&+ 4 \Tr W^{(n)}_{\psi_{A,n}(j),\psi_{B,n}(k)}
\Big(\sum_{j' \neq j}\overline{\Pi}^{(2)}(j')\Big) ,
\Label{ER3}
\end{align}
where Step $(a)$ follows from \cite[Lemma 2]{HN}.

All the terms in \eqref{ER3} except for the fifth term $\Tr W^{(n)}_{\psi_{A,n}(j),\psi_{B,n}(k)}
\Big(\sum_{j' \neq j}\overline{\Pi}^{(2)}(k,j')\Big)$
can be evaluated in the same way as Subsubsection \ref{SSS8}
by replacing the roles of Proposition \ref{lB-3} and Lemma \ref{XP7} 
by the roles of Proposition \ref{lB-2W} and Lemma \ref{YP7W}.
In this derivation,
Eqs. \eqref{8YW} and \eqref{8LW} in Lemma \ref{YP7W}
take the roles of 
Eqs. \eqref{8K} and \eqref{8I} in Lemma \ref{XP7} as follows.

For simplicity, we evaluate the first term of \eqref{ER3} when ${\cal T}$ is singleton.
Using any $s \in (0,1)$, 
$n_a:=n(\phi_{A,n(j)},a )$, and ${\cal N}_a:={\cal N}(\psi_{A,n(j)},a)$ for $a \in {\cal A}$,
\begin{align}
&  \Tr W^{(n)}_{\psi_{A,n}(j),\psi_{B,n}(k)} (I-\overline{\Pi}_{j,k}^{(1)}) \nonumber \\
\stackrel{(a)}{\le} & 
(n+1)^{\frac{s d_B(d_Y+2)(d_Y-1)}{2}}(C_{n}^{(1)})^s 
\Tr (W^{(n)}_{\psi_{A,n}(j),\psi_{B,n}(k)})^{1-s}  (\rho_{\phi_{A,n}(j)})^s \nonumber \\
= & 
(n+1)^{\frac{s d_B(d_Y+2)(d_Y-1)}{2}}(C_{n}^{(1)})^s 
\prod_{a \in {\cal A}}
\Tr (W^{(n_a)}_{\phi_{B,n}|_{ {\cal N}_a}(k)})^{1-s} 
(\rho_{\Univ, n_a})^s \nonumber \\
\stackrel{(b)}{\le} & 
(n+1)^{\frac{s d_B(d_Y+2)(d_Y-1)}{2}+d_B d_A} 
e^{n s(R_B+r_B- I_{1-s}(B;Y|A))}
,\Label{KJ2B}
\end{align}
where each step is shown as follows.
Step $(a)$ follows from the combination of \eqref{LN9} and the condition in $\Pi_{j,k}^{(1)}$.
Step $(b)$ follows in the same way as \eqref{KJ2}.
When ${\cal T}$ is not singleton,
applying \eqref{KJ2B} to the $d_T$ blocks, we evaluate the first term of \eqref{ER3}; 
\begin{align}
&  \Tr W^{(n)}_{\psi_{A,n}(j),\psi_{B,n}(k)} (I-\overline{\Pi}_{j,k}^{(1)}) \nonumber \\
\stackrel{(a)}{\le} & 
\prod_{t \in {\cal T}}
(m_t+1)^{\frac{s d_B(d_Y+2)(d_Y-1)}{2}+d_B d_A} 
e^{m_t s(R_B+r_B- I_{1-s}(B;Y|A,T=t))} \nonumber \\
\le&
(n+1)^{d_T(\frac{s d_B(d_Y+2)(d_Y-1)}{2}+d_B d_A)} 
e^{n s(R_B+r_B- I_{1-s}(B;Y|A,T))} 
,\Label{KJ2C}
\end{align}
where 
Step $(a)$ follows by applying \eqref{KJ2B} to $d_T$ blocks.

Applying the same modification to \eqref{2nd-E}, 
we evaluate the second term of \eqref{ER3} as
\begin{align}
& \Tr W^{(n)}_{\psi_{A,n}(j),\psi_{B,n}(k)} (I-\overline{\Pi}_{j}^{(2)}) \nonumber \\
\le &|T_n({\cal X})| 
(n+1)^{d_T (\frac{s d_A(d_Y+2)(d_Y-1)}{2}+d_A)} 
e^{n s(R_A+r_A-s I_{1-s}(A;Y|T))},
\Label{2nd-E2}
\end{align}

Modifying \eqref{3rd-E} and using the notation $\hat{\cal M}_{AB,n}:=\hat{\cal M}_{A,n}\times \hat{\cal M}_{B,n}$,
we evaluate the third term of \eqref{ER1} as
\begin{align}
&
\frac{1}{\mathsf{M}_{A,n} \mathsf{M}_{B,n}}
\sum_{j,k} \Tr W^{(n)}_{\psi_{A,n}(j),\psi_{B,n}(k)}
\Big(\sum_{j' (\neq j),k'(\neq k)}\Pi^{(3)}(j',k')\Big) \nonumber \\
=&
\frac{1}{\mathsf{M}_{A,n} \mathsf{M}_{B,n}}
\sum_{j',k' } \Tr (\sum_{j' (\neq j),k'(\neq k)} W^{(n)}_{\psi_{A,n}(j),\psi_{B,n}(k)})
\Pi^{(3)}(j',k') \nonumber \\
=&
\sum_{j',k' } \Tr \Big(
\sum_{ \ba \neq \psi_{A,n}(j')} 
\sum_{ \bb \neq \psi_{B,n}(k')} 
P_{\Unif,\hat{\cal M}_{AB,n} }(\ba,\bb)W^{(n)}_{\ba,\bb}
 \Big)
\Pi^{(3)}(j',k') \nonumber \\
\stackrel{(a)}{=} &
\sum_{j',k' } \Tr \Big(
\sum_{g \in S_{\psi_{A,n}(j),\psi_{B,n}(k) }} \frac{1}{|S_{\psi_{A,n}(j),\psi_{B,n}(k) }|}
\sum_{ \ba \neq \psi_{A,n}(j')} 
\sum_{ \bb \neq \psi_{B,n}(k')} 
P_{\Unif,\hat{\cal M}_{AB,n} }\circ g(\ba,\bb)W^{(n)}_{\ba,\bb}
\Big)
\Pi^{(3)}(j',k') \nonumber \\
\stackrel{(b)}{\le} & 
e^{2n^{3/4}}
\sum_{j',k' } \Tr \Big(
\sum_{ \ba \neq \psi_{A,n}(j')} 
\sum_{ \bb \neq \psi_{B,n}(k')} 
P^n(\ba,\bb|\bt)W^{(n)}_{\ba,\bb} \Big)
\Pi^{(3)}(j',k') \nonumber \\
\le & 
e^{2n^{3/4}}
\sum_{j',k' } \Tr \Big(
\sum_{ \ba ,\bb} P^n(\ba,\bb|\bt)W^{(n)}_{\ba,\bb} \Big)
\Pi^{(3)}(j',k') \nonumber \\
= & 
2e^{2n^{3/4}}
\sum_{j',k' } \Tr 
\Big(\bigotimes_{t \in {\cal T}}
W_{P_{AB|T=t}}^{\otimes m_t}\Big)
\Pi^{(3)}(j',k') \nonumber \\
\stackrel{(c)}{\le} & 
e^{2n^{3/4}}
\sum_{j',k' } \Tr 
 (n+1)^{\frac{d_T(d_Y+2)(d_Y-1)}{2}}
\big(\bigotimes_{t \in {\cal T}}\rho_{\Univ, m_t}\big)
\Pi^{(3)}(j',k') \nonumber \\
\le & e^{2n^{3/4}}
 (n+1)^{\frac{d_T(d_Y+2)(d_Y-1)}{2}}
\sum_{j',k' } (C_{n}^{(1)}C_{n}^{(2)})^{-1}
=
e^{2n^{3/4}} (n+1)^{\frac{d_T(d_Y+2)(d_Y-1)}{2}}
{\mathsf{M}_{A,n} \mathsf{M}_{B,n}}
 (C_{n}^{(1)} C_{n}^{(2)})^{-1}\nonumber \\
=&
(n+1)^{\frac{d_T(d_Y+2)(d_Y-1)}{2}}
e^{-n(r_A+r_B)},
\Label{3rd-E2}
\end{align}
where each step can be shown as follows.
Step $(a)$ holds because 
the projection $\Pi^{(3)}(j',k')$ is invariant with respect to $S_{\psi_{A,n}(j),\psi_{B,n}(k) }$.
Step $(b)$ follows from \eqref{8YW}.
Step $(c)$ follows from \eqref{LN9}.
Step $(d)$ follows from the condition in $\Pi_{j',k'}^{(3)}$.
Applying the same modification to \eqref{4th-E}, 
we evaluate the fourth term of \eqref{ER3} as
\begin{align}
&\frac{1}{\mathsf{M}_{A,n} \mathsf{M}_{B,n}}
\sum_{j,k}
 \Tr W^{(n)}_{\psi_{A,n}(j),\psi_{B,n}(k)}
\Big(\sum_{k' \neq k}\overline{\Pi}^{(1)}(j,k')\Big) \nonumber \\
\le& e^{n^{3/4}}
 (n+1)^{\frac{d_Td_A(d_Y+2)(d_Y-1)}{2}} e^{-n r_B},
\Label{4th-E2}
\end{align}
where 
we use \eqref{8LW} instead of \eqref{8I}.

We evaluate the fifth term of \eqref{ER3} as follows.
\begin{align}
 &
\sum_{j,k } \Tr  \Big[ W^{(n)}_{\psi_{A,n}(j),\psi_{B,n}(k)}
\Big(\sum_{j'\neq j}\overline{\Pi}^{(2)}(j')\Big)\Big]\nonumber \\
= &
\sum_{j'}\sum_{k} 
\Tr \Big[
\Big(\sum_{j\neq j'} W^{(n)}_{\psi_{A,n}(j),\psi_{B,n}(k)}\Big)
\overline{\Pi}^{(2)}(j')\Big]\nonumber \\
=&
\sum_{j'}\sum_{k} 
\Tr \Big[
  \Big(
  \sum_{\ba\neq \psi_{A,n}(j')}
P_{\Unif,\hat{\cal M}_{AB,n}} (\ba,\psi_{B,n}(k)) 
   W^{(n)}_{\ba,\psi_{B,n}(k)}\Big) \overline{\Pi}^{(2)}(j')\Big]\nonumber \\
=&\sum_{j',k'}  \frac{1}{\mathsf{M}_{B,n} -1} \sum_{k \neq k'} 
\Tr \Big[
  \Big(
  \sum_{\ba\neq \psi_{A,n}(j')}
P_{\Unif,\hat{\cal M}_{AB,n}} (\ba,\psi_{B,n}(k)) 
   W^{(n)}_{\ba,\psi_{B,n}(k)}\Big) \overline{\Pi}^{(2)}(j')\Big]\nonumber \\
=&\sum_{j',k'}  \frac{1}{\mathsf{M}_{B,n} -1} 
\Tr \Big[
  \Big(
  \sum_{\ba\neq \psi_{A,n}(j'),\bb\neq \psi_{B,n}(k')}
P_{\Unif,\hat{\cal M}_{AB,n}} (\ba,\bb) 
   W^{(n)}_{\ba,\bb}\Big) \overline{\Pi}^{(2)}(j')\Big]\nonumber \\
\stackrel{(a)}{=}&\sum_{j',k'}  \frac{1}{\mathsf{M}_{B,n} -1} \sum_{k \neq k'} 
\Tr \Bigg[
  \sum_{g \in S_{(\psi_{A,n}(j'),\psi_{B,n}(k'))}} \frac{1}{|S_{(\psi_{A,n}(j'),\psi_{B,n}(k'))}|}
\nonumber \\
  &\hspace{10ex} \cdot \Big(
  \sum_{\ba\neq \psi_{A,n}(j'),\bb\neq \psi_{B,n}(k')}
P_{\Unif,\hat{\cal M}_{AB,n}} \circ g (\ba,\bb) 
   W^{(n)}_{\ba,\bb}\Big) \overline{\Pi}^{(2)}(j')
   \Bigg]\nonumber \\
\stackrel{(b)}{\le} &
\sum_{j',k'} \frac{e^{2 n^{3/4}} }{\mathsf{M}_{B,n} -1}
\Tr \Big[
  \Big(
  \sum_{\ba\neq \psi_{A,n}(j'),\bb\neq \psi_{B,n}(k')}
P_{AB}^n (\ba,\bb|\bt) 
   W^{(n)}_{\ba,\bb}\Big) \overline{\Pi}^{(2)}(j')\Big]\nonumber \\
{\le} &
\sum_{j',k'} \frac{e^{2 n^{3/4}} }{\mathsf{M}_{B,n} -1}
\Tr \Big[  \Big(  \sum_{\ba,\bb}
P_{AB}^n (\ba,\bb|\bt) 
   W^{(n)}_{\ba,\bb}\Big) \overline{\Pi}^{(2)}(j')\Big]\nonumber \\
= &
\sum_{j',k'} \frac{e^{2 n^{3/4}} }{\mathsf{M}_{B,n} -1}
\Tr \Big[
\Big(\bigotimes_{t \in {\cal T}}W_{P_{AB|T=t}}^{\otimes m_t}\Big)\overline{\Pi}^{(2)}(j')\Big]\nonumber \\
\stackrel{(c)}{\le} &
\sum_{j',k'} \frac{e^{2 n^{3/4}} }{\mathsf{M}_{B,n} -1}
\Tr \Big[
 (n+1)^{\frac{d_T(d_Y+2)(d_Y-1)}{2}}
\big(\bigotimes_{t \in {\cal T}}\rho_{\Univ, m_t}\big)
\overline{\Pi}^{(2)}(j')\Big]\nonumber \\
\stackrel{(d)}{\le} &
\sum_{j',k'} \frac{e^{2 n^{3/4}} }{\mathsf{M}_{B,n} -1}
 (n+1)^{\frac{d_T(d_Y+2)(d_Y-1)}{2}}
 (C_{n}^{(2)})^{-1}\nonumber \\
= & 
\frac{e^{2 n^{3/4}} \mathsf{M}_{B,n} }{\mathsf{M}_{B,n} -1}
(n+1)^{\frac{d_T(d_Y+2)(d_Y-1)}{2}}
e^{-nr_A},
\Label{2nd-UK}
\end{align}
where each step can be shown as follows.
Step $(a)$ follows from the invariance of $\overline{\Pi}^{(2)}(j')$ by any action 
in the group $ S_{(\psi_{A,n}(j'),\psi_{B,n}(k'))}$.
Step $(b)$ follows from \eqref{8YW} of Lemma \ref{YP7W}.
Step $(c)$ follows from \eqref{LN9}.
Step $(d)$ follows from the condition in the projection $\overline{\Pi}^{(2)}(j')$.

Hence, since we can choose $s$ freely in \eqref{KJ2C} and \eqref{2nd-E2}, 
from the combination of \eqref{ER3}, \eqref{KJ2C}, \eqref{2nd-E2}, \eqref{3rd-E2}, \eqref{4th-E2}, and
\eqref{2nd-UK}, 
we obtain the following lower bond of the exponent of the decoding error probability of Receiver $Y$;
\begin{align}
&
\min \Big(
\max_s s(I_{1-s}(A;Y)-R_A-r_A),
\max_s s(I_{1-s}(B;Y|A)-R_B-r_B),r_A,r_A+r_B,r_B\Big) \nonumber \\
= &
\min \Big(
\max_s s(I_{1-s}(A;Y)-R_A-r_A),
\max_s s(I_{1-s}(B;Y|A)-R_B-r_B),r_A,r_B\Big) \nonumber \\
=&
\min \Big(
\min (\max_s s(I_{1-s}(A;Y)-R_A-r_A),r_A),
\min (\max_s s(I_{1-s}(B;Y|A)-R_B-r_B),r_B)
\Big) .
\end{align}
Thus, we obtain Eq. \eqref{Ex3}.

\section{Universal classical-quantum MAC coding with separate decoding}\Label{S7V}
In this section, 
we construct universal classical-quantum MAC code with separate decoding that achieves the general points.
\subsection{Code construction}
The universal encoder is the same as Subsection \ref{S7-1}.
Our decoder with separate decoding is composed of two POMs by using the same notations given in Section \ref{S7}.
Using two positive numbers $r_A$ and $r_B$, we 
define the projections $\overline{\Pi}_{j,k}^{(4)}, 
\overline{\Pi}^B_{j,k}, \overline{\Pi}^B_{j},
\overline{\Pi}^A_{j,k}$, and $ \overline{\Pi}^A_{k}$
in addition to $\overline{\Pi}_{j,k}^{(1)}$ and $\overline{\Pi}_{j,k}^{(3)}$; 
\begin{align}
\overline{\Pi}_{j,k}^{(4)}:=& \Big\{
\Big(\bigotimes_{t \in {\cal T}}
\rho_{\psi_{A,n,t}(j),\psi_{B,n,t}(k)} \Big)\ge C_{n}^{(2)}
\Big(\bigotimes_{t \in {\cal T}}
\rho_{\psi_{B,n,t}(j)}\Big) \Big\}\\
\overline{\Pi}^B_{j,k}:=&\overline{\Pi}_{j,k}^{(1)}\overline{\Pi}_{j}^{(3)}, \quad 
\overline{\Pi}^B_{j}:= \sum_{j}\overline{\Pi}^B_{j,k}  \\
\overline{\Pi}^A_{j,k}:=&\overline{\Pi}_{j,k}^{(4)}\overline{\Pi}_{j}^{(3)},\quad
\overline{\Pi}^A_{j}:=\sum_k \overline{\Pi}^A_{j,k}.
\end{align}
The projection $\overline{\Pi}_{j}^{(3)}$
is commutative with $\overline{\Pi}_{j,k}^{(1)}$ and $\overline{\Pi}_{j,k}^{(4)}$.
Then, the decoders with separate decoding are given as
\begin{align}
\overline{D}^B(k):=&
\Big(\sum_{k'}\overline{\Pi}^B(k')\Big)^{-1/2}
\overline{\Pi}^B(k) \Big(\sum_{k'}\overline{\Pi}^B(k')\Big)^{-1/2} \\
\overline{D}^A(j):=&
\Big(\sum_{j'}\overline{\Pi}^A(j')\Big)^{-1/2}
\overline{\Pi}^A(j) \Big(\sum_{j'}\overline{\Pi}^A(j')\Big)^{-1/2}.
\end{align}
Since $\overline{\Pi}_{j,k}^{(1)}$ 
is not commutative with $\overline{\Pi}_{j,k}^{(4)}$ in general,
we cannot construct a decoder with joint decoding in this way
by using the projections $\overline{\Pi}_{j}^{(3)}$, $\overline{\Pi}_{j,k}^{(1)}$ and $\overline{\Pi}_{j,k}^{(4)}$.


\subsection{Error evaluation}
The decoding error probability of our code is evaluated in a quite similar way to Section \ref{S7}.
Since $\epsilon_A(\Psi_{S,n};W^{(n)})$ can be evaluated in the same way, 
we evaluate only $\epsilon_B(\Psi_{S,n};W^{(n)}))$.
The decoding error probability for message from $B$ is decomposed as 
\begin{align}
&\Tr W^{(n)}_{\psi_{A,n}(j),\psi_{B,n}(k)} (I-\overline{D}^B(k)) \nonumber \nonumber \\
\stackrel{(a)}{\le}  & 2 \Tr W^{(n)}_{\psi_{A,n}(j),\psi_{B,n}(k)} (I-\overline{\Pi}^B(k))
+ 4 \Tr W^{(n)}_{\psi_{A,n}(j),\psi_{B,n}(k)}
\Big(\sum_{k'\neq k}\overline{\Pi}^B(k')\Big) \nonumber \nonumber \\
\le & 2 \Tr W^{(n)}_{\psi_{A,n}(j),\psi_{B,n}(k)} (I-\overline{\Pi}_{j,k}^{(1)})
+ 2 \Tr W^{(n)}_{\psi_{A,n}(j),\psi_{B,n}(k)} (I-\overline{\Pi}_{j}^{(3)})\nonumber  \nonumber \\
&+ 4 \Tr W^{(n)}_{\psi_{A,n}(j),\psi_{B,n}(k)}
\Big(\sum_{j' (\neq j),k'(\neq k)}\overline{\Pi}^B(j',k')\Big) 
+ 4 \Tr W^{(n)}_{\psi_{A,n}(j),\psi_{B,n}(k)}
\Big(\sum_{k' \neq k}\overline{\Pi}^B(j,k')\Big)\nonumber  \nonumber \\
\le& 2 \Tr W^{(n)}_{\psi_{A,n}(j),\psi_{B,n}(k)} (I-\overline{\Pi}_{j,k}^{(1)})
+ 2 \Tr W^{(n)}_{\psi_{A,n}(j),\psi_{B,n}(k)} (I-\overline{\Pi}_{j}^{(3)}) \nonumber \nonumber \\
&+ 4 \Tr W^{(n)}_{\psi_{A,n}(j),\psi_{B,n}(k)}
\Big(\sum_{j' (\neq j),k'(\neq k)}\overline{\Pi}^{(3)}(j',k')\Big) 
+ 4 \Tr W^{(n)}_{\psi_{A,n}(j),\psi_{B,n}(k)}
\Big(\sum_{k' \neq k}\overline{\Pi}^{(1)}(j,k')\Big) 
\Label{ER3V}
\end{align}
where Step $(a)$ follows from \cite[Lemma 2]{HN}.

All the terms in \eqref{ER3} except for 
the second term $\Tr W^{(n)}_{\psi_{A,n}(j),\psi_{B,n}(k)} (I-\overline{\Pi}_{j,k}^{(3)})$
has been evaluated in Section \ref{S7}.
The second term $\Tr W^{(n)}_{\psi_{A,n}(j),\psi_{B,n}(k)} (I-\overline{\Pi}_{j,k}^{(3)})$
is evaluated by using \eqref{KJ1}.
That is, applying the same type of modification as \eqref{KJ2C} to \eqref{KJ1},
we have
\begin{align}
&  \Tr W^{(n)}_{\psi_{A,n}(j),\psi_{B,n}(k)} (I-\overline{\Pi}_{j,k}^{(3)})\nonumber  \nonumber \\
{\le} & 
(n+1)^{d_T(\frac{s d_X(d_Y+2)(d_Y-1)}{2}+d_X)} 
e^{n s(R_A+r_A+R_B+r_B- I_{1-s}(AB;Y|T))}. \Label{MLP}
\end{align}
Hence, since we can choose $s$ freely in \eqref{KJ2C} and \eqref{MLP}, 
from the combination of \eqref{KJ2C}, \eqref{MLP}, \eqref{MLP}, \eqref{3rd-E2}, and \eqref{4th-E2}, 
we obtain the following lower bond of the exponent of the decoding error probability 
$\epsilon_B(\Psi_{S,n};W^{(n)}))$.
\begin{align}
&
\min \Big(
\max_s s(I_{1-s}(B;Y|A)-R_B-r_B),
\max_s s(I_{1-s}(AB;Y)-R_A-R_B-r_A -r_B),r_A+r_B,r_B\Big) .
\end{align}
Thus, we obtain Eq. \eqref{Ex3B}.

In the same way, 
we obtain Eq. \eqref{Ex3A} by replacing the role of Eq. \eqref{8LW} by Eq. \eqref{8LWH}.

\section{Conclusions}\Label{S8}
As the first main result, we have given a c-q universal superposition code
by combining the generalized packing lemma by \cite{Korner-Sgarro} and
the modification of universal decoder given from the Schur duality \cite{Ha1}. Applying this code, we have derived the capacity region of c-q compound BCD.
As the second main result, we have shown c-q universal MAC code with joint decoding by modifying the above universal code with use of another generalized packing lemma by \cite{Liu-Hughes}.
This code works well for corner points.
As the third main result, we have shown a c-q universal MAC code with separate decoding by constructing another universal decoder in the above universal code.
Combing the universal code with separate decoding with Eq. \eqref{MGY}, we have shown a single-letterized formula for the capacity region of a c-q compound MAC.

The key point of our method is the combination of 
the construction of a code with separation decoding and the gentle operator lemma \cite{Win,O-G,Springer}.
We can expect application of this kind of combination 
to various topics of quantum information theory with the multiple user setting. 

Further, the encoder of our universal codes does not depend on the output dimension unlike the preceding studies \cite{Korner-Sgarro, PW, Liu-Hughes}.
Hence, similar to the paper \cite{Conti}, there is a possibility that our encoder can be used for universal codes for c-q BCD and c-q MAC even with infinite-dimensional output systems.
Such an extension is another future study.
In addition, The derivations in examples in Subsection \ref{S93A} and \ref{S91} employ several numerical calculations.
Hence, their analytical derivations are also future studies.

In addition, 
one might be interested in the problem
whether our method for universal coding can be applied to 
the case where shared randomness between the senders is allowed
in a cq-MAC model.
The classical version of the above model 
was started by Willems \cite{Willems},
where the senders of a (classical) MAC may communicate to build a common randomness or they may share a common randomness for help their cooperation to send messages. 
This model often is addressed as MAC with conferencing encoders \cite{BLW} and has been extended to quantum \cite{BN}. 
Since our model is different from the above model, 
it is another interesting future problem to extend our result to the above setting.

\section*{Acknowledgments}
MH would like to thank Prof. Tomohiko Uyematsu for informing the reference \cite{Liu-Hughes}, which takes a central role in our code construction.

\appendices

\section{Another universal decoder for superposition coding}\Label{A2}
\subsection{Decoder construction}
In this appendix, 
we give another universal decoder only for Reciever $Y$ 
with the same encoder as Subsection \ref{S6-1}
for superposition coding, which has a different exponent from the exponents of the decoder given in Subsubsection \ref{SS5-1}.
The decoder presented here is similar to that given in Section \ref{S7V}.
As explained later, this decoder has an exponent different from that given in Section \ref{S6}.

We choose $C_{n}^{(1)}:= e^{n (R_B+r_B)}$ and $C_{n}^{(2)}:=e^{n (R_A+r_A)}$.
We define the projection $\hat{\Pi}_{j,k}:=\Pi_{j,k}^{(1)}\Pi_{j,k}^{(3)}$.
The decoder is given as
\begin{align}
\hat{D}(j,k):=&
\Big(\sum_{j',k'}\hat{\Pi}(j',k')\Big)^{-1/2}
\hat{\Pi}(j,k) \Big(\sum_{j',k'}\hat{\Pi}(j',k')\Big)^{-1/2}.
\end{align}

\subsection{Error evaluation}
We evaluate the decoding error probability of Receiver $Y$ as
\begin{align}
&\Tr W^{(n)}_{\phi_{B,n}(j,k)} (I-\hat{D}(j,k))
\le 2 \Tr W^{(n)}_{\phi_{B,n}(j,k)} (I-\Pi(j,k))
+ 4 \Tr W^{(n)}_{\phi_{B,n}(j,k)}
\Big(\sum_{(j',k')\neq (j,k)}\Pi(j',k')\Big) \nonumber \\
=& 2 \Tr W^{(n)}_{\phi_{B,n}(j,k)} (I-\Pi_{j,k}^{(3)})
+ 2 \Tr W^{(n)}_{\phi_{B,n}(j,k)} (I-\Pi_{j,k}^{(1)}) \nonumber \\
&+ 4 \Tr W^{(n)}_{\phi_{B,n}(j,k)}
\Big(\sum_{j',k' \neq j}\Pi(j',k')\Big) 
+ 4 \Tr W^{(n)}_{\phi_{B,n}(j,k)}
\Big(\sum_{k' \neq k}\Pi(j,k')\Big) \nonumber \\
\le & 2 \Tr W^{(n)}_{\phi_{B,n}(j,k)} (I-\Pi_{j,k}^{(3)})
+ 2 \Tr W^{(n)}_{\phi_{B,n}(j,k)} (I-\Pi_{j,k}^{(1)}) \nonumber \\
&+ 4 \Tr W^{(n)}_{\phi_{B,n}(j,k)}
\Big(\sum_{j',k' \neq j}\Pi^{(3)}(j',k')\Big) 
+ 4 \Tr W^{(n)}_{\phi_{B,n}(j,k)}
\Big(\sum_{k' \neq k}\Pi^{(1)}(j,k')\Big) .
\Label{ER4}
\end{align}
These four terms of \eqref{ER4} are calculated in \eqref{KJ1}, \eqref{KJ2}, \eqref{3rd-E}, and \eqref{4th-E}.
Hence, since we can choose $t$ freely in \eqref{KJ2} and \eqref{2nd-E}, 
from the combination of \eqref{ER4}, \eqref{KJ1}, \eqref{KJ2}, \eqref{3rd-E}, and \eqref{4th-E},
we obtain the following lower bond of the exponent;
\begin{align}
&
\min \Big(
\max_s s(I_{1-s}(X;Y)-R_A-R_B-r_A-r_B),
\max_s s(I_{1-s}(X;Y|U)-R_B-r_B),r_A,r_A+r_B\Big) \nonumber \\
=&
\min \Big(
\min (\max_s s(I_{1-s}(X;Y)-R_A-R_B-r_A-r_B),r_A+r_B), \nonumber\\
&\quad \min (\max_s s(I_{1-s}(X;Y|U)-R_B-r_B),r_B)
\Big) .
\end{align}
We maximize it by choosing $r_A$ and $r_B$; 
\begin{align}
&
\max_{r_A,r_B}
\min \Big(
\min (\max_s s(I_{1-s}(X;Y)-R_A+R_B-r_A-r_B),r_A+r_B), \nonumber \\
&\quad \min (\max_s s(I_{1-s}(X;Y|U)-R_B-r_B),r_B)
\Big) \nonumber \\
= &
\min \Big(
\max_{0 \le s\le 1}\frac{s(I_{1-s}(X;Y) -R_A-R_B)}{1+s},
\max_{0 \le s\le 1}\frac{s(I_{1-s}(X;Y|U) -R_B)}{1+s}
\Big) ,
\end{align}
where the maximum is achieved when
\begin{align}
r_A+r_B=&\max_{0 \le s\le 1}\frac{s(I_{1-s}(X;Y) -R_A-R_B)}{1+s}\nonumber \\
r_B=&\max_{0 \le s\le 1}\frac{s(I_{1-s}(X;Y|U) -R_B)}{1+s} .
\end{align}

\end{document}